\documentclass[11pt]{article}

\usepackage[top=30truemm,bottom=30truemm,left=25truemm,right=25truemm]{geometry}

\usepackage[utf8]{inputenc} 
\usepackage[T1]{fontenc}    
\usepackage{textcomp}
\usepackage{bm}
\usepackage{bbm}  

\usepackage{amsmath}
\usepackage{amssymb}
\usepackage{amsthm}
\usepackage{amsfonts}
\usepackage{mathrsfs}
\usepackage{nicefrac}
\usepackage{mathtools}

\usepackage{booktabs}
\usepackage{multirow}

\usepackage[authoryear,round]{natbib}

\usepackage[svgnames]{xcolor}
\usepackage{colortbl}  
\usepackage[hypertexnames=false]{hyperref}       
\hypersetup{
  colorlinks=true,
  linkcolor=blue,
  citecolor=blue,
}

\usepackage{stmaryrd}
\usepackage{tgtermes}
\usepackage{comment}

\usepackage[ruled,vlined,linesnumbered]{algorithm2e}
\SetArgSty{textnormal}

\usepackage{microtype}
\usepackage[labelfont=bf]{caption}
\usepackage{wrapfig}
\usepackage{tikz}

\usepackage{url}
\usepackage{enumitem}
\usepackage{thm-restate}

\usepackage[normalem]{ulem}
\usepackage{tcolorbox}
\usepackage{mdframed} 

\usepackage{threeparttable}

\usepackage[nameinlink]{cleveref}
\crefname{equation}{}{}
\Crefname{algocf}{Algorithm}{Algorithms}

\crefname{footnote}{footnote}{footnotes}
\Crefname{footnote}{Footnote}{Footnotes}

\crefalias{AlgoLine}{line} 
\crefalias{algocfline}{line}

\theoremstyle{plain}
\newtheorem{theorem}{Theorem}
\newtheorem{lemma}{Lemma}
\newtheorem{proposition}{Proposition}

\theoremstyle{definition}
\newtheorem{definition}{Definition}

\newtheoremstyle{myremark}
  {3pt}{3pt}      
  {\normalfont}   
  {}              
  {\bfseries}     
  {.}             
  {0.5em}         
  {}              

\theoremstyle{myremark}
\newtheorem{remark}{Remark}
\newcommand{\rqed}{\unskip\nobreak\hfill$\square$} 

\renewcommand{\hat}{\widehat}
\renewcommand{\tilde}{\widetilde}
\renewcommand{\epsilon}{\varepsilon}

\def\E{\mathbb{E}}

\def\R{\mathbb{R}}

\def\N{\mathbb{N}}

\def\calA{\mathcal{A}}

\def\calI{\mathcal{I}}
\def\calJ{\mathcal{J}}
\def\calK{\mathcal{K}}

\def\calM{\mathcal{M}}

\def\calS{\mathcal{S}}

\DeclareMathOperator*{\argmax}{arg\,max}
\DeclareMathOperator*{\argmin}{arg\,min}

\DeclareMathOperator*{\diag}{diag}

\DeclareMathOperator{\interior}{int}

\DeclarePairedDelimiter{\abs}{\lvert}{\rvert}
\DeclarePairedDelimiter{\brk}{[}{]}

\DeclarePairedDelimiter{\set}{\{}{\}}
\DeclarePairedDelimiter{\prn}{(}{)}
\DeclarePairedDelimiter{\nrm}{\|}{\|}
\DeclarePairedDelimiter{\inpr}{\langle}{\rangle}

\DeclarePairedDelimiter{\floor}{\lfloor}{\rfloor}

\newcommand{\zeros}{\mathbf{0}}
\newcommand{\ones}{\mathbf{1}}

\newcommand{\ie}{\textit{i.e.,}}
\newcommand{\eg}{\textit{e.g.,}}

\newcommand{\Reg}{\mathsf{Reg}}

\newcommand{\SwapReg}{\mathsf{SwapReg}}

\renewcommand{\t}{{t}}
\newcommand{\T}{{T}}
\newcommand{\Tp}{{T+1}}
\newcommand{\tp}{{t+1}}
\newcommand{\tm}{{t-1}}
\newcommand{\s}{{s}}

\newcommand{\sm}{{s-1}}

\newcommand{\Amax}{A_{\max}}

\newcommand{\Adiff}{A_{\mathrm{diff}}}
\newcommand{\Abardiff}{\bar{A}_{\mathrm{diff}}}

\newcommand{\Umax}{U_{\max}} 

\newcommand{\util}{\tilde{u}}
\newcommand{\ybm}{\bm{y}}
\newcommand{\hbm}{\bm{h}}
\newcommand{\zbm}{\bm{z}}

\newcommand{\xrm}{x}
\newcommand{\yrm}{y}

\newcommand{\mx}{m_\xrm}
\newcommand{\my}{m_\yrm}

\newcommand{\Qmat}{Q}

\newcommand{\ltil}{\widetilde{\ell}}
\newcommand{\gtil}{\widetilde{g}}

\renewcommand{\hbar}{\bar{h}}
\newcommand{\ubar}{\bar{u}}

\newcommand{\xhat}{\widehat{x}}
\newcommand{\yhat}{\widehat{y}}

\newcommand{\Cx}{C_\xrm}
\newcommand{\Cy}{C_\yrm}
\newcommand{\Chatx}{\hat{C}_\xrm}
\newcommand{\Chaty}{\hat{C}_\yrm}

\newcommand{\Ctilx}{\tilde{C}_\xrm}
\newcommand{\Ctily}{\tilde{C}_\yrm}

\newcommand{\mmax}{m}

\newcommand{\nsqrt}[1]{\sqrt{\smash[b]{#1}}\vphantom{#1}} 
\usepackage{xparse}
\NewDocumentCommand{\tsqrt}{O{1.2ex} m}{%
  \sqrt{\smash[b]{#2}\rule{0pt}{#1}}%
}

\newcommand{\nn}{\nonumber\\}
\newcommand{\n}{\nonumber}
\newcommand{\per}{\,.}
\newcommand{\com}{\,,}

\newcommand{\sumT}{\sum_{t=1}^T}

\title{Scale-Invariant Fast Convergence in Games}

\author{
  Taira Tsuchiya\footnote{
    The University of Tokyo and RIKEN; 
    \texttt{tsuchiya@mist.i.u-tokyo.ac.jp}.
  }
  \and
  Haipeng Luo\footnote{
    University of Southern California; \texttt{haipengl@usc.edu}.
  }
  \and
  Shinji Ito\footnote{
    The University of Tokyo and RIKEN; \texttt{shinji@mist.i.u-tokyo.ac.jp}.
  }
}

\begin{document}
\maketitle

\begin{abstract}
Scale-invariance in games has recently emerged as a widely valued desirable property. Yet, almost all fast convergence guarantees in learning in games require prior knowledge of the utility scale. To address this, we develop learning dynamics that achieve fast convergence while being both \emph{scale-free}, requiring no prior information about utilities, and \emph{scale-invariant}, remaining unchanged under positive rescaling of utilities. For two-player zero-sum games, we obtain scale-free and scale-invariant dynamics with external regret bounded by $\tilde{O}(A_{\mathrm{diff}})$, where $A_{\mathrm{diff}}$ is the payoff range, which implies an $\tilde{O}(A_{\mathrm{diff}} / T)$ convergence rate to Nash equilibrium after $T$ rounds. For multiplayer general-sum games with $n$ players and $m$ actions, we obtain scale-free and scale-invariant dynamics with swap regret bounded by $O(U_{\mathrm{max}} \log T)$, where $U_{\mathrm{max}}$ is the range of the utilities, ignoring the dependence on the number of players and actions. This yields an $O(U_{\mathrm{max}} \log T / T)$ convergence rate to correlated equilibrium. Our learning dynamics are based on optimistic follow-the-regularized-leader with an adaptive learning rate that incorporates the squared path length of the opponents' gradient vectors, together with a new stopping-time analysis that exploits negative terms in regret bounds without scale-dependent tuning. For general-sum games, scale-free learning is enabled also by a technique called doubling clipping, which clips observed gradients based on past observations.
\end{abstract}

\newpage

\section{Introduction}\label{sec:introduction}
Learning in games studies repeated strategic interactions in which each player adapts their strategy online to minimize regret~\citep{freund99adaptive,hart00simple,cesabianchi06prediction}.
A key fact is that no-regret learning enables equilibrium computation: for example in two-player zero-sum games, if each player achieves external regret $\Reg^T$, then the average play after $T$ rounds is an $O(\Reg^T/T)$-approximate Nash equilibrium~\citep{freund99adaptive}.
This paradigm underlies practical successes ranging from superhuman game AI~\citep{bowling15heads,moravvcik17deepstack,perolat22mastering,meta22human} to recent methods for aligning LLMs~\citep{munos24nash,swamy24minimaxlist}.

A standard assumption in learning in games is that the underlying utility functions are bounded and that the exact bound is known to the players.
In two-player zero-sum games, this corresponds to knowing a scale parameter such as maximum absolute magnitude $\Amax = \max_{i,j} \abs{A_{i,j}}$ for a payoff matrix $A$.
In multiplayer general-sum games, this corresponds to knowing $\Umax$ such that all utility values lie in $[-\Umax,\Umax]$.
However, in many applications, it is difficult to know the scale of utilities in advance, making it desirable to design learning dynamics that do not require such instance-dependent tuning.

Motivated by the connection between no-regret online learning and equilibrium computation, a natural way to remove the assumption that the scale is known in advance is to adopt \emph{scale-free} online learning~\citep{cesabianchi07improved}.
For example, in two-player zero-sum games, using scale-free online learning, we can obtain an external regret bound of $O(\Amax \sqrt{T})$ without knowing $\Amax$, which implies an $O(\Amax/\sqrt{T})$ convergence rate to Nash equilibrium.

However, this result is not satisfactory: when the scale is known, one can obtain much faster convergence rates~\citep{daskalakis11near}.
One representative approach to obtain such fast rates is to use an optimistic online learning algorithm, such as optimistic follow-the-regularized-leader (OFTRL) or optimistic online mirror descent~\citep{chiang13beating,rakhlin13optimization}.
In two-player zero-sum games, if $\Amax$ is known, the external regret can be bounded by $\tilde{O}(\Amax)$~\citep{rakhlin13optimization,syrgkanis15fast}, which yields a $\tilde{O}(\Amax / T)$ convergence rate to Nash equilibrium.
In multiplayer general-sum games, if $\Umax$ is known, the swap regret is bounded by $O(\Umax \log T)$~\citep{anagnostides22uncoupled}, which yields an $O(\Umax \log T / T)$ convergence rate to correlated equilibrium, ignoring the dependence on the number of players and actions.

These observations lead to a natural question: \emph{can we construct learning dynamics that are agnostic to the scale of utilities while still achieving fast convergence?}
In particular, is it possible to construct scale-free learning dynamics whose convergence rates match, up to multiplicative factors, the best-known fast rates achieved when the scale is known in advance?

Scale-freeness and \emph{scale-invariance} are closely related.
Informally, learning dynamics are scale-invariant if rescaling the utilities by any positive constant does not change the sequence of strategies (see \Cref{subsec:scale} for formal definitions).
Scale-invariance is widely regarded as a practically useful property in game learning; for instance, regret matching, one of the most powerful methods for solving games, is scale-invariant~\citep{hart00simple}, and recent work argues that such invariance properties may be crucial for strong empirical performance when solving zero-sum games~\citep{chakrabarti24extensive,zhang25scale}.
Accordingly, our goal is to construct learning dynamics that achieve fast convergence while being both scale-free and scale-invariant.

\begin{table*}[t]
      \caption{Comparison of individual (external) regret upper bounds of the $x$-player in two-player zero-sum games with a payoff matrix $A \in [-\Amax, \Amax]^{\mx\times \my}$ after $T$ rounds, where $m = \max\set{\mx,\my}$.
      The parameter $\delta > 0$ is an instance-dependent constant, which can be arbitrarily small.
      The value $\Adiff = \max_{i,j} A_{i,j} - \min_{i,j} A_{i,j} \leq 2 \Amax$ is the range of the entries of $A$, which can be much smaller than $\Amax$.
      The bound of ``Scale-free OLO'' is achieved by simply using scale-free online linear optimization (OLO) algorithms.
      ``$\sqrt{\text{Corrupt}}$?'' means whether the effect of opponent deviations and corrupted observed utilities can be kept to a square-root dependence (see \Cref{subsec:corrupted_game_two_player} for details).
    }
    \label{table:regret}
    \centering
    \small
    \begin{tabular}{llll}
        \toprule
      \textbf{Reference} & \textbf{Regret bound} & \textbf{$\sqrt{\text{Corrupt}}$?} & \textbf{Scale-free \& Scale-invariant?}
      \\
      \midrule
      \citet{syrgkanis15fast}  & $\Amax \log m$ & No & No
      \\  
      \citet{tsuchiya25corrupted} & $\Amax \log m$ & Yes & No
      \\  
      \midrule
      Scale-free OLO & $\Amax \sqrt{T \log m} $ & Yes & \textbf{Yes} 
      \\
      \citet{tsuchiya25corrupted} & ${\Amax \log m}/{\delta}$ & Yes & \textbf{Yes} \hfill (see also \Cref{remark:zhang25scale})
      \\  
      \citet{zhang25scale} & ${\Amax m^{3/2}}/{\delta} $ & Yes & \textbf{Yes} \hfill (see also \Cref{remark:zhang25scale})
      \\  
      \textbf{This work (\Cref{thm:main})} & $\Adiff \log m$ & Yes & \textbf{Yes}
      \\  
      \bottomrule
    \end{tabular}
  \end{table*}

\subsection{Contributions of this paper}
We affirmatively answer the above question by constructing learning dynamics with scale-invariant and scale-free fast convergence in two-player zero-sum games and multiplayer general-sum games.

\paragraph{Two-player zero-sum games}
We start with two-player zero-sum games and provide learning dynamics with the following regret guarantees:
\begin{theorem}[Informal version of \Cref{thm:main}]\label{thm:indiv_reg_twoplayer_informal}
  In two-player zero-sum games with a payoff matrix $A$,
  there exists scale-free and scale-invariant learning dynamics such that
  the external regrets of the $x$- and $y$-players are bounded by
  $
  \Adiff \log m
  $,
  where $m$ is the maximum number of actions among the players and $\Adiff = \max_{i,j} A_{i,j} - \min_{i,j} A_{i,j} \leq 2 \Amax$ is the range of the entries of $A$.
  Consequently, the average play after $T$ rounds under this dynamic is an $O(\Adiff \log m / T)$-approximate Nash equilibrium.
  This learning dynamic is robust against the opponent deviations and adversarial corruption of observed utilities (see \Cref{subsec:corrupted_game_two_player} for details).
\end{theorem}
\Cref{thm:indiv_reg_twoplayer_informal} positively resolves the question raised in~\citet{zhang25scale} of whether one can retain the scale-free fast convergence while removing the dependence on the $1/\delta$ factor in their regret bound, where $\delta>0$ can be arbitrarily small (see the discussion after \Cref{thm:main} for detailed discussion).
Moreover, the above regret bounds and convergence rates depend on the payoff range $\Adiff$, which can be much smaller than the maximum absolute magnitude $\Amax$.
A comparison with existing bounds is provided in \Cref{table:regret}.

The learning dynamic that achieves the above convergence rate is based on optimistic Hedge with an adaptive learning rate.
Many existing works on fast rates in learning in games use a constant learning rate that does not depend on the time horizon or past observations (\eg~\citealt{syrgkanis15fast,anagnostides22uncoupled}).
However, to operate adaptively under an unknown scale, an adaptive learning rate that depends on past observations is essential.

In the context of learning in games, adaptive learning rates were introduced by \citet{rakhlin13optimization}, and an $O(1/T)$ convergence rate that is completely independent of $T$ was obtained by \citet{tsuchiya25corrupted} using OFTRL.
Nevertheless, their learning-rate choice cannot be directly used for scale-free learning:
to obtain fast convergence by exploiting the negative term in the regret upper bound of OFTRL, they tune the learning rate so that it is upper bounded by a constant, which requires prior knowledge of the scale $\Amax$.

To overcome this issue, we introduce two refinements in the design and analysis of optimistic Hedge.
The first is to incorporate the squared path length of the gradients observed by the opponent into the learning rate.
We then develop a new analysis that leverages the negative term: by defining an appropriate stopping time that tracks the growth of the squared path length of gradients, we can exploit the negative term without imposing a constant upper bound on the learning rate, allowing us to obtain a fast convergence rate without knowing the scale.
Note that, to obtain the $\Adiff$-dependent bound, it is important to use an AdaHedge-type analysis~\citep{orabona18scale}, rather than other scale-free online learning approaches such as \citet{cutkosky19artificial}.
It is also worth noting that we provide a learning dynamic with fast convergence with no communication of gradients in \Cref{app:nocom}.

\begin{table*}[t]
    \caption{Comparison of individual swap regret upper bounds of player $i$ in multiplayer general-sum games with $n$ players and $m$ actions after $T$ rounds.
      The value $\Umax$ is the scale of the underlying utility function.
      The upper bound of ``Scale-free OLO'' is achieved by using an optimistic AdaHedge-type scale-free online linear optimization (OLO) algorithm (see \Cref{prop:OptHedge_swap} in \Cref{subsec:OptHedge_swap} for details).
  }
  \label{table:regret_generalsum}
  \centering
  \small
  \begin{tabular}{llll}
    \toprule
    \textbf{Reference} & \textbf{Regret bound} & \textbf{$\sqrt{\text{Corrupt}}$?} & \textbf{Scale-free \& Scale-invariant?}
    \\
    \midrule
    \citet{anagnostides22uncoupled} &  $\Umax n m^{5/2} \log T$ & No & No
    \\  
    \citet{tsuchiya25corrupted} & $\Umax n m^{5/2} \log T$ & Yes & No
    \\
    \midrule
    Scale-free OLO (\Cref{prop:OptHedge_swap}) & $\Umax \sqrt{m T \log m} $ & Yes & \textbf{Yes} 
    \\
    \textbf{This work (\Cref{thm:indiv_swapreg})} & $\Umax n^{3/2} m^{5/2} \log T$ & No & \textbf{Yes}
    \\
    \bottomrule
  \end{tabular}
\end{table*}

\paragraph{Multiplayer general-sum games}
Building on the algorithm and analysis for two-player zero-sum games, we establish learning dynamics for multiplayer general-sum games. 
In multiplayer general-sum games, to obtain a \emph{correlated equilibrium}~\citep{aumann74subjectivity,foster97calibrated,hart00simple}, we focus on minimizing the \emph{swap regret} to get the following guarantees:
\begin{theorem}[Informal version of \Cref{thm:indiv_swapreg}]\label{thm:indiv_swapreg_informal}
  In multiplayer general-sum games with $n$-players and $m$-actions,
  there exists scale-free ($\Umax$-agnostic) and scale-invariant learning dynamics such that the swap regret of each player $i$ is bounded by
  $
    \Umax n^{3/2} m^{5/2} \log T ,
  $
  Consequently,
  the time-averaged history of joint play of this dynamic after $T$ rounds is an $O(\Umax n^{3/2} m^{5/2} \log T / T)$-approximate correlated equilibrium.
\end{theorem}
This is the first scale-free and scale-invariant result with fast swap-regret bounds and fast convergence in multiplayer general-sum games.
For a detailed comparison between our upper bound and existing bounds, see the discussion immediately following \Cref{thm:indiv_swapreg}.
A comparison with existing swap regret bounds is summarized in \Cref{table:regret_generalsum}.

We employ the well-known reduction by \citet{blum07external}, which converts swap regret minimization into several instances of external regret minimization, and we instantiate each external regret minimizer with OFTRL using the log-barrier regularizer and an adaptive learning rate.
As in the two-player zero-sum case, the best-known swap regret upper bound by~\citep{anagnostides22uncoupled} can be obtained by an adaptive learning rate \citep{tsuchiya25corrupted}, but their learning-rate tuning cannot be used directly here due to the same reason as in the case of two-player zero-sum games.
Thus, similarly to two-player zero-sums, one may attempt to incorporate the squared path length of the gradients observed by the opponent into the learning rate and to use a stopping-time argument, combined with a known scale-free online learning technique.

Unfortunately, this approach does not work.
In the reduction of \citet{blum07external}, the outputs of the several external regret minimizers must be aggregated, but their associated utility vectors can have different scales across experts and across rounds.
Because of this scale mismatch, it is difficult to smoothly relate the squared path length of the gradients to the square path length of strategies, and consequently we cannot directly apply standard scale-free online learning techniques such as \citet{orabona18scale,cutkosky19artificial}.
To address this issue, we introduce a new clipping procedure inspired by \citet{cutkosky19artificial}, which we refer to as a \emph{doubling clipping}, that clips the observed utilities based on past observations.
We show that this mechanism preserves the stability of the learning dynamics in an unknown-scale setting while ensuring that the resulting deterioration of swap regret is roughly bounded by $\Umax \log T$.
See \Cref{sec:multiplayer} for details.

\subsection{Additional related work}
Optimistic learning has been a central tool for achieving fast convergence to equilibrium.
Starting from~\citet{rakhlin13optimization,syrgkanis15fast}, there has been many subsequent developments, \eg~\citet{foster16learning,chen20hedging,anagnostides22near,anagnostides22uncoupled}.

Scale-free learning was first studied for prediction with expert advice~\citep{cesabianchi07improved}, and has since seen many refinements such as~\citet{luo2015achieving,koolen15second,mhammedi19lipschitz}, to name a few.
In online convex optimization, which includes online linear optimization as a special case, scale-free and parameter-free methods have been developed in a series of works such as~\citet{orabona18scale,cutkosky19artificial,chen21impossible}.

\section{Preliminaries}\label{sec:preliminaries}

\paragraph{Notation and conventions}
For any natural number $n \in \N$, we write $[n] = \set{1, \dots, n}$.
Let $\zeros$ and $\ones$ be the all-zero and all-one vectors, respectively.
For a vector $x$, we write $x(a)$ for its $a$-th entry and $\nrm{x}_p$ for its $\ell_p$-norm, where $p \in [1,\infty]$.
For a matrix $A$, we write $A(k,\cdot)$ for its $k$-th row and $A_{i,j}$ for its $(i,j)$ element.
We use $\Delta(\calK)$ for the set of probability distributions over $\calK$, and $\Delta_d = \set{x \in [0,1]^d : \nrm{x}_1 = 1}$ for the $(d-1)$-dimensional probability simplex.
For brevity, we may use $f \lesssim g$ to denote $f = O(g)$.
For a sequence $z = (z^{1}, \dots, z^{T})$ and $q \in [1, \infty]$, 
we define the (squared) path-length values in terms of the $\ell_q$-norm up to round $t-1$ and after round $t$ as
  $
  P_q^\t(z) = \sum_{s=1}^{t-1} \nrm{z^\s - z^\sm}_q^2
  $
  and
  $
  Q_q^\t(z) = \sum_{s=t}^{T} \nrm{z^\s - z^\sm}_q^2,
  $
respectively, where we omit the dependence on $T$ from $Q_q^\t$ and let $z^{0} = \zeros$ for simplicity.
Note that $\sum_{t=1}^{T} \nrm{z^\t - z^\tm}_q^2 = P_q^\t(z) + Q_q^\t(z)$.
Throughout this paper, we frequently use $i$ to index players, $a$ to index actions, and $d$ or $m$ to denote the dimension of a feasible set.

\subsection{Online linear optimization}\label{subsec:olo}

\paragraph{Setup}
Online linear optimization is a central problem in online learning and has been used to study learning dynamics in games.
In this setting, a player is given a convex set $\calK \subseteq \R^d$ before the interaction begins.
Then, at each round $t = 1, \dots, T$, the player chooses a point $w^\t \in \calK$ using the information observed so far, and then the environment selects a loss vector $h^\t \in \R^d$ without seeing~$w^\t$.
The player then incurs loss $\inpr{w^\t, h^\t}\in \R$ and observes $h^\t$.

\paragraph{External regret and swap regret}
The evaluation metrics relevant to this study are \emph{external regret} and \emph{swap regret}.
The external regret is defined as the difference between the cumulative loss actually incurred by the player and the cumulative loss of the best fixed point in hindsight; that is,
$
  \Reg^T
  =
  \max_{u \in \calK} 
  \Reg^T(u) 
$
for
$
  \Reg^T(w^*)
  =
  \sumT \inpr{w^\t - w^*, h^\t}
  .
$
The swap regret compares the learner's performance to any swapping rule. 
We consider only the case where the feasible set $\calK$ is the probability simplex, and 
let
$\calM_d
= 
\set{ M \in [0,1]^{d \times d} \colon M(k, \cdot) \in \Delta_d \mbox{ for } k \in [d]}
$ 
denote the set of $d \times d$ row-stochastic matrices, which we also call transition probability matrices.
Then, the swap regret is defined as
$
  \SwapReg^T
  =
  \max_{M \in \calM_{d}} 
  \SwapReg^T(M) 
$
for
$
  \SwapReg^T(M)
  =
  \sumT \inpr{w^\t, h^\t - M h^\t}
  .
$

\paragraph{Optimistic follow-the-regularized-leader}
A standard and widely used algorithmic framework for online linear optimization is follow-the-regularized-leader (FTRL).
Here we describe its generalization, optimistic FTRL (OFTRL)~\citep{chiang13beating,rakhlin13optimization}, which is known to be a particularly powerful framework for learning in games.
OFTRL selects a point $w^\t \in \calK$ at round $t \in [T]$ by
$
  w^\t 
  \in 
  \argmin_{w \in \calK} 
  \set{
    \inpr{w, m^\t + \sum_{s=1}^{t-1} h^\s} 
    +
    \psi^\t(w)
  },
$
where $\psi^\t$ is a convex regularizer over $\calK$ and $m^\t \in \R^d$ is an optimistic prediction of the true loss vector $h^\t$.
When the optimistic prediction $m^\t$ is the zero vector for all rounds, OFTRL corresponds to FTRL.

The prediction $m^\t$ needs to be computed using only information available up to round $t-1$, and the regret becomes smaller when $m^\t$ is closer to the true loss vector $h^\t$ (see, \eg~\Cref{lem:adahedge_opt_neg_mainbody}).
This is useful for obtaining fast convergence in learning in games, since the loss vectors are induced by the strategies of the other players, which are updated in a gradient-descent manner~\citep{rakhlin13optimization,syrgkanis15fast}.
For this reason, a common choice is $m^\t = h^\tm$.
OFTRL is known to achieve an RVU (Regret bounded by Variation in Utilities) regret bound that contains a negative term, which is useful for obtaining fast convergence in games.
For concrete statements, see \Cref{lem:adahedge_opt_neg_mainbody} and \Cref{lem:oftrl_logbarrier}, which provide RVU bounds for OFTRL with the negative Shannon entropy regularizer and with the log-barrier regularizer, respectively.

\subsection{Two-player zero-sum games}\label{subsec:two_player_preliminaries}
Here we describe the problem of learning in two-player zero-sum games.
It is characterized by an unknown payoff matrix $A \in [-\Amax, \Amax]^{\mx \times \my}$, where $\Amax \geq 0$ is a scale parameter and $\mx$ and $\my$ are the numbers of actions available to the $x$- and $y$-players, respectively.
Note that most existing formulations of learning in games assume that $\Amax = 1$ is known to each player.
The interaction proceeds as follows.
At each round $t = 1, \dots, T$, the $x$-player selects a \emph{mixed strategy} (or simply a strategy) $x^\t \in \Delta_{\mx}$, and the $y$-player simultaneously selects $y^\t \in \Delta_{\my}$.
The $x$-player then observes a gain vector $g^\t = A y^\t \in [-\Amax, \Amax]^{\mx}$ and receives reward $\inpr{x^\t, g^\t}$, while the $y$-player observes a loss vector $\ell^\t = A^\top x^\t \in [-\Amax, \Amax]^{\my}$ and incurs loss $\inpr{y^\t, \ell^\t}$.

For each player, this interaction can be viewed as online linear optimization: for the $x$-player, the loss vector is $(-g^\t)$ and the feasible set is $\calK = \Delta_{\mx}$, and for the $y$-player, the loss vector is $\ell^\t$ and $\calK = \Delta_{\my}$.
Hence, the external regret of the $x$-player and that of the $y$-player can naturally be defined as
$
  \Reg_{x}^T 
  =
  \max_{x^* \in \Delta_{\mx}} 
  {\Reg_{x}^T(x^*)}
$
for 
$
  \Reg_{x}^T(x^*)
  =
  \sumT \inpr{x^* - x^\t, g^\t}
$
and
$
  \Reg_{y}^T
  = 
  \max_{y^* \in \Delta_{\my}} 
  {\Reg_{y}^T(y^*)}
$
for 
$
  \Reg_{y}^T(y^*)
  =
  \sumT \inpr{y^\t - y^*, \ell^\t}
  ,
$
respectively.

We now define the notion of a Nash equilibrium.
In a two-player zero-sum game with payoff matrix $A$, a pair of probability distributions $\sigma = (x^*, y^*)$ over the action sets $[\mx]$ and $[\my]$ is called an \emph{$\epsilon$-approximate Nash equilibrium} for $\epsilon \geq 0$ if, for any $x \in \Delta_{\mx}$ and $y \in \Delta_{\my}$,
$
  x^\top A y^* - \epsilon
  \leq
  {x^*}^\top A y^*
  \leq 
  {x^*}^\top A y + \epsilon
  .
$
When $\epsilon = 0$, we call $\sigma$ a \emph{Nash equilibrium}.
The following theorem provides a connection between no-external-regret learning dynamics and Nash equilibrium computation.
\begin{theorem}[\citealt{freund99adaptive}]\label{thm:reg2nasheq}
  In two-player zero-sum games,
  the product distribution of the average play $\prn{\frac1T \sumT x^\t, \frac1T \! \sumT y^\t}$ is a $(\prn{\Reg_{x}^T + \Reg_{y}^T} / T)$-approximate Nash equilibrium.
\end{theorem}

\subsection{Multiplayer general-sum games}\label{subsec:multi_player_preliminaries}

We now turn to multiplayer general-sum games.
Let $n \geq 2$ denote the number of players, with player set $[n] = \set{1,\dots,n}$.
Each player $i \in [n]$ has an action set $\calA_i$ of size $\abs{\calA_i} = m_i$ and an unknown utility function
$u_i \colon \calA_1 \times \cdots \times \calA_n \to [-\Umax,\Umax]$,
where $\Umax \geq 0$ is an unknown scale parameter.
Note again that most existing formulations of learning in games assume that $\Umax = 1$ is known to each player.
The interaction proceeds as follows.
At each round $t = 1, \dots, T$, every player $i \in [n]$ selects a mixed strategy $x_i^\t \in \Delta_{m_i}$ and then observes an \emph{expected utility} vector $u_i^\t \in [-\Umax, \Umax]^{m_i}$.
Here, the $a_i$-th component of the expected utility vector $u_i^\t$ is given by
$
  u_i^\t(a_i) = \E_{a_{-i} \sim x_{-i}^\t} \brk*{u_i(a_i, a_{-i})}
  ,
$
which is the expected reward obtained by player $i$ when choosing action $a_i$ while the other players act according to $x_{-i}^\t = (x_1^\t, \dots, x_{i-1}^\t, x_{i+1}^\t, \dots, x_n^\t)$.
Note that two-player zero-sum games are a special case of multiplayer general-sum games.

Again, for each player, this interaction can be viewed as online linear optimization: for player~$i$, the loss vector is~$(- u_i^\t)$ and the feasible set is $\calK = \Delta_{m_i}$.
Thus, the swap regret of each player~$i$ can naturally be defined as
$
  \SwapReg_{x_i}^T
  =
  \max_{M \in \calM_{m_i}} 
  \SwapReg_{x_i}^T(M) 
$
for 
$
  \SwapReg_{x_i}^T(M)
  =
  \sumT \inpr{x_i^\t, M u_i^\t - u_i^\t}
  ,
$
where we recall that
$\calM_m
= 
\set{ M \in [0,1]^{m \times m} \colon M(k, \cdot) \in \Delta_m \mbox{ for } k \in [m]}
$ 
is the set of all $m \times m$ row stochastic matrices.

We now define the notion of a correlated equilibrium.
A probability distribution $\sigma$ over the joint action space $\times_{i=1}^n \calA_i$ is called an \emph{$\epsilon$-approximate correlated equilibrium} for $\epsilon \geq 0$ if, for every player $i \in [n]$ and every (swap) function $\phi_i \colon \calA_i \to \calA_i$ that replaces an action $a_i$ by $\phi_i(a_i)$, it holds that
$
  \E_{a \sim \sigma} \brk*{ u_i(a) }
  \geq 
  \E_{a \sim \sigma} \brk*{ u_i(\phi_i(a_i), a_{-i}) }
  -
  \epsilon
  .
$
When $\epsilon = 0$, we simply call $\sigma$ a \emph{correlated equilibrium}.

The following theorem connects no-swap-regret learning to correlated equilibrium computation.
\begin{theorem}[\citealt{foster97calibrated}]\label{thm:swapreg2coreq}
  In multiplayer general-sum games,
  let 
  $\sigma^\t = \otimes_{i \in [n]} x_i^\t \in \Delta(\times_{i=1}^n \calA_i)$ 
  denote the joint distribution at round $t$, defined by
  $\sigma^\t(a_1, \dots, a_n) = \prod_{i\in[n]} x_i^\t(a_i)$ for each $a_i \in \calA_i$.
  Then the time-averaged distribution $\sigma = \frac{1}{T} \sumT \sigma^\t$
  is a $(\max_{i \in [n]} \SwapReg_{x_i}^T / T)$-approximate correlated equilibrium.
\end{theorem}

\subsection{Scale-invariant and scale-free learning dynamics}\label{subsec:scale}
Here we give definitions of \emph{scale-invariance} and \emph{scale-freeness} for online linear optimization algorithms and for learning dynamics in games.

\begin{definition}[Scale-invariant / scale-free online linear optimization]
An online linear optimization algorithm is said to be \emph{scale-invariant} if, for any sequence of loss vectors $\prn{h^\t}_{t=1}^T$ and any constant $c > 0$, the algorithm outputs the same sequence of points $\prn{w^\t}_{t=1}^T$ when run on $\prn{h^\t}_{t=1}^T$ and when run on $\prn{c h^\t}_{t=1}^T$.
Moreover, an online linear optimization algorithm is said to be \emph{scale-free} if it does not rely on any prior information about the loss vectors $\prn{h^\t}_{t=1}^T$ except for their dimension.
\end{definition}

\begin{definition}[Scale-invariant / scale-free learning dynamics in games]\label[definition]{def:scale_games}
A learning dynamic in games is said to be \emph{scale-invariant} if, for any functions $u_1, \dots, u_n$ and any constant $c > 0$, the dynamic outputs the same sequence of points when run on utility functions $u_1, \dots, u_n$ and when run on $c u_1, \dots, c u_n$.
A learning dynamic in games is said to be \emph{scale-free} if the dynamic does not rely on any prior information abount the underlying utilities except for their dimension.
\end{definition}

\section{Scale-Invariant Learning Dynamics for Two-Player Zero-Sum Games}\label{sec:two_player}
This section investigates scale-invariant and scale-free learning dynamics for two-player zero-sum games.
We will upper bound the external regrets of the $x$- and $y$-players.

\subsection{Proposed learning dynamic}\label{subsec:dynamics_two_player}
Here we present scale-invariant and scale-free learning dynamics with a fast convergence rate.
Let $M_x = \max\set{4, \log m_x / 2^{3/2}}$ and $M_y = \max\set{4, \log m_y / 2^{3/2}}$. Define $M = \max\set{M_x, M_y}$.
We consider the optimistic Hedge algorithm, which determines the strategies of the $x$- and $y$-players via
\begin{equation}
  \begin{split}    
  x^\t(a) 
  &
  \propto 
  \exp\prn[\bigg]{ \eta_x^\t \prn[\bigg]{ \sum_{s=1}^{t-1} g^\s(a) + g^\tm(a)} }
  \com
  \quad
  \eta_x^\t
  =
  \sqrt{\frac{M_x}{P_\infty^\t(g) + P_\infty^\t(\ell)}}
  \com
  \\
  y^\t(a) 
  &
  \propto 
  \exp\prn[\bigg]{ - \eta_y^\t \prn[\bigg]{ \sum_{s=1}^{t-1} \ell^\s(a) + \ell^\tm(a)} }
  \com\quad
  \eta_y^\t
  =
  \sqrt{\frac{M_y}{P_\infty^\t(g) + P_\infty^\t(\ell)}}
  \com
  \end{split}
  \label{eq:OptHedge}
\end{equation}
where 
we let $g^0 = \ell^0 = \zeros$, $x^1 = \frac{1}{\mx} \ones$, and $y^1 = \frac{1}{\my} \ones$ for simplicity,
and 
recall that
$
  g^\t = A y^\t,
$
$
\ell^\t = A^\top x^\t,
$ and $P_q^t(z) = \sum_{s=1}^{t-1} \nrm{z^\s - z^\sm}_q^2$.
If the denominators of the learning rates $\eta_x^\t$ and $\eta_y^\t$ are zero, we set $\eta_x^\t = \eta_y^\t = \infty$.
The learning rates $\eta_x^\t$ and $\eta_y^\t$ can be computed using only the observations up to time $t-1$.
Note that optimistic Hedge corresponds to OFTRL with the negative Shannon entropy regularizer.

\subsection{External regret upper bounds}
Let $m = \max\set{m_x, m_y}$. Then the above learning dynamic guarantees the following bounds.
\begin{theorem}[Scale-invariant fast convergence to Nash equilibrium]\label{thm:main}
Suppose that the $x$- and $y$-players use the algorithms in \cref{eq:OptHedge}. 
Then, it holds that
$
  \Reg_x^T
  \leq
  8 \Adiff
  \sqrt{
    5 \prn{M + 1} M_x
  }
  =
  O(\Adiff \log m)
$
and
$
  \Reg_y^T
  \leq
  8 \Adiff
  \sqrt{
    5 \prn{M + 1} M_y
  }
  =
  O(\Adiff \log m)
  .
$
Consequently (by \Cref{thm:reg2nasheq}), the average play is an $O( \Amax \log m / T)$-approximate Nash equilibrium.
\end{theorem}
A comparison with existing external regret upper bounds is provided in \Cref{table:regret}.
Unlike existing scale-free learning results, our analysis shows that the convergence rate can be controlled not by the maximum entry of the payoff matrix $\Amax$, but rather by the payoff range $\Adiff$, the difference between the maximum and the minimum entries in the payoff matrix.
Note that our $O(\Adiff \log m)$ regret bound does not contradict the lower bound for optimistic Hedge~\citep{tsuchiya25tight}, since their lower bound applies to the optimistic Hedge algorithms with a constant learning rate. 
It is also worth noting that our analysis remains robust even in the corrupted regime~\citep{tsuchiya25corrupted}, where the opponent may deviate from the prescribed algorithm (here, optimistic Hedge in~\cref{eq:OptHedge}) or the observed gradients may be adversarially corrupted; we defer the details to \Cref{app:proof_two_player}.

\begin{remark}
\label[remark]{remark:zhang25scale}
We provide several comparison remarks related to \citet{zhang25scale}.
First, to obtain a fast convergence rate of $O(1/T)$, their regret analysis assumes $\delta = \min\set{\nrm{g^1}, \nrm{\ell^1}} > 0$.
However, this $\delta$ cannot be known to each player in advance and can be arbitrarily small.
\Cref{thm:main} positively resolves the question raised in the conclusion of \citet{zhang25scale} of whether one can retain the scale-free fast convergence while removing the dependence on the $1/\delta$ factor in their regret bound.
As we will see in the following proof, we address this issue by introducing an appropriate stopping time and developing a new analysis that effectively leverages the negative term in the RVU bound.

Note that, if one can assume $\delta > 0$, then simply replacing constants $\log_+(\mx)$ and $\log_+(\my)$ in the denominator of learning rates in \citet{tsuchiya25corrupted} by $4 \delta^2$ yields a scale-invariant and scale-free learnig dynamic with a regret bound of $O(\frac{\Adiff \log m}{\delta})$ without prior knowledge of $\delta$, whose rate is better than that of \citet{zhang25scale}.
As we will also see in the swap-regret analysis and in \Cref{app:nocom}, the main difficulty of scale-free learning arises when the scale parameters are small, and thus assuming $\delta > 0$ allows us to achieve scale-invariant and scale-free fast convergence relateively more easily.
That said, \citet{zhang25scale} primarily aim to explain the empirical success of regret matching in zero-sum games from a scale-invariance perspective and to propose a new regret-matching variant, making their contribution complementary to ours.

Moreover, their learning dynamic uses optimistic online gradient descent, rather than optimistic Hedge (\ie~optimistic FTRL with the Shannon entropy regularizer).
As a result, in two-player zero-sum games their dynamic yields a polynomial dependence on the number of actions $m$.
One can replace it with optimistic OMD with Shannon entropy to resolve this, but an additional $\log T$ factor seems inevitable even with a refined analysis; thus, it is preferable to use optimistic FTRL.
\rqed
\end{remark}

\subsection{Regret analysis}\label{subsec:two_player_analysis}
Here we provide the proof of \Cref{thm:main}.
As mentioned soon after \Cref{thm:main}, our algorithm is robust to deviations by the opponent or adversarial corruption of utilities.
Still, focusing on the honest regime, we can obtain a regret upper bound with favorable leading constants.
The proof for the more general corrupted regime is provided in \Cref{app:proof_two_player}.

\paragraph{Preliminary lemmas}

We begin with the following lemma, which is a minor variant of the well-known regret upper bound for the AdaHedge-type update in \cref{eq:OptHedge} (the proof is provided in \Cref{app:proof_oftrl}).

\begin{lemma}\label[lemma]{lem:adahedge_opt_neg_mainbody}
Consider online linear optimization over the probability simplex $\Delta_d$.
Suppose that for some $\nu^1, \dots, \nu^T \geq 0$, the points are chosen by optimistic Hedge with the adaptive learning rate given by
$
w^t(a) \propto \exp\prn[\big]{ - \eta^\t \prn[\big]{ \sum_{s=1}^{t-1} h^\s(a) + h^\tm(a) }} 
$
for $a \in [d]$
and
$
\eta^\t = \tsqrt{{D}/{\sum_{s=1}^{t-1} \prn*{ \nrm{h^\s - m^\s}_\infty^2 + \nu^\s} }}
$
for $D = \max\set{4, \log (d) / 2^{3/2}}$ and $h^{T+1} = \zeros$.
Then, for any $w^* \in \Delta_d$, it holds that
\begin{equation}
  \sumT \inpr{w^\t - w^*, h^\t}
  \leq
  \sqrt{ 32 \sumT \prn*{ \nrm{h^\t - m^\t}_\infty^2 + \nu^\t } \, D}
  - 
  \sumT 
  \frac{1}{4 \eta^\t} \nrm{w^\t - w^\tp}_1^2  
  \per 
  \n
\end{equation}
\end{lemma}

The following lemma is useful to obtain $\Adiff$-dependent regret bounds.
\begin{lemma}\label[lemma]{lem:grad_diff}
For some $m, n \in \N$ and $A \in \R^{m \times n}$, let $\Adiff = \max_{i,j} A_{i,j} - \min_{i, j} A_{i,j}$.
Then, for any $x, x' \in \Delta_{m}$ and $y, y' \in \Delta_{n}$, it holds that
$
  \nrm{A (y - y')}_\infty
  \leq
  \frac{\max_{i} \prn*{ \max_{j} A_{ij} - \min_{j} A_{ij} }}{2} \nrm{y - y'}_1
  \leq
  \frac{\Adiff}{2} \nrm{y - y'}_1  
$
and
$
  \nrm{A^\top (x - x')}_\infty
  \leq
  \frac{\max_{j} \prn*{ \max_{i} A_{ij} - \min_{i} A_{ij} }}{2} \nrm{x - x'}_1
  \leq
  \frac{\Adiff}{2} \nrm{x - x'}_1
  .
$
\end{lemma}
\begin{proof}
We first show the first inequality.
We notice that for any $\xi \in \R^n$, it holds that
$
\nrm{A (y - y')}_\infty
=
\nrm{(A - \xi \ones^\top) (y - y')}_\infty
=
\max_{i \in [m]} \abs{ \inpr{ A(i,\cdot) - \xi(i) \ones, y - y' } }
\leq
\max_{i \in [m]} \nrm{A(i,\cdot) - \xi(i) \ones}_\infty \nrm{y - y'}_1
,
$
where the first equality follows from $y, y' \in \Delta_n$ and the last inequality from H\"{o}lder's inequality.
Choosing $\xi(i) = (\max_j A_{ij} + \min_j A_{ij}) / 2$ for each $i \in [m]$ completes the proof.
\end{proof}

The following lemma is useful for exploiting the negative term in the RVU bound to obtain a fast regret bound of $O(1)$ even when the algorithm is agnostic to $\Amax$ and $\Adiff$.
\begin{lemma}\label[lemma]{lem:Pinfty_tradeoff}
Let $z_1, \dots, z_T \in \Delta_d$, $h^\t = A z^\t$, and $c_1, c_2 > 0$.
Denote $\tau^\circ(h; c_1, c_2) \in [T] \cup \set{\infty}$ by
\begin{equation}\label{eq:tau_circ}
  \tau^\circ(h; c_1, c_2)
  =
  \begin{cases}
    \displaystyle
    \min\set[\big]{
      t 
      \colon 
      P_\infty^t(h) 
      >
      {c_2}/{c_1} + \Adiff^2 
    } & 
    \mbox{if} \ P_\infty^T(h) > {c_2}/{c_1} + \Adiff^2
    \!\com
    \\
    \infty & \text{otherwise}
    \com
  \end{cases}
\end{equation}
where $P_\infty^t(h) = \sum_{s=1}^t \nrm{h^\s - h^\sm}_\infty^2$.
Then, we have $\tau^\circ \geq 2$, and if $\tau^\circ(h; c_1, c_2) \leq T$, it holds that
\begin{equation}
  P_\infty^{\tau^\circ}(h)
  \leq
  \frac{c_2}{c_1}
  +
  2 \Adiff^2
  \com
  \
  P_\infty^{\tau^\circ - 1}(h) \geq \frac{c_2}{c_1}
  \com
  \
  c_1 \sqrt{P_\infty^{\tau^\circ}(h)}
  +
  \frac{c_2}{ \tsqrt{P_\infty^{\tau^\circ-1}(h)} }
  \leq
  \sqrt{2} c_1 \Adiff + 2 \sqrt{c_1 c_2}
  \per
  \n
\end{equation}
\end{lemma}

We are now ready to provide a proof sketch of \Cref{thm:main}.
\begin{proof}[Proof sketch of \Cref{thm:main}]
From \Cref{lem:grad_diff}, we have
\begin{equation}\label{eq:gldiff2Adiff_main}
  \nrm{g^\t - g^\tm}_\infty \leq (\Adiff / 2) \nrm{y^\t - y^\tm}_1
  \com\quad
  \nrm{\ell^\t - \ell^\tm}_\infty \leq (\Adiff / 2) \nrm{x^\t - x^\tm}_1
  \per
\end{equation}
Now we fix arbitrary $\tau_x, \tau_y \in [T]$. 
Let $\Abardiff = \Adiff /2$ for simplicity.
Then, from \Cref{lem:adahedge_opt_neg_mainbody},
\begin{align}
  \Reg_x^T
  &\leq
  \sqrt{ 32 (P_\infty^T(g) + P_\infty^T(\ell)) M_x}
  -
  \sumT \frac{1}{4 \eta_x^\t} \nrm{x^\tp - x^\t}_1^2
  \nn
  &\leq
  \sqrt{ 32
    \brk[\big]{ 
      P_\infty^{\tau_y}(g) \!+\! \Abardiff^2 Q_1^{\tau_y}(y) 
      \!+\!
      P_\infty^{\tau_x}(\ell) \!+\! \Abardiff^2 Q_1^{\tau_x}(x) 
    } M_x}
  -
  \frac{1}{4 \eta_x^{\tau_x-1}} Q_1^{\tau_x}(x)
  \com
  \label{eq:reg_x_upper_1_main}
\end{align}
where the last inequality follows since $(\eta_x^\t)_t$ is nonincreasing and \cref{eq:gldiff2Adiff_main}.
Similarly,
\begin{equation}
  \Reg_y^T
  \leq
  \sqrt{ 32
    \brk[\big]{ 
      P_\infty^{\tau_y}(g) \!+\! \Abardiff^2 Q_1^{\tau_y}(y) 
      \!+\!
      P_\infty^{\tau_x}(\ell) \!+\! \Abardiff^2 Q_1^{\tau_x}(x) 
    } M_y}
  -
  \frac{1}{4 \eta_y^{\tau_y-1}} Q_1^{\tau_y}(y)
  \per
  \label{eq:reg_y_upper_1_main}
\end{equation}
Combining the upper bounds on $\Reg_x^T$ in~\cref{eq:reg_x_upper_1_main} and $\Reg_y^T$ in~\cref{eq:reg_y_upper_1_main}, we obtain
\begin{align}
  &
  \Reg_x^T + \Reg_y^T
  \leq
  \sqrt{ 128
  \prn{ 
    P_\infty^{\tau_y}(g) 
    + 
    P_\infty^{\tau_x}(\ell)
  } M}
  +
  \sqrt{ 128 \Abardiff^2 Q_1^{\tau_y}(y) M }
  -
  \frac{1}{4 \eta_y^{\tau_y-1}} Q_1^{\tau_y}(y)
  \nn
  &\qquad\qquad\qquad\qquad\qquad\qquad\qquad\qquad\qquad\qquad+
  \sqrt{ 128 \Abardiff^2 Q_1^{\tau_x}(x) M }
  -
  \frac{1}{4 \eta_x^{\tau_x-1}} Q_1^{\tau_x}(x)
  \nn
  &\leq
  \sqrt{ \! 128
  \prn{ 
    P_\infty^{\tau_y}(g) 
    \!+\! 
    P_\infty^{\tau_x}(\ell)
  } M}
  \!+\!
  256 \Abardiff^2 M (\eta_x^{\tau_x-1} \!+\! \eta_y^{\tau_y-1})
  \!-\!
  \prn[\bigg]{
  \frac{1}{8 \eta_x^{\tau_x-1}} Q_1^{\tau_x}(x)
  +
  \frac{1}{8 \eta_y^{\tau_y-1}} Q_1^{\tau_y}(y)
  }
  \nn
  &\leq
  2^{7/2} \brk[\Bigg]{\!
  \sqrt{\!P_\infty^{\tau_x}(\ell) M}
  \!+\!
  \frac{2^{9/2} \Adiff^2 M^{3/2}}{\tsqrt{P_\infty^{\tau_x-1}(\ell)}}
  \!+\!
  \sqrt{\!P_\infty^{\tau_y}(g) M}
  \!+\!
  \frac{2^{9/2} \Adiff^2 M^{3/2}}{\tsqrt{P_\infty^{\tau_y-1}(g)}}
  \!
  }
  \!-\!
  \prn[\bigg]{\!
  \frac{1}{8 \eta_x^{\tau_x}} Q_1^{\tau_x}(x)
  \!+\!
  \frac{1}{8 \eta_y^{\tau_y}} Q_1^{\tau_y}(y)
  \!
  }
  \com
  \label{eq:social_reg_upper_1_main}
\end{align}
where in the second inequality we considered the worst cases with respect to $Q_1^{\tau_y}(y)$ and $Q_1^{\tau_x}(x)$ using the inequality $b\sqrt{z} - az \leq b^2 / (4 a)$ that holds for $a > 0$, $b \geq 0$, and $z \geq 0$,
and in the last inequality we used the subadditivity of $\sqrt{\cdot}$, $\Abardiff = \Adiff /2$, and the definitions of the learning rates in~\cref{eq:OptHedge}: $\eta_x^{\tau_x-1} \leq \nsqrt{M_x / P_\infty^{\tau_x-1}(\ell)}$ and $\eta_y^{\tau_y-1} \leq \nsqrt{M_y / P_\infty^{\tau_y-1}(g)}$.

Now using $\tau^\circ$ in \cref{eq:tau_circ},
we choose $\tau_x = \min\set{\tau^\circ(\ell), T}, \tau_y = \min\set{\tau^\circ(g), T}\in [T]$ for
\begin{equation}
  \begin{split}
    \tau^\circ(\ell)
    \coloneqq
    \tau^\circ(\ell; \sqrt{M}, 2^{9/2} \Adiff^2 M^{3/2})
    \com
    \quad
    \tau^\circ(g)
    \coloneqq
    \tau^\circ(g; \sqrt{M}, 2^{9/2} \Adiff^2 M^{3/2})    
    \per
  \end{split}
  \n
\end{equation}

In what follows, we focus on the case when $\tau^\circ(\ell), \tau^\circ(g) \leq T$.
Otherwise, the regret can be bounded by a similar argument; see \Cref{app:proof_thm_main}.
From \Cref{lem:Pinfty_tradeoff}, the terms in~\cref{eq:social_reg_upper_1_main} are evaluated as\looseness=-1
\begin{equation}
  \sqrt{P_\infty^{\tau_x}(\ell) M}
  +
  \frac{2^{9/2} \Adiff^2 M^{3/2}}{\tsqrt{P_\infty^{\tau_x-1}(\ell)}}
  \leq
  11 \Adiff M
  \com\
  \sqrt{P_\infty^{\tau_y}(g) M}
  +
  \frac{2^{9/2} \Abardiff^2 M^{3/2}}{\tsqrt{P_\infty^{\tau_y-1}(g)}}
  \leq
  11 \Adiff M
  \com
  \label{eq:PxPy_inf_main}
\end{equation}
where we used $M \geq 1$.
Continuing from \cref{eq:social_reg_upper_1_main} with~\cref{eq:PxPy_inf_main}, we have
\begin{equation}
  \Reg_x^T + \Reg_y^T
  \leq
  22 \Adiff M
  -
  \prn[\bigg]{
  \frac{1}{8 \eta_x^{\tau_x-1}} Q_1^{\tau_x}(x)
  +
  \frac{1}{8 \eta_y^{\tau_y-1}} Q_1^{\tau_y}(y)
  }
  \per
  \label{eq:social_reg_upper_2_main}
\end{equation}
Now from \Cref{lem:Pinfty_tradeoff} we have $P_\infty^{\tau_x-1}(\ell) \geq 2^{9/2} \Adiff^2 M$ and thus
\begin{equation}
  \eta_x^{\tau_x-1}
  \leq
  \sqrt{{M_x}/{P_\infty^{\tau_x-1}(\ell)}}
  \leq
  \sqrt{{M_x}/\prn{2^{9/2} \Adiff^2 M}}
  \leq
  {1}/\prn{4.75 \Adiff}
  \per
  \n
\end{equation}
Therefore, combining \cref{eq:social_reg_upper_2_main} with the fact that $\Reg_x^T + \Reg_y^T \geq 0$ and the last inequality, we obtain
$
  Q_1^{\tau_x}(x)
  \leq
  22 \eta_x^{\tau_x-1} \Adiff M
  \leq
  5 M
  \per
$
By the same argument, we also have $Q_1^{\tau_y}(y) \leq 5 M$.
From \Cref{lem:Pinfty_tradeoff}, we also have
$
  P_\infty^{\tau_x}(\ell)
  \leq
  2 \Adiff^2 (8\sqrt{2} M + 1)
  \com\
  P_\infty^{\tau_y}(g)
  \leq 
  2 \Adiff^2 (8 \sqrt{2} M + 1)
  .
$
Finally, by plugging the above upper bounds on the squared path-length quantties $P_\infty^{\tau_x}(\ell)$, $P_\infty^{\tau_y}(g)$, $Q_1^{\tau_x}(x)$, and $Q_1^{\tau_y}(y)$ in~\cref{eq:reg_x_upper_1_main,eq:reg_y_upper_1_main}, 
we obtain
$
  \Reg_x^T
  \leq
  \tsqrt{
    32
    \brk{4 \Adiff^2 (8 \sqrt{2} M + 1) + 10 \Abardiff^2 M} M_x
  }
  \leq
  8 \Adiff
  \sqrt{
    5 \prn{M + 1} M_x
  }
  ,
$
and
$
  \Reg_y^T
  \leq
  8 \Adiff
  \sqrt{
    5 \prn{M + 1} M_y
  }
  ,
$
which completes the proof of \Cref{thm:main}.
\end{proof}

\section{Scale-Invariant Learning Dynamics for Multiplayer General-Sum Games}\label{sec:multiplayer}
This section presents scale-invariant and scale-free dynamics for multiplayer general-sum games.

\subsection{Preliminaries: Swap regret to external regret}\label{subsec:swapreg2reg} 
Here we describe the reduction method of \citet{blum07external} for swap regret minimization.
We focus on the construction for player $i$.
The reduction transforms swap regret minimization into running~$m_i$ separate external-regret algorithms, which we refer to as ``experts'' and index by actions $a \in \calA_i$.
For each $a \in \calA_i$, define the utility vector fed to expert $a$ at round $t$ by
$u_{i,a}^\t = x_i^\t(a) u_i^\t \in [- \Umax x_i^\t(a), \Umax x_i^\t(a)]^{m_i}$.
Let $y_{i,a}^\t \in \Delta_{m_i}$ denote the output distribution produced by expert $a$ at round $t$.
Using $\prn{y_{i,a}^\t}_{a \in \calA_i}$, we form a transition matrix $\Qmat_i^\t \in [0,1]^{m_i \times m_i}$ by setting its $a$-th row to be $y_{i,a}^\t$, \ie~$\Qmat_i^\t(a,\cdot)=\prn{y_{i,a}^\t}^\top$.
We then choose $x_i^\t$ as a stationary distribution of the Markov chain induced by $\Qmat_i^\t$, namely, $\prn{\Qmat_i^\t}^\top x_i^\t = x_i^\t$ (for the column vector $x_i^\t$),
and use this $x_i^\t$ as player $i$'s strategy.\looseness=-1

Let 
$
\Reg_{i,a}^T = \max_{y \in \Delta_{m_i}} \sumT \inpr{y - y_{i,a}^\t, u_{i,a}^\t}
$
for
$
\Reg_{i,a}^T(y^*) = \sumT \inpr{y^* - y_{i,a}^\t, u_{i,a}^\t}
$
be the external regret of expert $a \in \calA_i$ for player $i \in [n]$.
\citet{blum07external} showed that, under the above construction, the swap regret equals the sum of the external regrets of the experts, that is,
it holds that
$
\SwapReg_{x_i}^T(M) = \sum_{a \in \calA_i} \Reg_{i,a}^T (M(a, \cdot))
.
$

\begin{algorithm}[t]

\For{$t = 1, 2, \dots, T$}{
  Compute $y_{i,a}^\t \in \Delta_{m_i}$ by OFTRL in \eqref{eq:oftrl_multiplayer} for each expert $a \in \calA_i$. \\
  Let $\Qmat_i^\t \in [0,1]^{m_i \times m_i}$ be a matrix whose $a$-th row is $y_{i,a}^\t$, that is $\Qmat_i^\t(a, \cdot) = \prn{y_{i,a}^\t}^\top$. \\
  Let $x_i^\t$ be a stationary distribution of Markov chain $\Qmat_i^\t$, that is $\prn{\Qmat_i^\t}^\top x_i^\t = x_i^\t$. \\
  Play strategy $x_i^\t \in \Delta_{m_i}$. \\
  Observe a utility vector $u_i^\t \in [-\Umax,\Umax]^{m_i}$ where $u_i^\t(a_i) = \E_{a_{-i} \sim x_{-i}^\t} \brk*{u_i(a_i, a_{-i})}$. \\
  Let ${u}_{i,a}^\t = x_i^\t(a) {u}_i^\t \in [-x_i^\t(a) \Umax, x_i^\t(a) \Umax]^{m_i}$ for each $a \in \calA_i$.
  \\
  Compute $B^\tp = \min\set{ 2^k B^\t \colon k \in \set{0} \cup \N,\, 2^{k+1} B^\t \geq \max_{i \in [n]} \nrm{u_i^\t}_\infty}$ and clip the gradients to obtain $\prn{\ubar_{i,a}^\t}_{a \in \calA_i}$ via \cref{eq:clipped_grad}.
  \label{line:doubling_clipping}
  \\
}
\caption{
  No-swap-regret algorithm of player $i$ in multiplayer general-sum games
}
\label{alg:multiple_player_swap}
\end{algorithm}

\subsection{Proposed learning dynamic}
By the reduction above, to minimize $\SwapReg_{x_i}^T(M)$, it suffices to design an external-regret minimization algorithm that guarantees small regret ${\Reg}_{i,a}^T(M(a,\cdot))$ for each expert $a \in \calA_i$.
To this end, we use OFTRL with the log-barrier regularizer together with an adaptive learning rate.
Still, obtaining scale-invariant and scale-free fast convergence calls for additional algorithmic modifications, along with the analytical techniques introduced below.
Our complete algorithm is summarized in \Cref{alg:multiple_player_swap}, and we explain its components below.

\paragraph{Doubling clipping}
The key technique for achieving scale-invariant fast convergence is the \emph{doubling clipping} method that we introduce in this paper.
In general, two standard approaches to obtaining a scale-free guarantee in online linear optimization are: (i) leveraging an AdaHedge-type analysis~\citep{orabona18scale} as in \Cref{sec:two_player}, and (ii) clipping the losses based on the maximum gradient norm observed so far~\citep{cutkosky19artificial}.
However, as we discuss in the next section, to the best of our efforts, these approaches do not work well for swap regret minimization, where one needs to aggregate the outputs of multiple external regret minimizers.

Accordingly, as implemented in \Cref{line:doubling_clipping}, at the end of each round $t$ we compute
$
B^\tp = \min\set{ 2^k B^\t : k \in \set{0} \cup \N,\, 2^{k+1} B^\t \geq \max_{i \in [n]} \nrm{u_i^\t}_\infty}.
$
This update ensures that $B^\tp \neq B^\t$ only when the maximum gradient norm across players, $\max_{i \in [n]} \nrm{u_i^\t}_\infty$, becomes at least twice the previous value $B^\t$; in that case, $B^\tp$ is the smallest value of the form $2^k B^\t$ (for some $k \in \N$) that is at least $\max_{i \in [n]} \nrm{u_i^\t}_\infty$.
Using this $B^\t$, we compute the clipped gradients $\ubar_{i,a}^\t, \ubar_i^\t$ of $u_{i,a}^\t, u_{i}^\t$ by
\begin{equation}\label{eq:clipped_grad}
  \ubar_{i,a}^\t = \frac{B^\t}{B^\tp} u_{i,a}^\t
  \com\quad
  \ubar_{i}^\t = \frac{B^\t}{B^\tp} u_{i}^\t
  \per
\end{equation}

\paragraph{Optimistic FTRL with adaptive learning rate}
We run OFTRL with the log-barrier regularizer on the clipped gradients in~\cref{eq:clipped_grad}.
Let $U^\t = \max_{s \in [t]} \max_{i \in [n]} \nrm{u_i^\s}_\infty$.
Then, for each $a \in \calA_i$ we compute $y_{i,a}^\t \in \Delta_{m_i}$ for expert $a \in \calA_i$ of player $i$ by
\begin{equation}\label{eq:oftrl_multiplayer}
  y_{i,a}^\t 
  \!=\!
  \argmax_{y \in \Delta_{m_i}}
  \set*{
    \!
    \inpr[\bigg]{y, \ubar_{i,a}^\tm \!+\! \sum_{s=1}^{t-1} \ubar_{i,a}^\s \!}
    \!-\! 
    \frac{\phi(y)}{\eta_{i}^\t}
    \!
  }
  \com \;
  \eta_{i}^\t
  \!=\!
  \min\set[\bigg]{
    \frac{\alpha_i}{\tsqrt{\gamma (U^\tm)^2 \!+\! \sum_{j \in [n]} P_\infty^\t(u_j)}}
    ,
    \frac{\beta_i}{B^\t}
  }
\end{equation}
with $\alpha_i = m_i \sqrt{\log T}$, $\beta_i = 1/(256 \sqrt{m_i})$, and $\gamma = 8 n$\footnote{If the swap-regret upper bound is allowed to depend on $\log(\Umax/\omega)$ for $\omega = \max_{i \in [n]} \nrm{u_i^{1}}_\infty$, then it suffices to set $\gamma = 8$; see also \Cref{fn:def_omega}.}.
Here, $\eta_{i}^\t$ is the learning rate of player $i$ at round $t$, $\phi(x) = - \sum_k \log(x(k))$ is the logarithmic barrier function, and recall that $P_\infty^\t(\ubar_{i}) = \sum_{s=1}^{t-1} \nrm{\ubar_{i}^\s - \ubar_{i}^\sm}_\infty^2$.
For each $j \in [n]$, if the denominator of the learning rate $\eta_j^\t$ is zero, we set $\eta_j^\t = \infty$, and set $y_{j,a}^\t = \frac{1}{m_j} \ones$, and thus $x_i^\t = \frac{1}{m_i} \ones$ for each $i \in [n]$.
Here, we let $\ubar_{i}^{(0)} = \zeros$ for simplicity.
Note also that the learning rate can be computed using only the information observed so far.

The value $\gamma (U^\tm)^2$ in the denominator of the learning rate in \cref{eq:oftrl_multiplayer} is helpful for ensuring the stability of the Markov chain (\Cref{lem:stationary_stab}) and to ensure the stability of OFTRL (see \cref{eq:swap_two_terms_pre}), even when the scale is unknown.
The learning rate is similar to that of \cref{eq:OptHedge} in that its denominator involves the opponent's gradient path length.
Thanks to a structure similar, though slightly different, to that of two-player zero-sum games, we can prove fast convergence without knowing the scale $\Umax$.
Unlike the OFTRL with log-barrier regularizer and adaptive learning rates used in \citet{tsuchiya25corrupted}, we use the same learning rate for each expert $a \in \calA_i$. This is because, when the scale $\Umax$ is unknown, assigning expert-dependent learning rates becomes difficult (see discussion below).

\subsection{Swap regret upper bounds}
Here we present regret upper bounds for \Cref{alg:multiple_player_swap}.
Let $\mmax = \max_{i \in [n]} m_i$ denote the maximum number of actions. 
Then, the above learning dynamic guarantees the following bounds.
\begin{theorem}[Scale-invariant fast convergence to correlated equilibrium]\label{thm:indiv_swapreg}
  If every player uses \Cref{alg:multiple_player_swap}, then
    $
    \SwapReg_{x_i}^T 
    \lesssim
    \Umax n^{3/2} m^{5/2} \log T
    $
  for all $i \in [n]$.
  Consequently (by \Cref{thm:swapreg2coreq}), 
  the time-averaged history of joint play after $T$ rounds is an $O(\Umax n^{3/2} m^{5/2} \log T / T)$-approximate correlated equilibrium.
\end{theorem}
The proof can be found in \Cref{app:proof_multiplayer}.
A comparison against existing bounds can be found in \Cref{table:regret_generalsum}.
To our knowledge, this is the first scale-invariant and scale-free fast convergence to correlated equilibrium in multiplayer general-sum games.
The fast convergence becomes possible thanks to the doubling-clipping technique; see \Cref{subsec:tech_swap} for details.
It is also worth noting that the swap regret upper bound and the convergence rate does not depend on $\omega \coloneqq \max_{i \in [n]} \nrm{u_i^{1}}_\infty$\footnote{\label{fn:def_omega}More precisely, the constant $\omega$ is defined as $\omega \coloneqq \max_{i \in [n]} \nrm{u_i^{\hat{t}}}_\infty$ for $\hat{t} = \min\set{t \in [T] : \max_{i \in [n]} \nrm{u_i^\t}_\infty > 0}$, which is the first time such that $\max_{i \in [n]} \nrm{u_i^\t}_\infty > 0$ holds. This paper assumes $\hat{t} = 1$ and thus $\max_{i \in [n]} \nrm{u_i^{1}} = \omega > 0$ for simplicity of presentation.}, which is related to $\delta$ in \Cref{sec:two_player}. This also becomes possible thanks to the property of the doubling clipping.

Still, the result in \Cref{thm:indiv_swapreg} has two limitations.
The first limitation is that, compared with the best known non-scale-free swap regret upper bound of $O(n m^{5/2} \log T)$~\citep{anagnostides22uncoupled}, the dependence on $n$ is worse by a factor of $\sqrt{n}$.
This is because, to exploit the negative term in the RVU bound under $\Umax$-agnosticity, our learning rates (as in \cref{eq:oftrl_multiplayer}) depend on the gradients of all players; whether this dependence can be improved is an important direction for future work.
That said, in settings where $n$ can be regarded as a constant (in particular, two-player general-sum games), our result achieves the same convergence rate as existing learning dynamics that know $\Umax$ in advance~\citep{anagnostides22uncoupled,tsuchiya25corrupted}.

The second limitation is the lack of robustness to adversarial corruption of utilities.
As discussed in \citet[Remark 8]{tsuchiya25corrupted}, achieving such robustness requires assigning different learning rates to each expert of each player.
However, when $\Umax$ is unknown, this expert-wise tuning is difficult and resolving this issue is another important direction for future work.

\subsection{Techniques used in the swap regret analysis}\label{subsec:tech_swap}
We first discuss how the doubling clipping in \cref{eq:clipped_grad} is useful for the swap-regret analysis.
Let $\calJ'$ be the set of rounds at which a jump occurs in \Cref{line:doubling_clipping} of \Cref{alg:multiple_player_swap}, and let $\calI' = [T]\setminus \calJ'$ be its complement.
Define $\calJ = \calJ' \cup \set{t + 1 \in [T] : t \in \calJ'}$ and $\calI = [T]\setminus \calJ$.
In the analysis, we handle the jump rounds and the rounds immediately after them in $\calJ$ separately from the remaining rounds in~$\calI$.

The motivation to consider $\calI$ is that, on rounds $t \in \calI$, we can guarantee the stability of OFTRL with the log-barrier in \cref{eq:oftrl_multiplayer}, and consequently the stability of the stationary distribution $Q_i^\t$ induced by the outputs $(y_{i,a}^\t)_{a\in\calA_i}$ (\Cref{lem:stationary_stab}).
This allows us to exploit the negative term in the OFTRL bound (see, \eg~\Cref{lem:oftrl_logbarrier} and \Cref{eq:negterm_lower}).
Moreover, on rounds $t \in \calI$, we can smoothly relate the squared path length of gradients to the squared path length of strategies, which is difficult with known scale-free online learning techniques because swap regret minimization requires aggregating the outputs of multiple external regret minimizers.
These benefits come from, for example, the fact that if $t \in \calI$, then $B^\tp = B^\t = B^\tm$, and hence $\ubar_{i,a}^\t = u_{i,a}^\t$ and $\ubar_{i,a}^\tm = u_{i,a}^\tm$ (see also \Cref{lem:doubling_clipping_prop,lem:BU_relation,lem:clip_grad_upper_onestep_back}).

The regret incurred on rounds in $\calJ$ is upper bounded by $O(\Umax m_i)$.
Naively, since the number of jumps is at most $O(\log(\Umax/\omega))$ (recall that $\omega = \max_{i \in [n]} \nrm{u_i^{1}}_\infty$), and each such round contributes at most $O(\Umax)$, the regret incurred on rounds in $\calJ$ is $O(\Umax m_i \log (1/\omega))$.
This can be large when $\omega$ is small for small $\omega$.
But, using the fact that, when $B^\t$ is small the utility $u_i^\t$ is also small, we can completely remove the dependence on $\omega$.
The bias due to running OFTRL on $\ubar_{i,a}^\t$ instead of $u_{i,a}^\t$ is also bounded $O(\Umax)$, which is independent of $\omega$, thanks to the property of the doubling clipping (see \Cref{lem:BU_relation,lem:clip_grad_upper_onestep_back}).

Finally, it is worth noting that we further extend the stopping-time analysis used for two-player zero-sum games.
This development allows us to leverage the negative term without knowing the scale in general-sum games (see \Cref{app:proof_swap_subsec} for details).

\bibliography{references.bib}

\begin{thebibliography}{32}
\providecommand{\natexlab}[1]{#1}
\providecommand{\url}[1]{\texttt{#1}}
\expandafter\ifx\csname urlstyle\endcsname\relax
  \providecommand{\doi}[1]{doi: #1}\else
  \providecommand{\doi}{doi: \begingroup \urlstyle{rm}\Url}\fi

\bibitem[Anagnostides et~al.(2022{\natexlab{a}})Anagnostides, Daskalakis, Farina, Fishelson, Golowich, and Sandholm]{anagnostides22near}
Ioannis Anagnostides, Constantinos Daskalakis, Gabriele Farina, Maxwell Fishelson, Noah Golowich, and Tuomas Sandholm.
\newblock Near-optimal no-regret learning for correlated equilibria in multi-player general-sum games.
\newblock In \emph{Proceedings of the 54th Annual ACM SIGACT Symposium on Theory of Computing}, page 736^^e2^^80^^93749. Association for Computing Machinery, 2022{\natexlab{a}}.

\bibitem[Anagnostides et~al.(2022{\natexlab{b}})Anagnostides, Farina, Kroer, Lee, Luo, and Sandholm]{anagnostides22uncoupled}
Ioannis Anagnostides, Gabriele Farina, Christian Kroer, Chung-Wei Lee, Haipeng Luo, and Tuomas Sandholm.
\newblock Uncoupled learning dynamics with {$O(\log T)$} swap regret in multiplayer games.
\newblock In \emph{Advances in Neural Information Processing Systems}, volume~35, pages 3292--3304. Curran Associates, Inc., 2022{\natexlab{b}}.

\bibitem[Aumann(1974)]{aumann74subjectivity}
Robert~J. Aumann.
\newblock Subjectivity and correlation in randomized strategies.
\newblock \emph{Journal of Mathematical Economics}, 1\penalty0 (1):\penalty0 67--96, 1974.

\bibitem[Blum and Mansour(2007)]{blum07external}
Avrim Blum and Yishay Mansour.
\newblock From external to internal regret.
\newblock \emph{Journal of Machine Learning Research}, 8\penalty0 (47):\penalty0 1307--1324, 2007.

\bibitem[Bowling et~al.(2015)Bowling, Burch, Johanson, and Tammelin]{bowling15heads}
Michael Bowling, Neil Burch, Michael Johanson, and Oskari Tammelin.
\newblock Heads-up limit hold'em poker is solved.
\newblock \emph{Science}, 347\penalty0 (6218):\penalty0 145--149, 2015.

\bibitem[Cesa-Bianchi and Lugosi(2006)]{cesabianchi06prediction}
Nicolo Cesa-Bianchi and G{\'a}bor Lugosi.
\newblock \emph{Prediction, Learning, and Games}.
\newblock Cambridge University Press, 2006.

\bibitem[Cesa-Bianchi et~al.(2007)Cesa-Bianchi, Mansour, and Stoltz]{cesabianchi07improved}
Nicol{\`o} Cesa-Bianchi, Yishay Mansour, and Gilles Stoltz.
\newblock Improved second-order bounds for prediction with expert advice.
\newblock \emph{Machine Learning}, 66\penalty0 (2):\penalty0 321--352, 2007.

\bibitem[Chakrabarti et~al.(2024)Chakrabarti, Grand-Cl\'{e}ment, and Kroer]{chakrabarti24extensive}
Darshan Chakrabarti, Julien Grand-Cl\'{e}ment, and Christian Kroer.
\newblock Extensive-form game solving via blackwell approachability on treeplexes.
\newblock In \emph{Advances in Neural Information Processing Systems}, volume~37, pages 35257--35287. Curran Associates, Inc., 2024.

\bibitem[Chen et~al.(2021)Chen, Luo, and Wei]{chen21impossible}
Liyu Chen, Haipeng Luo, and Chen-Yu Wei.
\newblock Impossible tuning made possible: A new expert algorithm and its applications.
\newblock In \emph{Proceedings of Thirty Fourth Conference on Learning Theory}, volume 134, pages 1216--1259. PMLR, 2021.

\bibitem[Chen and Peng(2020)]{chen20hedging}
Xi~Chen and Binghui Peng.
\newblock Hedging in games: Faster convergence of external and swap regrets.
\newblock In \emph{Advances in Neural Information Processing Systems}, volume~33, pages 18990--18999. Curran Associates, Inc., 2020.

\bibitem[Chiang et~al.(2013)Chiang, Lee, and Lu]{chiang13beating}
Chao-Kai Chiang, Chia-Jung Lee, and Chi-Jen Lu.
\newblock Beating bandits in gradually evolving worlds.
\newblock In \emph{Proceedings of the 26th Annual Conference on Learning Theory}, volume~30, pages 210--227. PMLR, 2013.

\bibitem[Cutkosky(2019)]{cutkosky19artificial}
Ashok Cutkosky.
\newblock Artificial constraints and hints for unbounded online learning.
\newblock In \emph{Proceedings of the Thirty-Second Conference on Learning Theory}, volume~99, pages 874--894. PMLR, 2019.

\bibitem[Daskalakis et~al.(2011)Daskalakis, Deckelbaum, and Kim]{daskalakis11near}
Constantinos Daskalakis, Alan Deckelbaum, and Anthony Kim.
\newblock Near-optimal no-regret algorithms for zero-sum games.
\newblock In \emph{Proceedings of the Twenty-Second Annual ACM-SIAM Symposium on Discrete Algorithms}, page 235^^e2^^80^^93254. Society for Industrial and Applied Mathematics, 2011.

\bibitem[{FAIR} et~al.(2022){FAIR}, Bakhtin, Brown, Dinan, Farina, Flaherty, Fried, Goff, Gray, Hu, Jacob, Komeili, Konath, Kwon, Lerer, Lewis, Miller, Mitts, Renduchintala, Roller, Rowe, Shi, Spisak, Wei, Wu, Zhang, and Zijlstra]{meta22human}
Meta Fundamental AI Research Diplomacy~Team {FAIR}, Anton Bakhtin, Noam Brown, Emily Dinan, Gabriele Farina, Colin Flaherty, Daniel Fried, Andrew Goff, Jonathan Gray, Hengyuan Hu, Athul~Paul Jacob, Mojtaba Komeili, Karthik Konath, Minae Kwon, Adam Lerer, Mike Lewis, Alexander~H. Miller, Sasha Mitts, Adithya Renduchintala, Stephen Roller, Dirk Rowe, Weiyan Shi, Joe Spisak, Alexander Wei, David Wu, Hugh Zhang, and Markus Zijlstra.
\newblock Human-level play in the game of {{D}}iplomacy by combining language models with strategic reasoning.
\newblock \emph{Science}, 378\penalty0 (6624):\penalty0 1067--1074, 2022.

\bibitem[Foster and Vohra(1997)]{foster97calibrated}
Dean~P. Foster and Rakesh~V. Vohra.
\newblock Calibrated learning and correlated equilibrium.
\newblock \emph{Games and Economic Behavior}, 21\penalty0 (1):\penalty0 40--55, 1997.

\bibitem[Foster et~al.(2016)Foster, Li, Lykouris, Sridharan, and Tardos]{foster16learning}
Dylan~J Foster, Zhiyuan Li, Thodoris Lykouris, Karthik Sridharan, and Eva Tardos.
\newblock Learning in games: Robustness of fast convergence.
\newblock In \emph{Advances in Neural Information Processing Systems}, volume~29, pages 4734--4742. Curran Associates, Inc., 2016.

\bibitem[Freund and Schapire(1999)]{freund99adaptive}
Yoav Freund and Robert~E. Schapire.
\newblock Adaptive game playing using multiplicative weights.
\newblock \emph{Games and Economic Behavior}, 29\penalty0 (1):\penalty0 79--103, 1999.

\bibitem[Hart and Mas-Colell(2000)]{hart00simple}
Sergiu Hart and Andreu Mas-Colell.
\newblock A simple adaptive procedure leading to correlated equilibrium.
\newblock \emph{Econometrica}, 68\penalty0 (5):\penalty0 1127--1150, 2000.

\bibitem[Koolen and Van~Erven(2015)]{koolen15second}
Wouter~M. Koolen and Tim Van~Erven.
\newblock Second-order quantile methods for experts and combinatorial games.
\newblock In \emph{Proceedings of The 28th Conference on Learning Theory}, volume~40, pages 1155--1175. PMLR, 2015.

\bibitem[Luo and Schapire(2015)]{luo2015achieving}
Haipeng Luo and Robert~E. Schapire.
\newblock Achieving all with no parameters: {AdaNormalHedge}.
\newblock In \emph{Proceedings of The 28th Conference on Learning Theory}, volume~40, pages 1286--1304. PMLR, 2015.

\bibitem[Mhammedi et~al.(2019)Mhammedi, Koolen, and Van~Erven]{mhammedi19lipschitz}
Zakaria Mhammedi, Wouter~M Koolen, and Tim Van~Erven.
\newblock Lipschitz adaptivity with multiple learning rates in online learning.
\newblock In \emph{Proceedings of the Thirty-Second Conference on Learning Theory}, volume~99, pages 2490--2511. PMLR, 2019.

\bibitem[Morav{\v{c}}{\'\i}k et~al.(2017)Morav{\v{c}}{\'\i}k, Schmid, Burch, Lis{\`y}, Morrill, Bard, Davis, Waugh, Johanson, and Bowling]{moravvcik17deepstack}
Matej Morav{\v{c}}{\'\i}k, Martin Schmid, Neil Burch, Viliam Lis{\`y}, Dustin Morrill, Nolan Bard, Trevor Davis, Kevin Waugh, Michael Johanson, and Michael Bowling.
\newblock Deepstack: Expert-level artificial intelligence in heads-up no-limit poker.
\newblock \emph{Science}, 356\penalty0 (6337):\penalty0 508--513, 2017.

\bibitem[Munos et~al.(2024)Munos, Valko, Calandriello, Gheshlaghi~Azar, Rowland, Guo, Tang, Geist, Mesnard, Fiegel, Michi, Selvi, Girgin, Momchev, Bachem, Mankowitz, Precup, and Piot]{munos24nash}
Remi Munos, Michal Valko, Daniele Calandriello, Mohammad Gheshlaghi~Azar, Mark Rowland, Zhaohan~Daniel Guo, Yunhao Tang, Matthieu Geist, Thomas Mesnard, C\^{o}me Fiegel, Andrea Michi, Marco Selvi, Sertan Girgin, Nikola Momchev, Olivier Bachem, Daniel~J Mankowitz, Doina Precup, and Bilal Piot.
\newblock {N}ash learning from human feedback.
\newblock In \emph{Proceedings of the 41st International Conference on Machine Learning}, volume 235, pages 36743--36768. PMLR, 2024.

\bibitem[Nesterov and Nemirovskii(1994)]{nesterov94interior}
Yurii Nesterov and Arkadii Nemirovskii.
\newblock \emph{Interior-Point Polynomial Algorithms in Convex Programming}.
\newblock Society for Industrial and Applied Mathematics, 1994.

\bibitem[Orabona and P^^c3^^a1l(2018)]{orabona18scale}
Francesco Orabona and D^^c3^^a1vid P^^c3^^a1l.
\newblock Scale-free online learning.
\newblock \emph{Theoretical Computer Science}, 716:\penalty0 50--69, 2018.
\newblock Special Issue on ALT 2015.

\bibitem[Perolat et~al.(2022)Perolat, Vylder, Hennes, Tarassov, Strub, de~Boer, Muller, Connor, Burch, Anthony, McAleer, Elie, Cen, Wang, Gruslys, Malysheva, Khan, Ozair, Timbers, Pohlen, Eccles, Rowland, Lanctot, Lespiau, Piot, Omidshafiei, Lockhart, Sifre, Beauguerlange, Munos, Silver, Singh, Hassabis, and Tuyls]{perolat22mastering}
Julien Perolat, Bart~De Vylder, Daniel Hennes, Eugene Tarassov, Florian Strub, Vincent de~Boer, Paul Muller, Jerome~T. Connor, Neil Burch, Thomas Anthony, Stephen McAleer, Romuald Elie, Sarah~H. Cen, Zhe Wang, Audrunas Gruslys, Aleksandra Malysheva, Mina Khan, Sherjil Ozair, Finbarr Timbers, Toby Pohlen, Tom Eccles, Mark Rowland, Marc Lanctot, Jean-Baptiste Lespiau, Bilal Piot, Shayegan Omidshafiei, Edward Lockhart, Laurent Sifre, Nathalie Beauguerlange, Remi Munos, David Silver, Satinder Singh, Demis Hassabis, and Karl Tuyls.
\newblock Mastering the game of {Stratego} with model-free multiagent reinforcement learning.
\newblock \emph{Science}, 378\penalty0 (6623):\penalty0 990--996, 2022.

\bibitem[Rakhlin and Sridharan(2013)]{rakhlin13optimization}
Sasha Rakhlin and Karthik Sridharan.
\newblock Optimization, learning, and games with predictable sequences.
\newblock In \emph{Advances in Neural Information Processing Systems}, volume~26, pages 3066--3074. Curran Associates, Inc., 2013.

\bibitem[Swamy et~al.(2024)Swamy, Dann, Kidambi, Wu, and Agarwal]{swamy24minimaxlist}
Gokul Swamy, Christoph Dann, Rahul Kidambi, Steven Wu, and Alekh Agarwal.
\newblock A minimaximalist approach to reinforcement learning from human feedback.
\newblock In \emph{Proceedings of the 41st International Conference on Machine Learning}, volume 235, pages 47345--47377. PMLR, 2024.

\bibitem[Syrgkanis et~al.(2015)Syrgkanis, Agarwal, Luo, and Schapire]{syrgkanis15fast}
Vasilis Syrgkanis, Alekh Agarwal, Haipeng Luo, and Robert~E Schapire.
\newblock Fast convergence of regularized learning in games.
\newblock In \emph{Advances in Neural Information Processing Systems}, volume~28, pages 2989--2997. Curran Associates, Inc., 2015.

\bibitem[Tsuchiya(2025)]{tsuchiya25tight}
Taira Tsuchiya.
\newblock Tight regret upper and lower bounds for optimistic {Hedge} in two-player zero-sum games.
\newblock \emph{arXiv preprint arXiv:2510.11691}, 2025.

\bibitem[Tsuchiya et~al.(2025)Tsuchiya, Ito, and Luo]{tsuchiya25corrupted}
Taira Tsuchiya, Shinji Ito, and Haipeng Luo.
\newblock Corrupted learning dynamics in games.
\newblock In \emph{Proceedings of Thirty Eighth Conference on Learning Theory}, volume 291, pages 5506--5552. PMLR, 2025.

\bibitem[Zhang et~al.(2025)Zhang, Anagnostides, and Sandholm]{zhang25scale}
Brian~Hu Zhang, Ioannis Anagnostides, and Tuomas Sandholm.
\newblock Scale-invariant regret matching and online learning with optimal convergence: Bridging theory and practice in zero-sum games.
\newblock \emph{arXiv preprint arXiv:2510.04407}, 2025.

\end{thebibliography}
\bibliographystyle{plainnat}

\newpage
\appendix

\crefalias{section}{appendix}
\crefalias{subsection}{appendix}
\crefalias{subsubsection}{appendix}

\paragraph{Additional notation for appendix}
In the appendix, we additionally use the following notation.
Let $D_\psi(x,y)$ denote the Bregman divergence between $x$ and $y$ induced by a differentiable convex function $\psi$, that is,
$
  D_{\psi}(x, y) 
  = 
  \psi(x) - \psi(y) - \inpr{\nabla \psi(y),x - y}
$.
We write $\nrm{\cdot}_*$ for the dual norm associated with a norm $\nrm{\cdot}$.
We use $\nrm{h}_{x,f} = \sqrt{h^\top \nabla^2 f(x) h}$
and $\nrm{h}_{*,x,f} = \sqrt{h^\top \prn{\nabla^2 f(x)}^{-1} h}$ to denote the local norm and its dual norm of a vector $h$ at a point $x$ with respect to a convex function~$f$, respectively.

\section{Regret Analysis of Optimistic Follow-the-Regularized-Leader}\label{app:proof_oftrl}

This section provides regret upper bounds of optimistic follow-the-regualrized leader (OFTRL) for online linear optimization.
First, we provide a general analysis of OFTRL. 
Then, we present RVU bounds for OFTRL with the negative Shannon entropy and the log-barrier regularizer, 
which are used respectively in two-player zero-sum games in \Cref{app:proof_two_player} and multiplayer general-sum games in \Cref{app:proof_multiplayer}.
The notation in this section follows that of the online linear optimization setup described in \Cref{subsec:olo}.

\subsection{Common analysis}
The following lemma provides a regret bound for the optimistic follow-the-regularized-leader (OFTRL) (which is adpoted from \citealt[Lemma 16]{tsuchiya25corrupted}).
\begin{lemma}\label[lemma]{lem:oftrl_bound}
  Let $\calK \subseteq \R^d$ be a nonempty closed and bounded convex set.
  Let
  $w^\t \in \argmin_{x \in \calK} \set{\inpr{x, m^\t + \sum_{s=1}^{t-1} h^\s} + \psi^\t(x)}$ be the output of OFTRL at round $t$.
  Then, for any $w^* \in \calK$,
  \begin{align}
    &\sumT \inpr{w^\t - w^*, h^\t}
    \leq
    \psi^{T+1}(w^*) - \psi^1(w^1)
    +
    \sumT \prn*{ \psi^\t(w^\tp) - \psi^\tp(w^\tp)}
    \nn 
    &\qquad+
    \sumT 
    \prn*{
      \inpr{w^\t - w^\tp, h^\t - m^\t}
      -
      D_{\psi^\t}(w^\tp, w^\t)
    }
    +
    \inpr{w^* - w^{T+1}, m^{T+1}}
    \per 
    \label{eq:oftrl_bound}
  \end{align}
\end{lemma}

\subsection{Analysis for optimistic Hedge}

\Cref{lem:oftrl_bound} yields the following regret upper bound for OFTRL with negative Shannon entropy, which corresponds to optimistic Hedge.
\begin{lemma}[RVU bound for optimistic Hedge]\label[lemma]{lem:oftrl_shannon}
  Let
  $\psi^\t(x) = - \frac{1}{\eta^\t} H(x)$ for
  $H(x) = \sum_{k=1}^d x(k) \log(1/x(k))$ be the negative Shannon entropy regularizer with nonincreasing learning rate $\eta^\t$
  and 
  $w^\t \in \argmin_{x \in \Delta_d} \set{\inpr{x, m^\t + \sum_{s=1}^{t-1} h^\s} + \psi^\t(x)}$ be the output of OFTRL at round $t$.
  Then, for any $w^* \in \Delta_d$,
  \begin{align}
    \sumT \inpr{w^\t - w^*, h^\t}
    &\leq
    \frac{\log d}{\eta^{T+1}}
    +
    \sumT
    \min\set*{2 \nrm{h^\t - m^\t}_\infty, \eta^\t \nrm{h^\t - m^\t}_\infty^2}
    \nn
    &\qquad\qquad\qquad- 
    \sumT 
    \frac{1}{4 \eta^\t} \nrm{w^\t - w^\tp}_1^2  
    + 2 \nrm{m^{T+1}}_\infty 
    \per 
    \label{eq:oftrl_bound_shannon_rvu_pre}
  \end{align}
\end{lemma}
\begin{proof}
We will upper bound the RHS of \cref{eq:oftrl_bound} in \Cref{lem:oftrl_bound}.
Since $H(x) \leq \log d$ for all $x \in \Delta_d$,
we have
\begin{align}
  \psi^{T+1}(w^*) 
  - \psi^1(w^1)
  +
  \sumT \prn*{ \psi^\t(w^\tp) - \psi^{T+1}(w^\tp)}
  \leq 
  \frac{\log d}{\eta^1}
  +
  \sumT \prn*{\frac{1}{\eta^\tp} - \frac{1}{\eta^\t}} \log d
  =
  \frac{\log d}{\eta^{T+1}}
  \per
  \n
\end{align}
Since $\psi^\t$ is $(1/\eta^\t)$-strongly convex with respect to~$\nrm{\cdot}_1$, we also have 
\begin{align}
  &
  \inpr{w^\t - w^\tp, h^\t - m^\t}
  -
  D_{\psi^\t}(w^\tp, w^\t)
  \nn
  &\leq 
  \nrm{w^\t - w^\tp}_1 \nrm{h^\t - m^\t}_\infty
  -
  \frac{1}{2 \eta^\t} \nrm{w^\t - w^\tp}_1^2
  \nn 
  &=
  \nrm{w^\t - w^\tp}_1 \nrm{h^\t - m^\t}_\infty
  -
  \frac{1}{4\eta^\t} \nrm{w^\t - w^\tp}_1^2
  -
  \frac{1}{4\eta^\t} \nrm{w^\t - w^\tp}_1^2
  \nn 
  &\leq 
  \min\set*{2 \nrm{h^\t - m^\t}_\infty, \eta^\t \nrm{h^\t - m^\t}_\infty^2}
  - 
  \frac{1}{4 \eta^\t} \nrm{w^\t - w^\tp}_1^2
  \com 
  \n
\end{align}
where the first inequality follows from H\"older's inequality and the $(1/\eta^\t)$-strong convexity of $\psi^\t$ with respect to $\nrm{\cdot}_1$, and the last inequality follows from $b \sqrt{z} - a z \leq b^2 / (4a)$ for $a > 0, b \geq 0$, and $z \geq 0$.
Combining \Cref{lem:oftrl_bound} with the three inequalities completes the proof of \Cref{lem:oftrl_shannon}.
\end{proof}

The following lemma, which follows from \Cref{lem:oftrl_shannon} and by choosing an adaptive learning rate, is a variant of well-known upper bounds in the literature of scale-free online learning (see, \eg~\citealp{orabona18scale}).
The difference is that we use OFTRL rather than FTRL here, and we keep the negative term in order to derive an RVU bound.
\begin{lemma}[RVU bound for optimistic Hedge with AdaHedge-type learning rate]\label[lemma]{lem:adahedge_opt_neg}
Suppose that the same assumptions as \Cref{lem:oftrl_shannon} hold.
We also assume that the learning rate is given by $\eta^\t = {\sqrt{c^2 / \sum_{s=1}^{t-1} \prn{ \nrm{h^\s - m^\s}_\infty^2 + \nu^\s } }}$ for some constant $c \geq 2$ and nonnegative reals $\nu^1, \dots, \nu^T$, and let $m^{T+1} = \zeros$.
Then, for any $w^* \in \Delta_d$,
\begin{equation}
  \sumT \inpr{w^\t - w^*, h^\t}
  \leq
  \prn*{\frac{\log d}{c} + 2 \sqrt{2} c}
  \sqrt{ \sumT \prn*{ \nrm{h^\t - m^\t}_\infty^2 + \nu^\t } }
  - 
  \sumT 
  \frac{1}{4 \eta^\t} \nrm{w^\t - w^\tp}_1^2  
  \per 
  \n
\end{equation}
In particular, when $c = c^* \coloneqq \max\set{2, \sqrt{\log d} / 2^{3/4}}$, 
for any $w^* \in \Delta_d$,
\begin{equation}
  \sumT \inpr{w^\t - w^*, h^\t}
  \leq
  \sqrt{ \sumT \prn*{ \nrm{h^\t - m^\t}_\infty^2 + \nu^\t } \log_+(d)}
  - 
  \sumT 
  \frac{1}{4 \eta^\t} \nrm{w^\t - w^\tp}_1^2  
  \com
  \n
\end{equation}
where we defined
\begin{equation}
  \log_+(d)
  =
  \begin{cases}
    ({\log (d)}/{2} + 4 \sqrt{2})^2 & \mbox{if} \ \log d \leq 2^{7/2}  \\
    2^{7/2} \log d & \mbox{if} \ \log d > 2^{7/2}
  \end{cases}
  =
  O(\log d)
  \per
  \n
\end{equation}
\end{lemma}

\Cref{lem:adahedge_opt_neg} immediate yields \Cref{lem:adahedge_opt_neg_mainbody} since
it holds that $\log_+(d) \leq 32 \max\set{4, \log d / 2^{3/2}}$.
\begin{proof}
We will upper bound the RHS of \cref{eq:oftrl_bound_shannon_rvu_pre}.
We denote $\Delta^\t = \nrm{h^\t - m^\t}_\infty^2$ for simplicity.
Then, 
\begin{align}
  \min\set*{2 \Delta^\t, \eta^\t (\Delta^\t)^2}
  &\leq
  \min\set*{2 \Delta^\t, \frac{c (\Delta^\t)^2}{\sqrt{\sum_{s=1}^{t-1} (\Delta^\s)^2}}}
  =
  \sqrt{
  \min\set*{4 (\Delta^\t)^2, \frac{c^2 (\Delta^\t)^4}{\sum_{s=1}^{t-1} (\Delta^\s)^2}}
  }
  \nn
  &\leq
  \sqrt{
  \frac{2}{\frac{1}{4 (\Delta^\t)^2} + \frac{\sum_{s=1}^{t-1} (\Delta^\s)^2}{c^2 (\Delta^\t)^4}}
  }
  =
  \frac{\sqrt{2} c (\Delta^\t)^2}{ \sqrt{ (c^2 / 4) (\Delta^\t)^2 + \sum_{s=1}^{t-1} (\Delta^\s)^2} }
  \per
  \n
\end{align}
From this inequality, the second term in \cref{eq:oftrl_bound_shannon_rvu_pre} is evaluated as
\begin{equation}
  \sumT \min\set*{2 \Delta^\t, \eta^\t (\Delta^\t)^2}
  \leq
  \sumT 
  \frac{\sqrt{2} c (\Delta^\t)^2}{ \sqrt{ (c^2 / 4) (\Delta^\t)^2 + \sum_{s=1}^{t-1} (\Delta^\s)^2} }
  \leq
  2 \sqrt{2} c
  \sqrt{\sumT (\Delta^\t)^2}
  \com
\end{equation}
where the last inequality follows from the assumption that $c \geq 2$ 
and the fact that
$\sumT z^\t / \sqrt{\sum_{s=1}^t z^\s} \leq 2 \sqrt{\sumT z^\t}$ for $z^{1}, \dots, z^\T \geq 0$.
Therefore, \cref{eq:oftrl_bound_shannon_rvu_pre} implies
\begin{equation}
  \sumT \inpr{w^\t - w^*, h^\t}
  \leq
  \prn*{\frac{\log d}{c} + 2 \sqrt{2} c}
  \sqrt{ \sumT \prn*{ \nrm{h^\t - m^\t}_\infty^2 + \nu^\t } }
  - 
  \sumT 
  \frac{1}{4 \eta^\t} \nrm{w^\t - w^\tp}_1^2  
  \com
  \label{eq:oftrl_bound_shannon_rvu_pre_2}
\end{equation}
which is the first desired inequality.

Now let $f(c) = \frac{\log d}{c} + 2 \sqrt{2} c$ for $c \geq 2$.
Then, $c^* = \max\set{2, \sqrt{\log d} / 2^{3/4}}$ minimizes $f$ and the optimal value satisfies
\begin{equation}
  f(c^*)
  =
  \begin{cases}
    {\log (d)}/{2} + 4 \sqrt{2} & \mbox{if} \ \log d \leq 2^{7/2} \\
    2^{7/4} \sqrt{\log d} & \mbox{if} \ \log d > 2^{7/2} 
  \end{cases}
  \leq
  \begin{cases}
    8 \sqrt{2} & \mbox{if} \ \log d \leq 2^{7/2} \\
    2^{7/4} \sqrt{\log d} & \mbox{if} \ \log d > 2^{7/2} 
  \end{cases}
  =
  O(\sqrt{\log d})
  \per
  \n
\end{equation}
Since we have $\log_+(d) = (f(c^*))^2$, 
combining this with \cref{eq:oftrl_bound_shannon_rvu_pre_2} gives the second desired upper bound, and we complete the proof.
\end{proof}

\subsection{Analysis for OFTRL with log-barrier regularization}
Here we provide RVU bounds for OFTRL with log-barrier regularization, slightly generalizing RVU bounds shown in \citet{tsuchiya25corrupted}.
The following lemma is a generalized variant of \citet[Lemma 23]{tsuchiya25corrupted}.
\begin{lemma}[RVU bound for OFTRL with self-concordant barrier and adaptive learning rate]\label[lemma]{lem:oftrl_selfconcordant}
  Let $\calK \subseteq \R^d$ be a nonempty closed convex set with a diameter $D = \max_{v, w \in \calK} \nrm{v - w}$ for a norm $\nrm{\cdot}$.
  Let $\phi$ be a $\vartheta$-self-concordant barrier for $\calK$
  and 
  $\psi^\t(w) = \frac{1}{\eta^\t} \phi(w)$ be a regularizer with nonincreasing and nonnegative learning rate $\prn{\eta^\t}_{t=1}^T$.
  For this $\psi^\t$, consider the OFTRL update
  $w^\t \in \argmin_{w \in \calK} \set{\inpr{w, m^\t + \sum_{s=1}^{t-1} h^\s} + \psi^\t(w)}$.
  Suppose that for some $\calI \in [T]$, the sequence of iterates $\prn{w^\t}_{t=1}^T$ satisfies
  $\nrm{w^\tp - w^\t}_{w^\t,\phi} \leq 1/2$ for all $t \in \calI$. 
  Then, for any $w^* \in \calK$,
  \begin{align}
    \sumT \inpr{w^\t - w^*, h^\t}
    &\leq
    \frac{\vartheta \log T}{\eta^\Tp}
    +
    4 \sum_{t \in \calI}
    \eta^\t \nrm{h^\t - m^\t}_{*,w^\t,\phi}^2
    \nn
    &\qquad- 
    \sum_{t \in \calI}
    \frac{1}{16 \eta^\t} \nrm{w^\tp - w^\t}_{w^\t, \phi}^2
    + 
    D \sum_{t \not\in \calI} \nrm{h^\t - m^\t}_*
    +
    3 D L
    \com 
    \n
  \end{align}
  where 
  $L = \max \set{ \max_{t \in [T]} \nrm{h^\t}_*,  \max_{t \in [T+1]} \nrm{m^\t}_*  }$.
\end{lemma}

\begin{proof}
The proof mostly follows that of \citet[Lemma 23]{tsuchiya25corrupted}, in which we upper bound the RHS of \cref{eq:oftrl_bound}.
One can replace Eq.~(12) in \citet{tsuchiya25corrupted} with the following inequality to see that the above lemma indeed holds:
if $t \in \calI$ we have
\begin{equation}
  \inpr{w^\t - w^\tp, h^\t - m^\t} - D_{\psi^\t}(w^\tp, w^\t)
  \leq
  4 \eta^\t \nrm{h^\t - m^\t}_{*,w^\t,\phi}^2 - \frac{1}{16 \eta^\t} \nrm{w^\tp - w^\t}_{w^\t, \phi}^2 
  \com
  \n
\end{equation}
and otherwise 
\begin{equation}
  \inpr{w^\t - w^\tp, h^\t - m^\t} - D_{\psi^\t}(w^\tp, w^\t)
  \leq
  \nrm{w^\t - w^\tp} \nrm{h^\t - m^\t}_*
  \leq
  D \nrm{h^\t - m^\t}_*
  \com
  \n
\end{equation}
where the first inequality follows from H\"{o}lder's inequality.
\end{proof}

The following lemma is a direct consequence of the above lemma and is a variant of \citet[Lemma 24]{tsuchiya25corrupted}:
\begin{restatable}[RVU bound for OFTRL with log-barrier regularizer and adaptive learning rate]{lemma}{lemoftrllogbarrier}\label[lemma]{lem:oftrl_logbarrier}
  Let
  $\psi^\t(w) = \frac{1}{\eta^\t} \phi(w)$ for $\phi(w) = - \sum_{k=1}^d \log(w(k))$ be the logarithmic barrier regularizer with nonincreasing learning rate $\eta^\t$
  and 
  $w^\t \in \argmin_{w \in \Delta_d} \set{\inpr{w, m^\t + \sum_{s=1}^{t-1} h^\s} + \psi^\t(w)}$ be the output of OFTRL at round $t$.
  Suppose that for some $\calI \subseteq [T]$, the sequence of iterates $\prn{w^\t}_{t=1}^T$ satisfies
  $\nrm{w^\tp - w^\t}_{w^\t,\phi} \leq 1/2$ for all $t \in \calI$. 
  Then, for any $w^* \in \Delta_d$,
  \begin{align}
    \sumT \inpr{w^\t - w^*, h^\t}
    &\leq
    \frac{d \log T}{\eta^\Tp}
    +
    4 \sum_{t \in \calI} \eta^\t \nrm{h^\t \!-\! m^\t}_{*,w^\t,\phi}^2
    \nn
    &\qquad-
    \sum_{t \in \calI}
    \frac{1}{16 \eta^\t} 
    \nrm{w^\tp - w^\t}_{w^\t, \phi}^2
    +
    2 \sum_{t \not\in \calI} \nrm{h^\t - m^\t}_\infty
    +
    6 L 
    \com 
    \label{eq:oftrl_logbarrier}
  \end{align}
  where $L = \max \set{ \max_{t \in [T]} \nrm{h^\t}_\infty,  \max_{t \in [T+1]} \nrm{m^\t}_\infty  }$.
\end{restatable}

\begin{proof}
Combining \Cref{lem:oftrl_selfconcordant} with $\nrm{\cdot} = \nrm{\cdot}_1$, the fact that the logarithmic barrier function $\phi$ is a $d$-self-concordant barrier, and $\max_{x, y \in \Delta_d} \nrm{x - y}_1 = 2$ yields the desired bound.
\end{proof}

\section{Deferred Proofs for Two-Player Zero-Sum Games from \Cref{sec:two_player}}\label{app:proof_two_player}
This section provides the details and deferred proofs from \Cref{sec:two_player}.

\subsection{Proof of \Cref{lem:Pinfty_tradeoff}}

\begin{proof}[Proof of \Cref{lem:Pinfty_tradeoff}]
For simplicty, we write $\tau^\circ = \tau^\circ(h; c_1, c_2)$, omitting the dependence on $h, c_1, c_2$ in this proof.
Since $\tau^\circ \leq T$, the definition of $\tau^\circ$ implies
\begin{equation}
  P_\infty^{\tau^\circ - 1}(h) 
  \leq 
  c_2 / c_1 + \Adiff^2
  \per
  \n
\end{equation}
We then have 
\begin{equation}
  P_\infty^{\tau^\circ}(h)
  =
  P_\infty^{\tau^\circ - 1}(h)
  +
  \nrm{h^{\tau^\circ - 1} - h^{\tau^\circ - 2}}_\infty^2
  \leq
  \frac{c_2}{c_1}
  +
  \Adiff^2
  +
  \frac{\Adiff^2}{4} \nrm{z^{\tau^\circ} - z^{\tau^\circ - 1}}_1^2
  \leq
  \frac{c_2}{c_1}
  +
  2 \Adiff^2
  \com
  \n
\end{equation}
where the first inequality follows 
$\nrm{h^{\tau^\circ - 1} - h^{\tau^\circ - 2}}_\infty^2 \leq ({\Adiff}/{2})^2 \nrm{z^{\tau^\circ - 1} - z^{\tau^\circ - 2}}_1^2$ by \Cref{lem:grad_diff} and the last inequality from $\nrm{z^{\tau^\circ} - z^{\tau^\circ - 1}}_1^2 \leq 4$ since $z^\t \in \Delta_d$.
In a similar manner, we can show that
\begin{equation}
  P_\infty^{\tau^\circ - 1}(h)
  =
  P_\infty^{\tau^\circ}(h)
  -
  \nrm{h^{\tau^\circ - 1} - h^{\tau^\circ - 2}}_\infty^2
  \geq
  \frac{c_2}{c_1}
  +
  \Adiff^2
  -
  \frac{\Adiff^2}{4} \nrm{z^{\tau^\circ} - z^{\tau^\circ - 1}}_1^2
  \geq
  \frac{c_2}{c_1}
  \per
  \n
\end{equation}
Combining the last two inequalities, we obtain
\begin{equation}
  c_1 \sqrt{P_\infty^{\tau^\circ}(h)}
  +
  \frac{c_2}{ \sqrt{P_\infty^{\tau^\circ-1}(h)} }
  \leq
  c_1 
  \sqrt{
    \frac{c_2}{c_1}
    +
    2 \Adiff^2
  }
  +
  \frac{c_2}{\sqrt{c_2 / c_1}}
  \leq
  \sqrt{2} c_1 \Adiff + 2 \sqrt{c_1 c_2}
  \com
  \n
\end{equation}
which completes the proof.
\end{proof}

\subsection{Corrupted procedure in two-player zero-sum games}\label{subsec:corrupted_game_two_player}
Before describing the deferred arguments for the proof of \Cref{thm:main} and the proof of robustness against opponent deviations and adversarially corrupted utilities, we first formally introduce the corrupted regime~\citep{tsuchiya25corrupted}.
Let $\xhat^\t \in \Delta_{\mx}$ and $\yhat^\t \in \Delta_{\my}$ denote the strategies suggested by the prescribed algorithm of the $x$-player and $y$-player, respectively, at round $t$, and let $x^\t \in \Delta_{\mx}$ and $y^\t \in \Delta_{\my}$ denote the strategy actually chosen by each player at round $t$.
We then summarize the corrupted learning procedure for a two-player zero-sum game with payoff matrix $A$:
\begin{mdframed}
At each round $t = 1, \dots, T$: 
\begin{enumerate}[topsep=2pt, itemsep=1pt, partopsep=0.5pt, leftmargin=27pt]
  \item A prescribed algorithm suggests strategies $\xhat^\t \in \Delta_{\mx}$ and $\yhat^\t \in \Delta_{\my}$;
  \item The $x$-player selects a strategy $x^\t \leftarrow \xhat^\t + \hat{c}_\xrm^\t$ and the $y$-player selects $y^\t \leftarrow \yhat^\t + \hat{c}_\yrm^\t$; 
  \item The $x$-player observes a corrupted expected reward vector $\gtil^\t = g^\t + \tilde{c}_\xrm^\t$ for $g^\t = A y^\t$ and the $y$-player observes a corrupted expected loss vector $\ltil^\t = \ell^\t + \tilde{c}_\yrm^\t$ for $\ell^\t = A^\top x^\t$;
  \item The $x$-player gains a payoff of $\inpr{x^\t, g^\t}$ in Setting (I) and $\inpr{x^\t, \gtil^\t}$ in Setting (II), and the $y$-player incurs a loss of $\inpr{y^\t, \ell^\t}$ in Setting (I) and $\inpr{y^\t, \ltil^\t}$ in Setting (II);
\end{enumerate}
\end{mdframed}
Here, $\hat{c}_\xrm^\t$ and $\tilde{c}_\xrm^\t$ are deviated and corruption vectors of strategies and utility for the $x$-player at round $t$, respectively.
This satisfies $\sumT \nrm{\hat{c}_\xrm^\t}_1 = \sumT \nrm{ x^\t - \xhat^\t }_1 \leq \hat{C}_\xrm$, 
$\sumT \nrm{\tilde{c}_\xrm^\t}_\infty = \sumT \nrm{ g^\t - \gtil^\t }_\infty \leq \tilde{C}_\xrm$, and $C_\xrm = \Amax \hat{C}_\xrm + 2 \tilde{C}_\xrm$\footnote{The choice of the coefficients here follows from \Cref{prop:reg_relation}.}. 
Similarly, $\hat{c}_\yrm^\t$ and $\tilde{c}_\yrm^\t$ are corruption levels of strategies and utility of the $y$-player, respectively, such that 
$\sumT \nrm{\hat{c}_\yrm^\t}_1 = \sumT \nrm{ y^\t - \yhat^\t }_1 \leq \hat{C}_\yrm$, 
$\sumT \nrm{\tilde{c}_\yrm^\t}_\infty = \sumT \nrm{ \ell^\t - \ltil^\t }_\infty \leq \tilde{C}_\yrm$, and $C_\yrm = \Amax \hat{C}_\yrm + 2 \tilde{C}_\yrm$. 
Note that the corrupted regime with $\Chatx = \Chaty = \Ctilx = \Ctily = 0$ corresponds to the \emph{honest regime} presented in \Cref{subsec:two_player_preliminaries}, in which there is no deviation from the prescribed strategies and no corruption in utilities.
See \citet[Remark 8]{tsuchiya25corrupted} for the motivation of considering two different settings (I) and (II).

In the corrupted regime, the following four types of external regret can be defined:
\begin{align}
\Reg_{x,g}^T(x^*) &= \sumT \inpr{x^* - x^\t, g^\t}\com
\quad \Reg_{\xhat,g}^T(x^*) = \sumT \inpr{x^* - \xhat^\t, g^\t}
\com \nn
\Reg_{x,\gtil}^T(x^*) &= \sumT \inpr{x^* - x^\t, \gtil^\t}\com
\quad
\Reg_{\xhat,\gtil}^T(x^*) = \sumT \inpr{x^* - \xhat^\t, \gtil^\t}\per
\n
\end{align}
As discussed in \citet[Remark 8]{tsuchiya25corrupted}, depending on the cause of the corruption in the utility vector $g^\t$, it is natural to consider either $\Reg_{x,g}^T(x^*)$ or $\Reg_{x,\gtil}^T(x^*)$ as the evaluation metric for the player.
Specifically, in Setting~(I), it is natural to use $\Reg_{x,g}^T(x^*)$, while in Setting~(II), it is natural to use $\Reg_{x,\util_i}^T(x^*)$.
The four different types of external regret of the $y$-player, 
$
\Reg_{y,\ell}^T(y^*)
$,
$\Reg_{\yhat,\ell}^T(y^*)$,
$\Reg_{y,\ltil}^T(y^*)$, and
$
\Reg_{\yhat,\ltil}^T(y^*)
$
can be similarly defined.
Note that we have $\Reg_{x}^T(x^*) = \Reg_{x,g}^T(x^*)$ and $\Reg_{y}^T(y^*) = \Reg_{y,\ell}^T(y^*)$.

For these external regrets, the following inequalities hold, which are minor generalizations of \citet[Proposition 9]{tsuchiya25corrupted} for the scale-free scenario.
\begin{proposition}\label[proposition]{prop:reg_relation}
  For any $x^* \in \Delta_m$,
  $\abs{\Reg_{x,g}^T(x^*) - \Reg_{\xhat,g}^T(x^*)} \leq \Amax \hat{C}_x$,
  $\abs{\Reg_{x,\gtil}^T(x^*) - \Reg_{\xhat,\gtil}^T(x^*)} \leq \Amax \hat{C}_x$,
  $\abs{\Reg_{x,g}^T(x^*) - \Reg_{x,\gtil}^T(x^*)} \leq 2 \tilde{C}_x$,
  and
  $\abs{\Reg_{\xhat,g}^T(x^*) - \Reg_{\xhat,\gtil}^T(x^*)} \leq 2 \tilde{C}_x$.
  The similar inequalities hold for 
  $
  \Reg_{y,\ell}^T(y^*),
  \Reg_{\yhat,\ell}^T(y^*),
  \Reg_{y,\ltil}^T(y^*),
  \Reg_{\yhat,\ltil}^T(y^*).
  $
\end{proposition}

\begin{proof}[Proof of \Cref{prop:reg_relation}]
  From the triangle inequality and the Cauchy--Schwarz inequality,
  we have 
  \begin{equation}    
    \abs*{\Reg_{x,g}^T(x^*) - \Reg_{\xhat,g}^T(x^*)}
    =
    \abs*{\sumT \inpr{\xhat^\t - x^\t, g^\t}}
    \leq 
    \sumT \nrm{\xhat^\t - x^\t}_1 \nrm{g^\t}_\infty 
    \leq 
    \Amax \hat{C}_x
    \per
    \n
  \end{equation}
  Similarly, we also have
  \begin{equation}
    \abs*{\Reg_{x,g}^T(x^*) - \Reg_{x,\gtil}^T(x^*)}
    =
    \abs*{\sumT \inpr{x^* - x^\t, g^\t - g^\t}}
    \leq 
    \sumT \nrm{x^* - x^\t}_1 \nrm{g^\t - \gtil^\t}_\infty 
    \leq 
    2 \tilde{C}_x
    \per
    \n
  \end{equation}
  The other inequalities can be proven in the same manner.
\end{proof}

\subsection{Extension of \Cref{thm:main} to corrupted regime}\label{app:proof_thm_main}
Here we provide the deferred arguments for the proof of \Cref{thm:main} and the extension of \Cref{thm:main} to the corrupted regime.

We use the algorithms for the corrupted regime, which simply replaces the non-corrupted (unobserved) gradients $g^\t$ and $\ell^\t$ in \cref{eq:OptHedge} with corrupted (observable) gradients $\gtil^\t$ and $\ltil^\t$.
In particular, the strategies of the $x$- and $y$-players, $\prn{\xhat^\t}_{t=1}^T$ and $\prn{\yhat^\t}_{t=1}^T$, are determined as follows:
\begin{equation}
  \begin{split}    
  \xhat^\t(i) 
  &
  \propto 
  \exp\prn*{ \eta_x^\t \prn*{ \sum_{s=1}^{t-1} \gtil^\s(i) + \gtil^\tm(i)} }
  \com
  \quad
  \eta_x^\t
  =
  \sqrt{\frac{M_x}{P_\infty^\t(\gtil) + P_\infty^\t(\ltil)}}
  \com
  \\
  \yhat^\t(i) 
  &
  \propto 
  \exp\prn*{ - \eta_y^\t \prn*{ \sum_{s=1}^{t-1} \ltil^\s(i) + \ltil^\tm(i)} }
  \com\quad
  \eta_y^\t
  =
  \sqrt{\frac{M_y}{P_\infty^\t(\gtil) + P_\infty^\t(\ltil)}}
  \com
  \end{split}
  \label{eq:OptHedge_corrupted}
\end{equation}
where $\eta_x^\t, \eta_y^\t > 0$ are the learning rates the $x$- and $y$-players at round $t$, and denote $\gtil^0 = \ltil^0 = g^0 = \ell^0 = \zeros$ for simplicity.
Note that $\xhat^1 = \frac{1}{\mx} \ones$ and $\yhat^1 = \frac{1}{\my} \ones$,
and recall that
$
\gtil^\t = g^\t + \tilde{c}_x^\t
$
for
$
  g^\t = A y^\t
$
and 
$
\tilde{\ell}^\t = \ell^\t + \tilde{c}_y^\t
$
for
$
\ell^\t = A^\top x^\t
$.

The learning dynamic specified by \cref{eq:OptHedge_corrupted} guarantees the following bounds:
\begin{theorem}[Fast scale-invariant convergence to Nash equilibrium in corrupted regime]\label{thm:main_corrupted}
Suppose that the $x$- and $y$-players use the algorithms in \cref{eq:OptHedge_corrupted} to obtain strategies $\prn{\xhat^\t}_{t=1}^T$ and $\prn{\yhat^\t}_{t=1}^T$, respectively.
Then, in the corrupted regime it holds that
\begin{align}
  \Reg_{x,\gtil}^T
  &
  =
  O\prn[\bigg]{
    \Adiff \sqrt{\log m \log \mx}
    \!+\!
    \sqrt{
      \brk*{
      \Ctilx \!+\! \Ctily
      \!+\!
      \Adiff^2 (\Chatx \!+\! \Chaty)
      \!+\!
      \Adiff (\Cx \!+\! \Cy)
      } \log \mx
    }
    \!+\!
    \Chatx
  }
  \com
  \nn
  \Reg_{y,\ltil}^T
  &
  =
  O\prn[\bigg]{
    \Adiff \sqrt{\log m \log \my}
    \!+\!
    \sqrt{
      \brk*{
      \Ctilx + \Ctily
      \!+\!
      \Adiff^2 (\Chatx \!+\! \Chaty)
      \!+\!
      \Adiff (\Cx \!+\! \Cy)
      } \log \my
    }
    \!+\!
    \Chaty
  }
  \com
  \n
\end{align}
and
\begin{align}
  \Reg_{x,g}^T
  &
  =
  O\prn[\bigg]{
    \Adiff \sqrt{\log m \log \mx}
    \!+\!
    \sqrt{
      \brk*{
      \Ctilx \!+\! \Ctily
      \!+\!
      \Adiff^2 (\Chatx \!+\! \Chaty)
      +
      \Adiff (\Cx \!+\! \Cy)
      } \log \mx
    }
    \!+\!
    \Cx
  }
  \com
  \nn
  \Reg_{y,\ell}^T
  &
  =
  O\prn[\bigg]{
    \Adiff \sqrt{\log m \log \my}
    \!+\!
    \sqrt{
      \brk*{
      \Ctilx \!+\! \Ctily
      \!+\!
      \Adiff^2 (\Chatx \!+\! \Chaty)
      \!+\!
      \Adiff (\Cx \!+\! \Cy)
      } \log \my
    }
    +
    \Cy
  }
  \per
  \n
\end{align}
\end{theorem}
Note that a linear dependence on a player's own corruption level is unavoidable; in fact, a matching lower bound is known \citep[Theorem 15 (ii)]{tsuchiya25corrupted}.
A key strength of the upper bounds in the corrupted regime is that the effect of the opponent's strategic deviation appears only through a square-root dependence, and moreover, if the regret is defined with respect to the corrupted utility gradients, then the impact of utility corruption is also only square-root.
These dependencies are also known to be optimal \citep[Theorem 15 (i), (iii)]{tsuchiya25corrupted}.

To prove \Cref{thm:main_corrupted}, we recall that from \Cref{lem:grad_diff}, we have
\begin{equation}\label{eq:gldiff2Adiff}
  \begin{split}    
  \nrm{g^\t - g^\tm}_\infty \leq (\Adiff / 2) \nrm{y^\t - y^\tm}_1
  \com
  \quad
  \nrm{\ell^\t - \ell^\tm}_\infty \leq (\Adiff / 2) \nrm{x^\t - x^\tm}_1
  \per
  \end{split}
\end{equation}
We then prepare the following lemma.
\begin{lemma}\label[lemma]{lem:Pinfty_gradtil_upper}
For any $\tau_x, \tau_y \in [T]$, it holds that
\begin{equation}
  Q_\infty^{\tau_y}(\gtil)
  \leq
  4 \Adiff^2 (Q_1^{\tau_y}(\yhat) + \Chaty + 1) + 8 \Ctilx
  \com\quad
  P_\infty^{\tau_x}(\ltil)
  \leq
  4 \Adiff^2 \prn{Q_1^{\tau_x}(\xhat) + \Chatx + 1} + 8 \Ctily
  \per
  \n
\end{equation}
\end{lemma}
\begin{proof}
We have
\begin{equation}
  P_\infty^{\tau_y}(\gtil)
  \leq
  \sum_{t=\tau_y}^T \nrm{\gtil^\t - \gtil^\tm}_\infty^2
  \leq
  2 \sum_{t=\tau_y}^T \nrm{g^\t - g^\tm}_\infty^2 + 8 \Ctilx
  \per
  \n
\end{equation}
The first term in the RHS in the last inequality can be evaluated as
\begin{align}
  2 \sum_{t=\tau_y}^T \nrm{g^\t - g^\tm} _\infty^2 
  &
  \leq 
  \frac{\Adiff^2}{2} \sum_{t=\tau_y}^T \nrm{y^\t - y^\tm}_1^2
  \tag{by \cref{eq:gldiff2Adiff}}
  \nn 
  &\leq 
  \frac{\Adiff^2}{2}
  \sum_{t=\tau_y}^T
  \prn*{
    2 \nrm{y^\t - \yhat^\t}_1^2 
    +
    4 \nrm{\yhat^\t - \yhat^\tm}_1^2 
    +
    2 \nrm{\yhat^\tm - y^\tm}_1^2 
  }
  \nn
  &\leq 
  2 \Adiff^2
  \prn*{
    \sum_{t=\tau_y}^T
    \nrm{y^\t - \yhat^\t}_1^2 
    +
    \sum_{t=\tau_y}^T
    \nrm{\yhat^\t - \yhat^\tm}_1^2 
    +
    1
  }
  \nn
  &\leq 
  2 \Adiff^2 \prn{\Chaty + Q_1^{\tau_y}(\yhat) + 1}
  \per
  \n
\end{align}
Combining the last two inequalities gives the desired bound on $P_\infty^{\tau_y}(\gtil)$.
The upper bound on $P_\infty^{\tau_x}(\ltil)$ can be proven by the same argument.
\end{proof}

Now we are ready to prove \Cref{thm:main_corrupted}.
\begin{proof}[Proof of \Cref{thm:main_corrupted}]  
When $\Adiff = 0$, both the $x$- and $y$-players have zero regret. Hence, we assume $\Adiff > 0$ in what follows.
Now we fix arbitrary $\tau_x, \tau_y \in [T]$ and let $\Abardiff = \Adiff /2$ for simplicity.
Then, from \Cref{lem:adahedge_opt_neg_mainbody},
\begin{align}
  \Reg_{\xhat,\gtil}^T
  &\leq
  \sqrt{32 (P_\infty^T(\gtil) + P_\infty^T(\ltil)) M_x}
  -
  \sumT \frac{1}{4 \eta_x^\t} \nrm{\xhat^\tp - \xhat^\t}_1^2
  \label{eq:reg_x_upper_1_mid}
  \\
  &=
  \sqrt{32 (P_\infty^{\tau_y}(\gtil) + P_\infty^{\tau_x}(\ltil) + Q_\infty^{\tau_y}(\gtil) + Q_\infty^{\tau_x}(\ltil)) M_x}
  -
  \sumT \frac{1}{4 \eta_x^\t} \nrm{\xhat^\tp - \xhat^\t}_1^2
  \nn
  &\leq
  \sqrt{ 32
    \brk*{ 
      P_\infty^{\tau_y}(\gtil)
      + 
      P_\infty^{\tau_x}(\ltil)
      +
      4 \Adiff^2 (Q_1^{\tau_x}(\xhat) + Q_1^{\tau_y}(\yhat) + \Chatx + \Chaty + 2) + 8 \Ctilx + 8 \Ctily
    } 
    M_x}
  \nn
  &\qquad\qquad-
  \frac{1}{4 \eta_x^{\tau_x - 1}} Q_1^{\tau_x}(\xhat)
  \com
  \label{eq:reg_x_upper_1}
\end{align}
where the last inequality follows from \Cref{lem:Pinfty_gradtil_upper}, \cref{eq:gldiff2Adiff}, and the inequality 
\begin{align}
  \sumT \frac{1}{4 \eta_x^\t} \nrm{\xhat^\tp - \xhat^\t}_1^2
  &= 
  \sum_{t=1}^{T+1} \frac{1}{4 \eta_x^\tm} \nrm{\xhat^\t - \xhat^\tm}_1^2
  -
  \frac{1}{4 \eta_x^1} \nrm{\xhat^1 - \xhat^0}_1^2
  \nn
  &\geq 
  \frac{1}{4 \eta_x^{\tau_x - 1}} \sum_{t=\tau_x}^T \nrm{\xhat^\t - \xhat^\tm}_1^2
  =
  \frac{1}{4 \eta_x^{\tau_x - 1}} Q_1^{\tau_x}(\xhat)
  \n
\end{align}
since $(\eta_x^\t)_t$ is nonincreasing and $\eta_x^1 = \infty$.
Similarly,
\begin{align}
  \Reg_{\yhat,\ltil}^T
  &
  \leq
  \sqrt{32  (P_\infty^T(\gtil) + P_\infty^T(\ltil)) M_y}
  -
  \sumT \frac{1}{4 \eta_y^\t} \nrm{\yhat^\tp - \yhat^\t}_1^2
  \label{eq:reg_y_upper_1_mid}
  \\
  &
  \leq
  \sqrt{ 32
    \brk*{ 
      P_\infty^{\tau_y}(\gtil)
      + 
      P_\infty^{\tau_x}(\ltil)
      +
      4 \Adiff^2 (Q_1^{\tau_x}(\xhat) + Q_1^{\tau_y}(\yhat) + \Chatx + \Chaty + 2) + 8 \Ctilx + 8 \Ctily
    } 
    M_y}
    \nn
  &\qquad\qquad-
  \frac{1}{4 \eta_y^{\tau_y - 1}} Q_1^{\tau_y}(\yhat)
  \per
  \label{eq:reg_y_upper_1}
\end{align}

In what follows, we provide the regret upper bounds on $\Reg_{x,g}^T$ and $\Reg_{y,\ell}^T$, and the regret bound for $\Reg_{x,\gtil}^T$ and $\Reg_{y,\ltil}^T$ can be obtain in a similar manner.
From \cref{eq:reg_x_upper_1,eq:reg_y_upper_1} and \Cref{prop:reg_relation},
we have
\begin{align}
  \Reg_{x,g}^T
  &\leq
  \sqrt{ 32
    \prn{ 
      P_\infty^{\tau_y}(\gtil)
      + 
      P_\infty^{\tau_x}(\ltil)
    } 
    M_x
  }
  +
  \sqrt{ 128
    \Adiff^2 (Q_1^{\tau_x}(\xhat) + Q_1^{\tau_y}(\yhat)) M_x
  }
  \nn
  &\qquad
  +
  \sqrt{ 128
    \brk{ 
      4 \Adiff^2 (\Chatx + \Chaty + 2) + 8 \Ctilx + 8 \Ctily
    } 
    M_x
  }
  +
  \Cx
  -
  \frac{1}{4 \eta_x^{\tau_x - 1}} Q_1^{\tau_x}(\xhat)
  \com
  \nn
  \Reg_{y,\ell}^T
  &\leq
  \sqrt{ 32
    \prn{ 
      P_\infty^{\tau_y}(\gtil)
      + 
      P_\infty^{\tau_x}(\ltil)
    } 
    M_y
  }
  +
  \sqrt{ 128
    \Adiff^2 (Q_1^{\tau_x}(\xhat) + Q_1^{\tau_y}(\yhat)) M_y
  }
  \nn
  &\qquad
  +
  \sqrt{ 128
    \brk{ 
      4 \Adiff^2 (\Chatx + \Chaty + 2) + 8 \Ctilx + 8 \Ctily
    } 
    M_y
  }
  +
  \Cy
  -
  \frac{1}{4 \eta_y^{\tau_y - 1}} Q_1^{\tau_y}(\yhat)
  \com
  \n
\end{align}
where we used the subadditivity of $z \mapsto \sqrt{z}$ for $z \geq 0$.

Define
\begin{equation}\label{eq:def_kappa}
  \kappa
  =
  \sqrt{ 128
    \brk{ 
      4 \Adiff^2 (\Chatx + \Chaty + 2) + 8 \Ctilx + 8 \Ctily
    } 
    M
  }
  +
  \Cx + \Cy
  \per
\end{equation}
Then, by combining the upper bounds on $\Reg_{x,g}^T$ and $\Reg_{y,\ell}^T$, we obtain
\begin{align}
  &
  \Reg_{x,g}^T + \Reg_{y,\ell}^T
  \nn
  &\leq
  \sqrt{ 128
    \prn{ 
      P_\infty^{\tau_y}(\gtil)
      + 
      P_\infty^{\tau_x}(\ltil)
    } 
    M
  }
  \nn
  &\qquad
  +
  \prn*{
  \sqrt{ 128 \Adiff^2 Q_1^{\tau_x}(\xhat) M }
  -
  \frac{1}{8 \eta_x^{\tau_x - 1}} Q_1^{\tau_x}(\xhat)
  }
  +
  \prn*{
  \sqrt{ 128 \Adiff^2 Q_1^{\tau_y}(\yhat) M }
  -
  \frac{1}{8 \eta_y^{\tau_y - 1}} Q_1^{\tau_y}(\yhat)
  }
  \nn
  &\qquad
  -
  \prn*{
    \frac{1}{8 \eta_x^{\tau_x - 1}} Q_1^{\tau_x}(\xhat)
    +
    \frac{1}{8 \eta_y^{\tau_y - 1}} Q_1^{\tau_y}(\yhat)
  }
  +
  \kappa
  \label{eq:social_reg_upper_1_mid}
  \\
  &\leq
  \sqrt{ 128
    \prn{ 
      P_\infty^{\tau_y}(\gtil)
      \!+\! 
      P_\infty^{\tau_x}(\ltil)
    } 
    M
  }
  \!+\!
  256 \Adiff^2 M (\eta_x^{\tau_x - 1} \!+\! \eta_y^{\tau_y - 1})
  \!-\!
  \prn[\bigg]{
    \frac{1}{8 \eta_x^{\tau_x - 1}} Q_1^{\tau_x}(\xhat)
    \!+\!
    \frac{1}{8 \eta_y^{\tau_y - 1}} Q_1^{\tau_y}(\yhat)
  }
  \!+\!
  \kappa
  \nn
  &\leq
  2^{7/2}
  \prn*{
    \sqrt{P_\infty^{\tau_x}(\ltil) M}
    +
    \frac{2^{9/2} \Adiff^2 M^{3/2}}{\sqrt{P_\infty^{\tau_x-1}(\ltil)}}
  }
  +
  2^{7/2}
  \prn*{
    \sqrt{P_\infty^{\tau_y}(\gtil) M}
    +
    \frac{2^{9/2} \Adiff^2 M^{3/2}}{\sqrt{P_\infty^{\tau_y-1}(\gtil)}}
  }
  \nn
  &\qquad\qquad-
  \prn*{
    \frac{1}{8 \eta_x^{\tau_x - 1}} Q_1^{\tau_x}(\xhat)
    +
    \frac{1}{8 \eta_y^{\tau_y - 1}} Q_1^{\tau_y}(\yhat)
  }
  +
  \kappa
  \com
  \label{eq:social_reg_upper_1}
\end{align}
where in the second inequality we considered the worst cases with respect to~$Q_1^{\tau_y}(\yhat)$ and $Q_1^{\tau_x}(\xhat)$ by the inequality $b\sqrt{z} - az \leq b^2 / (4 a)$ that holds for $a > 0$, $b \geq 0$, and $z \geq 0$,
and in the last inequality we used the subadditivity of $z \mapsto \sqrt{z}$ for $z \geq 0$ and $\eta_x^{\tau_x} \leq \sqrt{M_x / P_\infty^{\tau_x}(\ltil)}$ and $\eta_y^{\tau_y} \leq \sqrt{M_y / P_\infty^{\tau_y}(\gtil)}$.

To evaluate the first two terms in the last inequality,
we choose $\tau_x, \tau_y \in [T]$ by
\begin{equation}\label{eq:tau_main_corrupt}
  \begin{split}    
  \tau_x &= \min\set{\tau^\circ(\ltil), T}
  \com\quad
  \tau^\circ(\ltil)
  \coloneqq
  \tau^\circ(\ltil; \sqrt{M}, 2^{9/2} \Adiff^2 M^{3/2})
  \com
  \nn
  \tau_y &= \min\set{\tau^\circ(\gtil), T}
  \com\quad
  \tau^\circ(\gtil)
  \coloneqq
  \tau^\circ(\gtil; \sqrt{M}, 2^{9/2} \Adiff^2 M^{3/2})
  \com
  \end{split}
\end{equation}
where $\tau^\circ(\cdot; \cdot, \cdot)$ is defined in \cref{eq:tau_circ}.

We next list the inequality when $\tau^\circ(\ltil) \leq T$ holds.
Now from \Cref{lem:Pinfty_tradeoff},
the first two terms in~\cref{eq:social_reg_upper_1} is upper bounded as
\begin{equation} 
  \sqrt{P_\infty^{\tau_x}(\ltil) M}
  +
  \frac{2^{9/2} \Adiff^2 M^{3/2}}{\sqrt{P_\infty^{\tau_x-1}(\ltil)}}
  \lesssim
  \Adiff M
  \com
  \label{eq:Px_inf}
\end{equation}
where we used $M \geq 1$.
From \Cref{lem:Pinfty_tradeoff} we also have $P_\infty^{\tau_x - 1}(\ell) \geq (2^{9/2} \Adiff^2 M^{3/2}) / \sqrt{M} = 2^{9/2} \Adiff^2 M$ and thus
\begin{equation}\label{eq:eta_x_tm_lower}
  \eta_x^{\tau_x - 1}
  \leq
  \sqrt{\frac{M_x}{P_\infty^{\tau_x - 1}(\ell)}}
  \leq
  \sqrt{\frac{M_x}{2^{9/2} \Adiff^2 M}}
  \leq
  \frac{1}{4.75 \Adiff}
  \per
\end{equation}

Similarly, if $\tau^\circ(\gtil) \leq T$, the two terms in \cref{eq:social_reg_upper_1} is upper bounded as
\begin{equation}
  \sqrt{P_\infty^{\tau_y}(\gtil) M}
  +
  \frac{2^{9/2} \Adiff^2 M^{3/2}}{\sqrt{P_\infty^{\tau_y-1}(\gtil)}}
  \lesssim
  \Adiff M
  \per
  \label{eq:Py_inf}
\end{equation}
and 
\begin{equation}\label{eq:eta_y_tm_lower}
  \eta_y^{\tau_y - 1}
  \leq
  \frac{1}{4.75 \Adiff}
  \per
\end{equation}

\paragraph{Case 1: when $\tau^\circ(\ltil) \leq T$ and $\tau^\circ(\gtil) \leq T$}
Combining \cref{eq:social_reg_upper_1} with \cref{eq:Px_inf,eq:Py_inf}, we have
\begin{equation}
  \Reg_{x,g}^T + \Reg_{y,\ell}^T
  \leq
  O(\Adiff M)
  +
  \kappa
  -
  \prn*{
  \frac{1}{8 \eta_x^{\tau_x - 1}} Q_1^{\tau_x}(\xhat)
  +
  \frac{1}{8 \eta_y^{\tau_y - 1}} Q_1^{\tau_y}(\yhat)
  }
  \per
  \label{eq:social_reg_upper_2}
\end{equation}
Therefore, combining \cref{eq:social_reg_upper_2} with the fact that $\Reg_{x,g}^T + \Reg_{y,\ell}^T \geq 0$, we have
\begin{align}
  Q_1^{\tau_x}(\xhat)
  &\lesssim
  \eta_x^{\tau_x - 1} \prn{\Adiff M + \kappa}
  \nn
  &\lesssim
  M
  +
  \frac{1}{\Adiff}
  \prn*{
  \sqrt{
    \brk{ 
      \Adiff^2 (\Chatx + \Chaty + 1) + \Ctilx + \Ctily
    } 
    M
    }
    +
    \Cx + \Cy
  }
  \nn
  &\lesssim
  M
  +
  (\Chatx + \Chaty + 1)
  +
  \frac{1}{\Adiff}
  \prn*{
  \sqrt{
    \prn{ 
      \Ctilx + \Ctily
    } 
    M
    }
    +
    \Cx + \Cy
  }
  \com
  \label{eq:path_x_upper_1}
\end{align}
where the second inequality follows from \cref{eq:eta_x_tm_lower} and the definition of $\kappa$.
Similarly using \cref{eq:eta_y_tm_lower}, we can obtain 
\begin{equation}
  Q_1^{\tau_y}(\yhat)
  \lesssim
  M
  +
  (\Chatx + \Chaty + 1)
  +
  \frac{1}{\Adiff}
  \prn*{
  \sqrt{
    \prn{ 
      \Ctilx + \Ctily
    } 
    M
    }
    +
    \Cx + \Cy
  }
  \per
  \n
\end{equation}
From the definitions of $\tau_x, \tau_y$ and \Cref{lem:Pinfty_tradeoff}, we also have
\begin{equation}
  P_\infty^{\tau_x}(\ltil)
  \leq
  2^{9/2} \Adiff^2 M + 2 \Adiff^2
  \lesssim
  \Adiff^2 M
  \com
  \quad
  P_\infty^{\tau_y}(\gtil)
  \leq 
  2 \Adiff^2 (256 M + 1)
  \lesssim
  \Adiff^2 M
  \per
  \n
\end{equation}

Finally, by plugging the above upper bounds on $P_\infty^{\tau_x}(\ltil)$, $P_\infty^{\tau_y}(\gtil)$, $Q_1^{\tau_x}(\xhat)$, and $Q_1^{\tau_y}(\yhat)$ in \cref{eq:reg_x_upper_1}, we obtain
\begin{align}
  &
  \Reg_{\xhat,\gtil}^T
  \nn
  &\lesssim
  \sqrt{
    \brk[\bigg]{
      \Adiff^2 M
      \!+\!
      \Adiff^2
      \brk[\Big]{
        M
        \!+\!
        (\Chatx \!+\! \Chaty \!+\! 1)
        \!+\!
        \frac{1}{\Adiff}
        \prn[\Big]{ \!\!
        \sqrt{ 
          \prn{ 
            \Ctilx \!+\! \Ctily
          } 
          M
          }
          \!+\!
          \Cx \!+\! \Cy
        }
        \!+\!
        (\Chatx \!+\! \Chaty \!+\! 1)
      }
      \!+\!
      \Ctilx \!+\! \Ctily
    } 
    M_x
  }
  \nn
  &\leq
  \sqrt{
    \brk*{
      \Adiff^2 M
      +
      \Adiff^2
      +
      \Adiff 
      \prn[\bigg]{
        \frac{ 
          \Ctilx + \Ctily
        }{\Adiff}
        +
        \Adiff M
        +
        \Cx + \Cy
      }
      +
      \Adiff^2 (\Chatx + \Chaty) + \Ctilx + \Ctily
    } 
    M_x
  }
  \tag{AM--GM}
  \nn
  &\lesssim 
  \Adiff \sqrt{M M_x}
  +
  \sqrt{
    \brk*{
    (\Ctilx + \Ctily)
    +
    \Adiff^2 (\Chatx + \Chaty)
    +
    \Adiff (\Cx + \Cy)
    }
    M_x
  }
  \per
  \n
\end{align}
Similarly, plugging the above upper bounds on $P_\infty^{\tau_x}(\ltil)$, $P_\infty^{\tau_y}(\gtil)$, $Q_1^{\tau_x}(\xhat)$, and $Q_1^{\tau_y}(\yhat)$ in \cref{eq:reg_y_upper_1}, we obtain
\begin{equation}
  \Reg_{\yhat,\ltil}^T
  \lesssim
  \Adiff \sqrt{M M_y}
  +
  \sqrt{
    \brk*{
    (\Ctilx + \Ctily)
    +
    \Adiff^2 (\Chatx + \Chaty)
    +
    \Adiff (\Cx + \Cy)
    }
    M_y
  }
  \per
  \n
\end{equation}
Combining these upper bounds on $\Reg_{\xhat,\gtil}^T$ and $\Reg_{\yhat,\ltil}^T$ with \Cref{prop:reg_relation} yields the desired upper bounds in \Cref{thm:main} for the corrupted regime.

\paragraph{Case 2: when $\tau^\circ(\ltil) = \infty$ and $\tau^\circ(\gtil)\leq T$}
From the choice of $\tau^\circ(\ltil)$ in \cref{eq:tau_main_corrupt}, we have $P_\infty^T(\ltil) \leq 2^{9/2} \Adiff^2 M$, and thus 
setting $Q_1^{\tau_x}(\xhat) = 0$ in \cref{eq:social_reg_upper_1_mid} we have
\begin{align}
  &
  \Reg_{x,g}^T + \Reg_{y,\ell}^T
  \nn
  &\leq
  \sqrt{ 128
    \prn{ 
      P_\infty^{\tau_y}(\gtil)
      + 
      P_\infty^{T}(\ltil)
    } 
    M
  }
  +
  \prn*{
  \sqrt{ 128 \Adiff^2 Q_1^{\tau_y}(\yhat) M }
  -
  \frac{1}{8 \eta_y^{\tau_y - 1}} Q_1^{\tau_y}(\yhat)
  }
  -
  \frac{1}{8 \eta_y^{\tau_y - 1}} Q_1^{\tau_y}(\yhat)
  +
  \kappa
  \nn
  &\leq
  \sqrt{ 128
    \prn{ 
      P_\infty^{\tau_y}(\gtil)
      + 
      P_\infty^{T}(\ltil)
    } 
    M
  }
  +
  2^{7/2}
  \prn*{
    \sqrt{P_\infty^{\tau_y}(\gtil) M}
    +
    \frac{2^{9/2} \Adiff^2 M^{3/2}}{\sqrt{P_\infty^{\tau_y-1}(\gtil)}}
  }
  -
  \frac{1}{8 \eta_y^{\tau_y - 1}} Q_1^{\tau_y}(\yhat)
  +
  \kappa
  \nn
  &\leq
  O\prn*{
  \sqrt{
    \prn{ 
      P_\infty^{\tau_y}(\gtil)
      + 
      P_\infty^{T}(\ltil)
    } 
    M
  }
  +
  \Adiff M
  }
  -
  \frac{1}{8 \eta_y^{\tau_y - 1}} Q_1^{\tau_y}(\yhat)
  +
  \kappa
  \com
  \n
\end{align}
where 
the second inequality follows from the same argument as in \cref{eq:social_reg_upper_1}
and
the last inequality follows from the property of $\tau^\circ(\gtil) \leq T$ in \cref{eq:Py_inf}.
From this inequality, using the argument as in \cref{eq:path_x_upper_1}, we can show
\begin{equation}
  Q_1^{\tau_y}(\yhat)
  \lesssim
  M
  +
  (\Chatx + \Chaty + 2)
  +
  \frac{1}{\Adiff}
  \prn*{
  \sqrt{
    \prn{ 
      \Ctilx + \Ctily
    } 
    M
    }
    +
    \Cx + \Cy
  }
  \per
  \n
\end{equation}
Finally, by plugging the upper bound on $P_\infty^T(\ltil)$ and the inequalities $P_\infty^{\tau_y}(\gtil) \leq 2^{9/2} \Adiff^2 M$ and $Q_1^{\tau_y}(\yhat)$ in \cref{eq:reg_x_upper_1,eq:reg_y_upper_1} and using \Cref{prop:reg_relation}, we obtain the desired bounds.

\paragraph{Case 3: when $\tau^\circ(\ltil) \leq T$ and $\tau^\circ(\gtil)= \infty$}
Repeating the same analysis as in Case (2) gives the desired bounds.

\paragraph{Case 4: when $\tau^\circ(\ltil) = \infty$ and $\tau^\circ(\gtil)= \infty$}
From the choices of $\tau^\circ(\ltil), \tau^\circ(\gtil)$ in \cref{eq:tau_main_corrupt}, we have $P_\infty^T(\ltil) \leq 2^{9/2} \Adiff^2 M$ and $P_\infty^T(\gtil) \leq 2^{9/2} \Adiff^2 M$.
Combining these bounds with \cref{eq:reg_x_upper_1_mid,eq:reg_y_upper_1_mid} and using \Cref{prop:reg_relation}, we obtain the desired bounds.
\end{proof}

\subsection{Scale-free learning dynamic when no communication is possible}\label{app:nocom}
All of the learning dynamics presented in the main body achieve scale-free and scale-invariant fast convergence by tuning their learning rates using the gradients observed by the opponents.
This is not an issue if one views learning in games as an equilibrium computation problem.
However, in more game-theoretic settings, it is natural to ask whether one can design learning dynamics that guarantee scale-free or scale-invariant fast convergence without requiring such gradient communication.
Here we show how to construct a scale-free (though not scale-invariant) learning dynamic with fast convergence of $O(1/T)$ without such communication, at the cost of a worse dependence on the scale parameter $\Amax$ when $\Amax$ is small.
For simplicity, we focus on the honest regime here.

We consider the optimistic Hedge algorithm in \cref{eq:OptHedge} but we use the learning rates $\eta_x^\t, \eta_y^\t > 0$ given by
\begin{equation}\label{eq:lr_nocom}
  \eta_x^\t
  =
  \sqrt{\frac{M_x}{P_\infty^\t(\gtil) + P_1^\t(\xhat)}}
  \com
  \quad
  \eta_y^\t
  =
  \sqrt{\frac{M_y}{P_\infty^\t(\ltil) + P_1^\t(\yhat)}}
  \per
\end{equation}
Note that the learning dynamics defined by \cref{eq:lr_nocom} is scale-free but not scale-invariant (recall \Cref{def:scale_games}).

We prepare the following lemma, which is a minor variant of \Cref{lem:Pinfty_tradeoff}.
\begin{lemma}\label[lemma]{lem:Pinfty_tradeoff_z}
Let $z_1, \dots, z_T \in \Delta_d$
and $c_1, c_2 > 0$.
Denote $\hat{\tau}(z; c_1, c_2) \in [T] \cup \set{\infty}$ by
\begin{equation}\label{eq:tau_hat}
  \hat{\tau}(z; c_1, c_2)
  =
  \begin{cases}
    \min\set[\big]{
      t 
      \colon 
      P_1^t(z) 
      >
      {c_2}/{c_1} + 4
    } & 
    \mbox{if} \ P_1^T(z) > {c_2}/{c_1} + 4
    \com
    \\
    \infty & \text{otherwise}
    \com
  \end{cases}
\end{equation}
where $P_1^t(z) = \sum_{s=1}^{t-1} \nrm{z^\s - z^\sm}_1^2$.
Then, $\hat{\tau} \geq 2$, and if $\hat{\tau}(z; c_1, c_2) \leq T$, it holds that
\begin{equation}
  P_1^{\hat{\tau}}(z)
  \leq
  \frac{c_2}{c_1}
  +
  8
  \com
  \
  P_1^{\hat{\tau} - 1}(z) \geq \frac{c_2}{c_1}
  \com
  \
    c_1 \sqrt{P_1^{\hat{\tau}}(z)}
    +
    \frac{c_2}{ \tsqrt{P_1^{\hat{\tau}-1}(z)} }
  \leq
  \sqrt{8} c_1 + 2 \sqrt{c_1 c_2}
  \per
  \n
\end{equation}
\end{lemma}

\begin{proof}
For simplicity, we write $\hat{\tau} = \hat{\tau}(z; c_1, c_2)$, omitting the dependence on $z, c_1, c_2$ in this proof.
Since $\hat{\tau} \leq T$, the definition of $\hat{\tau}$ implies
$
  P_1^{\hat{\tau} - 1}(z) 
  \leq 
  c_2 / c_1 + 4
  ,
$
and thus 
\begin{equation}
  P_1^{\hat{\tau}}(z)
  =
  P_1^{\hat{\tau} - 1}(z)
  +
  \nrm{z^{\hat{\tau} - 1} - z^{\hat{\tau} - 2}}_\infty^2
  \leq
  \frac{c_2}{c_1}
  +
  8
  \per
  \n
\end{equation}
We also have
\begin{equation}
  P_1^{\hat{\tau} - 1}(z)
  =
  P_1^{\hat{\tau}}(z)
  -
  \nrm{z^{\hat{\tau} - 1} - z^{\hat{\tau} - 2}}_\infty^2
  \geq
  \frac{c_2}{c_1}
  \per
  \n
\end{equation}
Combining the last two inequalities, we obtain
\begin{equation}
  c_1 \sqrt{P_1^{\hat{\tau}}(z)}
  +
  \frac{c_2}{ \sqrt{P_1^{\hat{\tau}-1}(z)} }
  \leq
  c_1 
  \sqrt{
    \frac{c_2}{c_1}
    +
    8
  }
  +
  \frac{c_2}{\sqrt{c_2 / c_1}}
  \leq
  \sqrt{8} c_1 + 2 \sqrt{c_1 c_2}
  \com
  \n
\end{equation}
which completes the proof.
\end{proof}

\begin{theorem}\label{thm:main_nocom}
  In the honest regime, the social regret of the learning dynamics defined by \cref{eq:lr_nocom} is upper bounded as
  \begin{equation}
    \Reg_x^T + \Reg_y^T
    \lesssim
    (1 + \sqrt{\Adiff}) \sqrt{\Adiff} \log m
    \lesssim
    \begin{cases}
      \Adiff \log m & \mbox{if} \ \Adiff \geq 1 \com \\
      \sqrt{\Adiff} \log m & \mbox{if} \ \Adiff < 1  \per
    \end{cases}
    \n
  \end{equation}
  Consequently (by \Cref{thm:reg2nasheq}), the average play $\prn[\big]{ \frac1T \sumT x^\t, \frac1T \sumT y^\t}$ is an $O( (1 + \sqrt{\Adiff}) \sqrt{\Adiff} \log m / T)$-approximate Nash equilibrium.
\end{theorem}
When the scale parameter $\Amax$ is large, this regret bound yields the same regret guarantee and convergence rate as \Cref{thm:main}.
On the other hand, in the regime where $\Amax$ is small and approaches $0$, the dependence on the scale becomes worse by a factor of $1/\sqrt{\Amax}$.
As mentioned in \Cref{remark:zhang25scale}, this result highlights the difficulty of adapting to small values of $\Amax$ in scale-free online learning and scale-free learning dynamics.
Improving the convergence rate under no communication is an interesting direction for future work.

\begin{proof}
From \Cref{lem:adahedge_opt_neg_mainbody}, we have
\begin{align}
  \Reg_x^T
  &\leq
  \sqrt{ 32 (P_\infty^T(g) + P_\infty^T(x)) M_x}
  -
  \sumT \frac{1}{4 \eta_x^\t} \nrm{x^\tp - x^\t}_1^2
  \nn
  &\leq
  \sqrt{ 32
    \brk[\big]{ 
      P_\infty^{\tau_y}(g) \!+\! \Abardiff^2 Q_1^{\tau_y}(y) 
      \!+
      P_1^{\tau_x}(x) \!+\! Q_1^{\tau_x}(x) 
    } M_x}
  -
  \frac{1}{4 \eta_x^{\tau_x-1}} Q_1^{\tau_x}(x)
  \com
  \label{eq:reg_x_upper_1_main_nocom}
\end{align}
where the last inequality follows since $(\eta_x^\t)_t$ is nonincreasing and \cref{eq:gldiff2Adiff_main}.
Similarly,
\begin{equation}
  \Reg_y^T
  \leq
  \sqrt{ 32
    \brk[\big]{ 
      P_\infty^{\tau_x}(\ell) \!+\! \Abardiff^2 Q_1^{\tau_x}(x) 
      \!+\!
      P_1^{\tau_y}(y) \!+\! Q_1^{\tau_y}(y) 
    } M_y}
  -
  \frac{1}{4 \eta_y^{\tau_y-1}} Q_1^{\tau_y}(y)
  \per
  \label{eq:reg_y_upper_1_main_nocom}
\end{equation}
Let $B = \sqrt{1 + \Abardiff^2}$ for simplicity.
Then, by combining the upper bounds on $\Reg_x^T$ in~\cref{eq:reg_x_upper_1_main_nocom} and $\Reg_y^T$ in~\cref{eq:reg_y_upper_1_main_nocom}, we have
\begin{align}
  &
  \Reg_x^T + \Reg_y^T
  \nn
  &\lesssim
  \sqrt{
  \prn*{ 
    P_\infty^{\tau_y}(g) 
    + 
    P_\infty^{\tau_x}(\ell)
  } M}
  +
  \sqrt{ B^2 Q_1^{\tau_y}(y) M }
  -
  \frac{1}{4 \eta_y^{\tau_y-1}} Q_1^{\tau_y}(y)
  +
  \sqrt{ B^2 Q_1^{\tau_x}(x) M }
  -
  \frac{1}{4 \eta_x^{\tau_x-1}} Q_1^{\tau_x}(x)
  \nn
  &\lesssim
  \sqrt{
  \prn*{ 
    P_\infty^{\tau_y}(g) 
    + 
    P_\infty^{\tau_x}(\ell)
  } M}
  +
  B^2 M (\eta_x^{\tau_x-1} + \eta_y^{\tau_y-1})
  \nn
  &\lesssim
  \prn[\bigg]{
  \Abardiff \sqrt{P_1^{\tau_x}(x) M}
  +
  \frac{B^2 M^{3/2}}{\sqrt{P_1^{\tau_x-1}(x)}}
  }
  +
  \prn[\bigg]{
  \Abardiff \sqrt{P_1^{\tau_y}(y) M}
  +
  \frac{B^2 M^{3/2}}{\sqrt{P_1^{\tau_y-1}(y)}}
  }
  \com
  \label{eq:social_reg_upper_1_main_nocom}
\end{align}
where in the second inequality we considered the worst cases with respect to~$Q_1^{\tau_y}(y)$ and $Q_1^{\tau_x}(x)$ by the inequality $b\sqrt{z} - az \leq b^2 / (4 a)$ that holds for $a > 0$, $b \geq 0$, and $z \geq 0$,
and in the last inequality we used the subadditivity of $\sqrt{\cdot}$, $\Abardiff = \Adiff /2$, and the definitions of the learning rates in \cref{eq:lr_nocom}: $\eta_x^{\tau_x-1} \leq \sqrt{M_x / P_1^{\tau_x-1}(x)}$ and $\eta_y^{\tau_y-1} \leq \sqrt{M_y / P_1^{\tau_y-1}(y)}$.

Now using $\hat{\tau}$ defined in \cref{eq:tau_hat},
we choose the stopping time $\tau_x, \tau_y \in [T]$ by
\begin{align}
  \begin{split}
    \tau_x
    &=
    \min\set{\hat{\tau}(\ell), T}
    \com\
    \hat{\tau}(\ell)
    \coloneqq
    \hat{\tau}(\ell; \Abardiff \sqrt{M}, B^2 M^{3/2})
    \com
    \nn
    \tau_y
    &=
    \min\set{\hat{\tau}(g), T}
    \com\
    \hat{\tau}(g)
    \coloneqq
    \hat{\tau}(g; \Abardiff \sqrt{M}, B^2 M^{3/2})    
    \per
  \end{split}
\end{align}

In what follows, we focus on the case when $\hat{\tau}(\ell) \leq T$ and $\hat{\tau}(g) \leq T$.
Otherwise, the regret can be bounded by the similar argument; see \Cref{app:proof_thm_main} for details.

From \Cref{lem:Pinfty_tradeoff_z}, the first two terms in \cref{eq:social_reg_upper_1_main_nocom} is evaluated as 
\begin{align}
  \Abardiff \sqrt{P_1^{\tau_x}(x) M}
  +
  \frac{B^2 M^{3/2}}{\sqrt{P_1^{\tau_x-1}(x)}}
  &
  \lesssim
  \Adiff \sqrt{M} + \sqrt{\Adiff \sqrt{M} B^2 M^{3/2}}
  \nn
  &
  \lesssim
  \Adiff \sqrt{M} + \sqrt{\Adiff (1 + \Adiff)} M
  \nn
  &
  \lesssim
  \sqrt{\Adiff} (1 + \sqrt{\Adiff}) M
  \com
  \label{eq:Px_inf_main_nocom}
\end{align}
where we used $M \geq 1$.
Similarly, when $\tau_y \leq T$, the terms in \cref{eq:social_reg_upper_1_main_nocom} is upper bounded as
\begin{equation}
  \Abardiff \sqrt{P_1^{\tau_y}(y) M}
  +
  \frac{B^2 M^{3/2}}{\sqrt{P_1^{\tau_y-1}(y)}}
  \lesssim
  \sqrt{\Adiff} (1 + \sqrt{\Adiff}) M
  \per
  \label{eq:Py_inf_main_nocom}
\end{equation}
Combining \cref{eq:social_reg_upper_1_main_nocom} with~\cref{eq:Px_inf_main_nocom,eq:Py_inf_main_nocom}, we have
\begin{align}
  &
  \Reg_x^T + \Reg_y^T
  \lesssim
  \sqrt{\Adiff} (1 + \sqrt{\Adiff}) M
  \com
  \n
\end{align}
which is the desired bound.
\end{proof}

\section{Deferred Proofs for Multiplayer General-Sum Games from \Cref{sec:multiplayer}}\label{app:proof_multiplayer}
This section provides the proofs deferred from \Cref{sec:multiplayer}.
We first recall some notation.
The set $\calJ' = \set{ t \in [T] : \max_{i \in [n]} \nrm{u_i^\t}_\infty \geq 2 B^\t}$ is the set of rounds at which a jump occurs in \Cref{line:doubling_clipping} of \Cref{alg:multiple_player_swap}, and $\calI' = [T] \setminus \calJ'$ is the set of rounds with no jump.
For convenience, recall that we defined $\calJ = \calJ' \cup \set{t + 1 \in [T] : t \in \calJ'}$ and $\calI = [T] \setminus \calJ$.
Note that, from the definition of the clipped gradients in~\cref{eq:clipped_grad}, we have $B^\tp = B^\t$ and hence $\ubar_{i,a}^\t = u_{i,a}^\t$ when $t \in \calI$, whereas $B^\tp \neq B^\t$ and hence $\ubar_{i,a}^\t \neq u_{i,a}^\t$ when $t \in \calJ$.
We recall that 
$\nrm{h}_{x,f} = \sqrt{h^\top \nabla^2 f(x) h}$
and
$\nrm{h}_{*,x,f} = \sqrt{h^\top \prn{\nabla^2 f(x)}^{-1} h}$
are the local norm and its dual norm of a vector $h$ at a point $x$ with respect to a convex function $f$, respectively.
We also recall that $\phi(x) = - \sum_k \log(x(k))$ is the logarithmic barrier function.

\subsection{Stability analysis of stationary distributions}\label{subsec:stab_analysis}
Here we provide a stability analysis of stationary distributions of the Markov chain induced by $Q^\t$.
We begin with the following lemma, which upper bounds the increase of the reciprocal of the learning rate (a variant of \citealt[Lemma 29]{tsuchiya25corrupted}).
\begin{lemma}\label[lemma]{lem:lr_diff}
  Suppose that $t \in \calI' \supseteq \calI$.
  Then, the learning rate ${\eta_{i}^\t}$ in \cref{eq:oftrl_multiplayer}
  satisfies
  \begin{equation}
    \frac{1}{\eta_{i}^\tp} - \frac{1}{\eta_{i}^\t}
    \leq 
    \frac{4 \sqrt{\gamma} U^\t}{\alpha_i}
    \per 
    \n
  \end{equation}
\end{lemma}
\begin{proof}  
From the definition of the learning rate $\eta_{i}^\t$, we have
\begin{align}
  \frac{1}{\eta_{i}^\tp} \!-\! \frac{1}{\eta_{i}^\t}
  &=
  \frac{1}{
  \min\set*{
    \frac{\alpha_i}{ \sqrt{ \gamma (U^\t)^2 + P_\infty^\tp(\ubar_{i}) } }
    ,
    \frac{\beta_i}{B^\tp}
  }
  }
  -
  \frac{1}{
  \min\set*{
    \frac{\alpha_i}{ \sqrt{ \gamma (U^\tm)^2 + P_\infty^\t(\ubar_{i}) } }
    ,
    \frac{\beta_i}{B^\t}
  }
  }
  \nn
  &=
  \frac{1}{
  \min\set*{
    \frac{\alpha_i}{ \sqrt{ \gamma (U^\t)^2 + P_\infty^\tp(\ubar_{i}) } }
    ,
    \frac{\beta_i}{B^\t}
  }
  }
  -
  \frac{1}{
  \min\set*{
    \frac{\alpha_i}{ \sqrt{ \gamma (U^\tm)^2 + P_\infty^\t(\ubar_{i}) } }
    ,
    \frac{\beta_i}{B^\t}
  }
  }
  \tag{since $t \in \calI'$}
  \nn
  &\leq
  \frac{1}{\alpha_i}
  \sqrt{ \gamma (U^\t)^2 + P_\infty^\tp(\ubar_{i})}
  -
  \frac{1}{\alpha_i}
  \sqrt{\gamma (U^\tm)^2 + P_\infty^\t(\ubar_{i})
  }
  \nn 
  &\leq 
  \frac{1}{\alpha_i}
  \sqrt{\gamma ((U^\t)^2 - (U_{i}^\tm)^2) + \nrm{\ubar_{i}^\t - \ubar_{i}^\tm}_\infty^2}
  \nn
  &\leq 
  \frac{1}{\alpha_i}
  \prn*{
    2 \sqrt{\gamma} U^\t + \sqrt{2} ( \nrm{\ubar_{i}^\t}_\infty + \nrm{\ubar_{i}^\tm}_\infty )
  }
  \leq 
  \frac{4 \sqrt{\gamma} U^\t}{\alpha_i}
  \com 
  \n
\end{align}
where 
in the first inequality we used the fact that 
$z \mapsto 1/\min\set{a, z} - 1/\min\set{b, z}$ is nondecreasing in $z$ when $a \leq b$,
in the second inequality we used the subadditivity of $\sqrt{\cdot}$,
and in the third inequality we used $\nrm{\cdot}_{*,y_a^\t,\phi} \leq \nrm{\cdot}_\infty$ and the triangle inequality,
and in the last inequality we used $\nrm{\ubar_i^\t}_\infty \leq U^\t$.
This completes the proof.
\end{proof}

The following lemma is useful to evaluate the stability of the Markov chain in the proof of \Cref{lem:diff_local_oftrl}.
\begin{lemma}[{\citealt[Lemma 28]{tsuchiya25corrupted}}]\label[lemma]{lem:localnorm_suffcond}
  Let $\calK$ be a closed and bounded nonempty convex set and $y \in \calK$.
  Let $\delta > 0$ and $f$ be a real-valued strictly convex function over $\calK$ and $x^* = \argmin_{x' \in \calK} f(x')$ be the unique minimizer of $f$.
  Suppose that for any $z \in \calK$ such that $\nrm{z - y} = \delta$ for a norm $\nrm{\cdot}$, it holds that $f(z) \geq f(y)$.
  Then, $\nrm{x^* - y} < \delta$.
\end{lemma}

We also exploit the fact that the log-barrier function $\phi(x) = - \sum_{k=1}^d \log(x(k))$ is a $d$-self-concordant barrier over the positive orthant.
In particular, we use the following two lemmas.
\begin{lemma}[{\citealt[Theorem 2.1.1]{nesterov94interior}}]\label[lemma]{lem:sc_hessian_stab}
  Let $\calS$ be an open convex subset of a finite-dimensional real vector space.
  Let $f$ be a self-concordant function on $\calS$.
  Then, for any $y \in \calS$ such that $\nrm{x - y}_{y,f} < 1$, 
  \begin{equation}
    \prn{1 - \nrm{x - y}_{y, f}}^2 
    \nabla^2 f(y) 
    \preceq
    \nabla^2 f(x) 
    \preceq 
    \frac{1}{\prn{1 - \nrm{x - y}_{y,f}}^2}
    \nabla^2 f(y) 
    \per 
    \n
  \end{equation}
\end{lemma}
The following lemma is adopted from \citet[Lemma 22]{tsuchiya25corrupted}.
\begin{lemma}\label[lemma]{lem:scbarrier_properties}
  Let $f$ be a $\vartheta$-self-concordant barrier on $\calK$.
  Then, for any $x \in \interior(\calK)$ such that $\nabla^2 f(x)$ is invertible,
  it holds that $\nrm{\nabla f(x)}_{*,x,f}^2 \leq \vartheta$.
\end{lemma}

Using \Cref{lem:scbarrier_properties,lem:localnorm_suffcond,lem:sc_hessian_stab,lem:lr_diff}, we prove the following lemma, which guarantees the stability of the Markov chain under the adaptive learning rate $\eta_{i}^\t$ in \cref{eq:oftrl_multiplayer}.
In the lemma, we ignore the player index $i \in [n]$ for notational simplicity; for example, we abbreviate $x_i^\t$ as $x^\t$, $y_{i,a}^\t$ as $y_a^\t$, $u_{i,a}^\t$ as $u_a^\t$, $u_i^\t$ as $u^\t$, $\eta_{i}^\t$ as $\eta^\t$, and $\calA_i$ as $\calA$, where we recall that we use $a$ for the index for actions and $i$ for the index of players.

\begin{lemma}\label[lemma]{lem:diff_local_oftrl}
  Suppose that $T \geq 3$ and 
  consider the following OFTRL update in~\cref{eq:oftrl_multiplayer}:
  \begin{align}
    y_{a}^\t
    =
    \argmax_{y \in \Delta_{m}}
    \set*{
      - \Phi_a^\t(y) 
    }
    \com
    \quad
    \Phi_a^\t(y)  
    \coloneqq
    -
    \eta^\t
    \inpr[\bigg]{y, \ubar_{a}^\tm + \sum_{s=1}^{t-1} \ubar_{a}^\s}
    +
    \phi(y)
    \per
    \n
  \end{align}
  Then, if $t \in \calI' \supseteq \calI$, it holds that 
  \begin{equation}\label{eq:diff_local_oftrl}
    \sum_{a \in \calA}
    \nrm{y_a^\tp - y_a^\t}_{y_a^\t,\phi} 
    \leq 
    \frac12 \per
  \end{equation}
\end{lemma}
This is a variant of \citet[Lemma 30]{tsuchiya25corrupted}, modified to work for the scale-free and scale-invariant setting.
\begin{proof}[Proof of \Cref{lem:diff_local_oftrl}]
We recall that $m = \abs{\calA}$ 
and 
use $\calM_m = (\Delta_m)^m = \Delta_m \times \cdots \times \Delta_m$ to denote the Cartesian product of $m$ probability simplices.\footnote{We use $\calM_m$ to denote the Cartesian product of $m$ probability simplices since this is equivalent to the set of all row stochastic matrices.}
Define a strictly convex function $\Psi^\tp \colon \calM_m \to \R$ by
\begin{equation}
  \Psi^\tp(\bm{w})
  =
  \Psi^\tp(w_1, \dots, w_m) 
  = 
  \sum_{a \in \calA} \Phi_a^\tp (w_a)
  \per
  \nonumber
\end{equation}
Note that
for any $\hbm = (h_1, \dots, h_m) \in \calM_m$, 
the local norm $\nrm{\hbm}_{\ybm^\t, \Psi^\tp}$ is given by
\begin{align}
  \nrm{\hbm}_{\ybm^\t, \Psi^\tp} 
  &=
  \sqrt{
    \hbm^\top
    \diag\prn*{
      \set[\big]{
        \nabla^2 \Phi_a^\tp(y_a^\t)
      }_{a \in \calA}
    }
    \hbm
  }
  \nn
  &=
  \sqrt{
    \sum_{a \in \calA} 
    h_a^\top \nabla^2 \Phi_a^\tp(y_a^\t) h_a
  }
  =
  \sqrt{
    \sum_{a \in \calA} 
    \nrm{h_a}_{y_a^\t, \Phi_a^\tp}^2
  }
  \per
  \label{eq:prod_local_norm}
\end{align}
Now from the fact that $y_a^\tp$ is the minimizer of the strongly convex function $\Phi_a^\tp$ for each $a \in \calA$,
the point $\ybm^\t \coloneqq (y_1^\t, \dots, y_m^\t) \in \calM_m$ is the unique minimizer of $\Psi^\tp$.
Hence from \Cref{lem:localnorm_suffcond}, to prove the claim of the lemma,
it suffices to prove that 
for any $\zbm = (z_1, \dots, z_m) \in \calM_m$ satisfying $\nrm{\zbm - \ybm^\t}_{\ybm^\t, \Psi^\tp} = 1/(2 \sqrt{m})$, it holds that $\Psi^\tp(\zbm) \geq \Psi^\tp(\ybm^\t)$.
In fact, if this is proven, then \Cref{lem:localnorm_suffcond} implies $\nrm{\ybm^\tp - \ybm^\t}_{\ybm^\t, \Psi^\tp} \leq 1/(2 \sqrt{m})$, and thus the LHS of \cref{eq:diff_local_oftrl} is upper bounded by
\begin{equation}
  \sum_{a \in \calA} \nrm{y_a^\tp - y_a^\t}_{y_a^\t, \phi}
  =
  \sum_{a \in \calA} \nrm{y_a^\tp - y_a^\t}_{y_a^\t, \Phi_a^\tp}
  \!\leq\!
  \sqrt{
    m
    \sum_{a \in \calA} \nrm{y_a^\tp \!- y_a^\t}_{y_a^\t, \Phi_a^\tp}^2
  }
  \!=\!
  \sqrt{m} \nrm{\ybm^\tp \!- \ybm^\t}_{y^\t, \Psi^\tp}
  \!\leq\!
  \frac12
  \com 
  \n
\end{equation}
where the first inequality follows from the Cauchy--Schwarz inequality
and 
the equality follows from~\cref{eq:prod_local_norm}.
This is the claim of \Cref{lem:diff_local_oftrl}.

In what follows, we will prove that 
for any $\zbm = (z_1, \dots, z_m) \in \calM_m$ satisfying $\nrm{\zbm - \ybm^\t}_{\ybm^\t, \Psi^\tp} = 1/(2\sqrt{m})$, it holds that $\Psi^\tp(\zbm) \geq \Psi^\tp(\ybm^\t)$.
Let $h_a = z_a - y_a^\t \in \R^m$ for each $a \in \calA$.
Note that this $h_a$ satisfies
\begin{equation}\label{eq:new_refactor}
  \nrm{h_a}_{y_a^\t, \phi}
  =
  \nrm{z_a - y_a^\t}_{y_a^\t, \phi}
  =
  \nrm{z_a - y_a^\t}_{y_a^\t, \Phi_a^\tp}
  \leq
  \nrm{\zbm - \ybm^\t}_{\ybm^\t, \Psi^\tp}
  =
  \frac{1}{2 \sqrt{m}}
  \per
\end{equation}
Let us fix $a \in \calA$ and we will lower bound $\Phi_a^\tp(z_a)$.
From Taylor's theorem,
there exists a point $\xi_a^\t = \gamma z_a + (1 - \gamma) y_a^\t$ for some $\gamma \in [0, 1]$ such that
\begin{equation}\label{eq:taylor_suma}
  \Phi_a^\tp(z_a) 
  =
  \Phi_a^\tp(y_a^\t)
  +
  \inpr{ \nabla \Phi_a^\tp (y_a^\t), h_a}
  +
  \frac12 h_a^\top \nabla^2 \Phi_a^\tp(\xi_a^\t) h_a
  \per
\end{equation}

We will lower bound the second term in the RHS of \cref{eq:taylor_suma} below.
From the first-order optimality condition at $y_a^\t$, we have
\begin{align}
  \nabla \Phi_a^\tp (y_a^\t)
  &=
  -
  \eta^\tp \prn*{
    2 \ubar_a^\t + \sum_{s=1}^{t} \ubar_a^\s
  }
  + 
  \nabla \phi(y_a^\t)
  \nn 
  &=
  - \eta^\tp (2 \ubar_a^\t - \ubar_a^\tm)
  -
  \frac{\eta^\tp}{\eta^\t}
  \eta^\t
  \prn*{
    \ubar_a^\tm + \sum_{s=1}^{t-1} \ubar_a^\s
  }
  + 
  \nabla \phi(y_a^\t)
  \nn 
  &=
  - \eta^\tp (2 \ubar_a^\t - \ubar_a^\tm)
  -
  \frac{\eta^\tp}{\eta^\t}
  \nabla \phi(y_a^\t)
  + 
  \nabla \phi(y_a^\t)
  \nn 
  &=
  - \eta^\tp \brk*{
    (2 \ubar_a^\t - \ubar_a^\tm)
    -
    \prn*{\frac{1}{\eta^\tp} - \frac{1}{\eta^\t}}
    \nabla \phi(y_a^\t)
  }
  \com
  \label{eq:firstorder_diff_local_oftrl}
\end{align}
where the third equality follows from 
$\nabla \Phi_a^\t(y_a^\t) = - \eta^\t (\ubar_a^\tm + \sum_{s=1}^{t-1} \ubar_a^\s) + \nabla \phi(y_a^\t) = 0$, which follows from the first-order optimality condition of $y_a^\t$.
Hence, the term in the RHS of \cref{eq:taylor_suma} is lower bounded as 
\begin{align}
  \inpr{ \nabla & \Phi_a^\tp (y_a^\t), h_a}
  \geq 
  -
  \eta^\tp
  \nrm*{ 
    - 2 \ubar_a^\t + \ubar_a^\tm 
    + 
    \prn*{\frac{1}{\eta^\tp} - \frac{1}{\eta^\t}} \nabla \phi(y_a^\t)
  }_{*, y_a^\t, \phi}
  \nrm{h_a}_{y_a^\t, \phi}
  \tag{by \cref{eq:firstorder_diff_local_oftrl} and H\"{o}lder}
  \nn 
  &\geq 
  -
  \eta^\tp
  \prn*{
    \nrm{ 
      - 2 \ubar_a^\t + \ubar_a^\tm 
    }_{*, y_a^\t, \phi}
    +
    \prn*{\frac{1}{\eta^\tp} - \frac{1}{\eta^\t}}
    \nrm{ 
      \nabla \phi(y_a^\t)
    }_{*, y_a^\t, \phi}
  }
  \cdot
  \frac{1}{2 \sqrt{m}}
  \tag{by \cref{eq:new_refactor}}
  \nn
  &\geq 
  -
  \frac{\eta^\tp}{2 \sqrt{m}}
  \prn*{
    \nrm{ 
      - 2 \ubar_a^\t + \ubar_a^\tm 
    }_\infty
    +
    \frac{4 \sqrt{\gamma} U^\t}{\alpha}
    \nrm{ 
      \nabla \phi(y_a^\t)
    }_{*, y_a^\t, \phi}
  }
  \tag{$
  \nrm{\cdot}_{*,y_a^\t,\phi} 
  \leq 
  \nrm{\cdot}_\infty
  $ and \Cref{lem:lr_diff} with the assumption that $t \in \calI'$}
  \nn
  &\geq 
  -
  \frac{\eta^\tp}{2 \sqrt{m}}
  \prn*{
    U^\t
    (
    2 x^\t(a)
    +
    x^\tm(a)
    )
    +
    \frac{4 \sqrt{\gamma} U^\t}{\alpha}
  }
  \nn
  &=
  -
  \frac{\eta^\tp \, U^\t}{2 \sqrt{m}}
  \prn*{
    2 x^\t(a)
    +
    x^\tm(a)
    +
    \frac{4 \sqrt{\gamma}}{\alpha}
  }
  \com
  \label{eq:taylor_second_suma_mid}
\end{align}
where the fourth inequality follows from 
$\nrm{\ubar_{a}^\t}_\infty \leq U^\t x^\t(a)$ and $\nrm{\nabla \phi(y_a^\t)}_{*, y_a^\t, \phi} \leq \sqrt{m}$,
which holds since $\phi$ is $m$-self-concordant barrier and \Cref{lem:scbarrier_properties}.

Combining \cref{eq:taylor_suma} with \cref{eq:taylor_second_suma_mid} gives
\begin{align}
  &
  \Phi_a^\tp(z_a)
  \nn
  &
  \geq
  \Phi_a^\tp(y_a^\t)
  \!-\!
  \frac{\eta^\tp \, U^\t}{2 \sqrt{m}}
  \prn[\bigg]{
    2 x^\t(a) \!+\! x^\tm(a)
    \!+\!
    \frac{4 \sqrt{\gamma}}{\alpha}
  }
  +
  \frac12
  \nrm{h_a}_{\xi_a^\t, \Phi_a^\tp}^2
  \per 
  \label{eq:Faz_lower}
\end{align}

We next consider the last term in the RHS of \cref{eq:Faz_lower}.
From the property of self-concordant barriers in \Cref{lem:sc_hessian_stab},
\begin{equation}\label{eq:taylor_third_suma_pre}
  \nrm{h_a}_{\xi_a^\t, \phi}^2
  \!\geq\!
  \prn[\Big]{1 - \nrm{y_a^\t - \xi_a^\t}_{y_a^\t, \phi}}^2
  \nrm{h_a}_{y_a^\t, \phi}^2
  \!=\!
  \prn[\Big]{1 - \gamma \nrm{z_a - y_a^\t}_{y_a^\t, \phi}}^2
  \nrm{h_a}_{y_a^\t, \phi}^2
  \!\geq\!
  \frac14 \nrm{h_a}_{y_a^\t, \phi}^2
  \com
\end{equation}
where the last inequality follows from \cref{eq:new_refactor}.
Using this inequality, we can lower bound the last term in the RHS of \cref{eq:Faz_lower} as
\begin{align}\label{eq:taylor_third_suma}
  &
  \frac12
  \sum_{a \in \calA}
  \nrm{h_a}_{\xi_a^\t, \Phi_a^\tp}^2
  =
  \frac{1}{2}
  \sum_{a \in \calA}
  \nrm{h_a}_{\xi_a^\t, \phi}^2
  \nn
  &\geq 
  \frac{1}{8}
  \sum_{a \in \calA}
  \nrm{h_a}_{y_a^\t, \phi}^2
  =
  \frac{1}{8}
  \sum_{a \in \calA}
  \nrm{h_a}_{y_a^\t, \Phi_a^\tp}^2
  =
  \frac{1}{8}
  \nrm{\hbm}_{\ybm^\t, \Psi^\tp}^2
  =
  \frac{1}{32 m}
  \com 
\end{align}
where the first inequality follows from \cref{eq:taylor_third_suma_pre},
the third equality from \cref{eq:prod_local_norm},
and the last equality from $\nrm{\hbm}_{\ybm^\t, \Psi^\tp} = 1/(2\sqrt{m})$.

Therefore, summing up the inequality \cref{eq:Faz_lower} over $a \in \calA$ and using \cref{eq:taylor_third_suma}, we obtain
\begin{align}
  \Psi^\tp(\zbm)
  =
  \sum_{a \in \calA}
  \Phi_a^\tp(z_a) 
  &
  \geq
  \sum_{a \in \calA}
  \Phi_a^\tp(y_a^\t)
  -
  \frac{\eta^\tp U^\t}{2 \sqrt{m}}
  \prn*{3 + \frac{4 m \sqrt{\gamma}}{\alpha}}
  +
  \frac{1}{32}
  \nn 
  &
  \geq 
  \sum_{a \in \calA} \Phi_a^\tp(y_a^\t)
  =
  \Psi^\tp(\ybm^\t)
  \com 
  \n
\end{align}
where in the first inequality we used the fact that $x^\t, x^\tm \in \Delta_m$ are elements in the probability simplex 
and 
in the last inequality we used $T \geq 3$ and $\eta^\t \leq \sqrt{m} / (48 U^\t)$.
This completes the proof of \Cref{lem:diff_local_oftrl}.
\end{proof}

Now we can prove the following lemma relating the stability of the output of $m_i$-experts and the stability of the Markov chain defined by~$Q_i^\t$.
\begin{lemma}\label[lemma]{lem:stationary_stab}
  We assume the conditions of \Cref{lem:diff_local_oftrl}.
  Then, it holds that
  \begin{equation}
    \nrm{x^\t - x^\tm}_1^2 
    \leq 
    64 \abs{\calA}
    \sum_{a \in \calA}
    \nrm{ y_{a}^\t - y_{a}^\tm }_{y_a^\tm, \phi}^2
    \per 
    \n
  \end{equation}
\end{lemma}
This lemma will be used in the proof of \Cref{thm:indiv_swapreg} to evaluate the negative term in the RVU bound in~\cref{eq:oftrl_logbarrier}.
\begin{proof}
We have
\begin{equation}\label{eq:mua2yadiff}
  \mu_a^\t 
  \coloneqq
  \max_{b \in \calA} 
  \abs*{
    1 - \frac{y_a^\t(b)}{y_a^\tm(b)}
  }
  \leq  
  \sqrt{
    \sum_{b \in \calA} 
    \prn*{
      1 - \frac{y_a^\t(b)}{y_a^\tm(b)}
    }^2
  }
  =
  \nrm{y_a^\t - y_a^\tm}_{y_a^\tm, \phi}
  \per
\end{equation}
Taking the summation over $\calA$ of this inequality and using \Cref{lem:diff_local_oftrl},
we have
$\sum_{a \in \calA} \mu_a^\t \leq \nrm{y_a^\t - y_a^\tm}_{y_a^\tm, \phi} \leq 1/2$.
Hence, combining this with \citet[Eq.~(26) in the proof of Lemma 4.2]{anagnostides22uncoupled} gives
$
\nrm{x^\t - x^\tm}_1
\leq 
8 \sum_{a \in \calA} \mu_a^\t 
.
$
Finally, using the last inequality and the Cauchy--Schwarz inequality, we obtain
\begin{equation}
  \nrm{x^\t - x^\tm}_1^2
  \leq 
  64
  \prn*{
    \sum_{a \in \calA} \mu_a^\t 
  }^2
  \leq 
  64 \abs{\calA} 
  \sum_{a \in \calA} \prn*{\mu_a^\t}^2
  \leq 
  64 \abs{\calA}
  \sum_{a \in \calA}
  \nrm{y_a^\t - y_a^\tm}_{y_a^\tm, \phi}^2
  \com 
  \n
\end{equation}
where the last inequality from \cref{eq:mua2yadiff}.
This completes the proof.
\end{proof}

\subsection{Upper bounding the swap regret using lemmas from \Cref{subsec:stab_analysis}}\label{app:pre_swap_subsec}
Here we will upper bound the swap regret using the preliminary lemmas from the last section.
In this part of the analysis, the use of doubling clipping in the design of \Cref{alg:multiple_player_swap} plays an important role. 
We first prepare the following three lemmas, which are properties of the doubling clipping.

\begin{lemma}\label{lem:doubling_clipping_prop}
Suppose that $t \in \calI$. Then, it holds that $B^\tp = B^\t = B^\tm$,  $\ubar_{i,a}^\t = u_{i,a}^\t$, $\ubar_{i,a}^\tm = u_{i,a}^\tm$,  $\ubar_{i}^\t = u_{i}^\t$, and $\ubar_{i}^\tm = u_{i}^\tm$ for each $i \in [n]$ and $a \in \calA_i$.
\end{lemma}
\begin{proof}
If $t \in \calI'$ is not the round with jump, we then have $B^\tp = B^\t$, and thus $\ubar_{i,a}^\t = u_{i,a}^\t$ and $\ubar_{i}^\t = u_{i}^\t$.
Hence, if $t \in \calI$, from the definition of $\calI$, we have $t - 1 \not\in \calJ'$ and thus $B^\t = B^\tm$.
This implies $\ubar_{i,a}^\tm = u_{i,a}^\tm$ and  $\ubar_{i}^\tm = u_{i}^\tm$ for $t \in \calI$. 
\end{proof}

\begin{lemma}\label{lem:BU_relation}
For $t \in [T]$, it holds that $B^\tp \leq 2 U^\t$.
For $t \in \calI$, it holds that $B^\tp = B^\t = B^\tm \leq 2 U^\tm$.
\end{lemma}
\begin{proof}
From the definitons of $B^\t$ and $U^\t$, it holds for any $t \in [T]$ that $B^\tp \leq 2 U^\t$.
Hence, using this and \Cref{lem:doubling_clipping_prop} with $t \in \calI$, we have $B^\tp = B^\t = B^\tm \leq 2 U^\tm$.
\end{proof}

\begin{lemma}\label{lem:clip_grad_upper_onestep_back}
For all $t \in [T]$ and $i \in [n]$, it holds that
$
  \nrm{\ubar_i^\t}_\infty \leq 2 B^\t
  .
$
\end{lemma}
\begin{proof}
For $t \in \calJ'$ we have $\nrm{u_i^\t}_\infty \leq B^\tp$,
and for $t \in \calI'$ we have $\nrm{u_i^\t}_\infty \leq 2 B^\t \leq 2 B^\tp$ from the definition of $B^\t$.
Using these inequalities, we have
\begin{equation}
  \nrm{\ubar_i^\t}_\infty = \frac{B^\t}{B^\tp} \nrm{u_i^\t}_\infty \leq 2 B^\t
  \com
  \n
\end{equation}
which is the desired bound.
\end{proof}

The following lemma is useful to evaluate the regret coming from the jump clipping so that the swap regret upper bound does depend on the initial gradient value $\omega$ (see \Cref{fn:def_omega} for the formal definition).
\begin{lemma}\label{lem:ubar_sum_jump_rounds}
It holds for any $i \in [n]$ and $a \in \calA_i$ that
\begin{align}
  \sum_{t \in \calJ} \nrm{\ubar_{i,a}^\t}_\infty
  &\leq
  \sum_{t \in \calJ} \nrm{\ubar_i^\t}_\infty
  \leq 
  16 \Umax
  \com
  \nn
  \sum_{t \in \calJ} \nrm{\ubar_i^\t}_\infty^2
  &\leq 
  64 \Umax^2
  \per
  \n
\end{align}
\end{lemma}
\begin{proof}
Using \Cref{lem:clip_grad_upper_onestep_back}, we have
\begin{align}
  \sum_{t \in \calJ} \nrm{\ubar_{i}^\t}_\infty
  \leq
  2 \sum_{t \in  \calJ} B^\t
  \leq
  4 \sum_{k \in \N \cup \set{0} : 2^k \omega \leq 2 \Umax} 2^k \omega
  \leq
  16 \Umax
  \com
  \n
\end{align}
where the last inequality follows since, for
\begin{equation}
  k^*
  \coloneqq 
  \max\set{k \in \N : 2^k \omega \leq 2 \Umax} 
  =
  \floor{\log_2 (2 \Umax / \omega)}
  \com
  \n
\end{equation}
it holds that
\begin{equation}\label{eq:2kdelta_sum_upper}
  \sum_{k \in \N \cup \set{0} : 2^k \omega \leq 2 \Umax} 2^k \omega
  =
  \sum_{k = 0}^{k^*} 2^k \omega
  =
  2^{k^* + 1} \omega
  -
  \omega
  \leq
  4 \Umax - \omega
  \leq
  4 \Umax
  \com
\end{equation}
where the second inequality follows from $2^{k^* + 1} \omega \leq 2 \cdot 2^{k^*} \omega \leq 2 \cdot 2 \Umax = 4 \Umax$ from the definition of $k^*$.

By the similar argument, we obtain
\begin{align}
  \sum_{t \in \calJ} \nrm{\ubar_{i}^\t}_\infty^2
  \leq
  4 \sum_{t \in  \calJ} (B^\t)^2
  \leq
  8 \sum_{k \in \N \cup \set{0} : 2^k \omega \leq 2 \Umax} \prn*{2^k \omega}^2
  \leq
  64 \Umax^2
  \per
  \n
\end{align}
Here, the first inequality follows from \Cref{lem:clip_grad_upper_onestep_back} and the last inequality follows from
\begin{equation}
  \sum_{k \in \N \cup \set{0} : 2^k \omega \leq 2 \Umax} \prn{2^k \omega}^2
  =
  \sum_{k=0}^{k^*} \prn{2^k \omega}^2
  \leq
  2 \Umax
  \sum_{k = 0}^{k^*} 2^k \omega
  \leq
  8 \Umax^2
  \com
  \n
\end{equation}
where we used~\cref{eq:2kdelta_sum_upper}.
This completes the proof.
\end{proof}

We then upper bound the swap regret $\SwapReg_{x_i}^T$ by the following lemma.
\begin{lemma}\label[lemma]{lem:swapreg_pre_bound}
\Cref{alg:multiple_player_swap} achieves
\begin{align}
  &
  \SwapReg_{x_i}^T
  \nn
  &
  \leq
  \frac{U^\T m_i^2 \log T}{\beta_i}
  +
  \frac{m_i^2 \log T}{\alpha_i}
  \sqrt{ \gamma (U^\T)^2 + \sum_{j \in [n]} P_\infty^{T+1}(\ubar_j)}
  +
  16 \alpha_i \sqrt{P_\infty^{T+1}(\ubar_i)}
  \nn
  &\qquad
  -
  \frac{1}{2^{11} m_i \alpha_i}
  \sum_{t \in \calI}
  \sqrt{ \gamma (U^\tm)^2 + \sum_{j \in [n]} P_\infty^\t(\ubar_{j}) }
  \nrm{x_i^\tp - x_i^\t}_1^2
  +
  \Umax (38 m_i + 8)
  \per
  \label{eq:swap_upper_indiv_p}
\end{align}
\end{lemma}

Now we are ready to provide the proof of \Cref{thm:indiv_swapreg}.

\begin{proof}
The regret in each base learner $a \in \calA_i$ can be decomposed as
\begin{align}
  {\Reg}_{i,a}^T(y)
  =
  \sumT \inpr{y - y_{i,a}^\t, u_{i,a}^\t}
  &=
  \sumT \inpr{y - y_{i,a}^\t, u_{i,a}^\t - \ubar_{i,a}^\t}
  +
  \sumT \inpr{y - y_{i,a}^\t, \ubar_{i,a}^\t}
  \nn
  &=
  \sumT \inpr{y - y_{i,a}^\t, \ubar_{i,a}^\t}
  +
  \sumT \zeta_{i,a}^\t
  \per
  \label{eq:swap_decompose_clip}
\end{align}
where we defined 
\begin{equation}
  \zeta_{i,a}^\t
  \coloneqq
  \inpr{y - y_{i,a}^\t, u_{i,a}^\t - \ubar_{i,a}^\t}
  \per
  \n
\end{equation}

From the definition of $\prn{y_{i,a}^\t}_{t=1}^T$ and \Cref{lem:oftrl_logbarrier} with $\sum_{a \in \calA} \nrm{y_{i,a}^\tp - y_{i,a}^\t}_{y_{i,a}^\t, \phi} \leq 1/2$ for $t \in \calI$ in \Cref{lem:diff_local_oftrl},
for any $y \in \Delta(\calA_i)$ the first term in \cref{eq:swap_decompose_clip} is upper bounded as
\begin{align}
  \sumT \inpr{y - y_{i,a}^\t, \ubar_{i,a}^\t}
  &\leq 
  \frac{m_i \log T}{\eta_{i}^{T+1}}
  +
  4 \sum_{t \in \calI} \eta_{i}^\t \nrm{\ubar_{i,a}^\t - \ubar_{i,a}^\tm}_{*,y_{i,a}^\t,\phi}^2
  \nn
  &\qquad-
  \sum_{t \in \calI} \frac{1}{16 \eta_{i}^\t} \nrm{y_{i,a}^\tp - y_{i,a}^\t}_{y_{i,a}^\t, \phi}^2
  +
  2 \sum_{t \not\in \calI} \nrm{h^\t - m^\t}_\infty
  +
  6 \Umax
  \nn
  &\leq 
  \frac{m_i \log T}{\eta_{i}^{T+1}}
  +
  4 \sum_{t \in \calI} \eta_{i}^\t \nrm{\ubar_{i,a}^\t - \ubar_{i,a}^\tm}_{*,y_{i,a}^\t,\phi}^2
  \nn
  &\qquad-
  \sum_{t \in \calI} \frac{1}{16 \eta_{i}^\t} \nrm{y_{i,a}^\tp - y_{i,a}^\t}_{y_{i,a}^\t, \phi}^2
  +
  38 \Umax
  \per
  \label{eq:Regia_upper_pre}
\end{align}
where the last inequality holds since from the first statement of \Cref{lem:ubar_sum_jump_rounds}, we have
\begin{equation}
  2 \sum_{t \not\in \calI} \nrm{\ubar_{i,a}^\t - \ubar_{i,a}^\tm}_\infty
  =
  2 \sum_{t \in \calJ} \nrm{\ubar_{i,a}^\t - \ubar_{i,a}^\tm}_\infty
  \leq
  2 \sum_{t \in \calJ} \prn*{ \nrm{\ubar_{i,a}^\t}_\infty + \nrm{\ubar_{i,a}^\tm}_\infty }
  \leq
  32 \Umax
  \per
  \n
\end{equation}
Combining \cref{eq:swap_decompose_clip} with \cref{eq:reg_bias,eq:Regia_upper_pre}, for any $y \in \Delta_{m_i}$ we obtain
\begin{align}
  {\Reg}_{i,a}^T(y)
  &\leq 
  \frac{m_i \log T}{\eta_{i}^{T+1}}
  +
  4 \sum_{t \in \calI} \eta_{i}^\t \nrm{\ubar_{i,a}^\t - \ubar_{i,a}^\tm}_{*,y_{i,a}^\t,\phi}^2
  \nn
  &\qquad-
  \sum_{t \in \calI} \frac{1}{16 \eta_{i}^\t} \nrm{y_{i,a}^\tp - y_{i,a}^\t}_{y_{i,a}^\t, \phi}^2
  +
  38 \Umax
  +
  \sumT \zeta_{i,a}^\t
  \per
  \label{eq:Regia_upper}
\end{align}
Hence, using the reduction from swap regret minimization to the instances of external regret minimization discussed in \Cref{subsec:swapreg2reg} and \cref{eq:Regia_upper},
we have 
\begin{align}
  \SwapReg_{x_i}^T
  &=
  \sum_{a \in \calA_i} {\Reg}_{i,a}^T
  \nn
  &\leq
  \frac{m_i^2 \log T}{\eta_{i}^{T+1}}
  +
  4 \sum_{t\in\calI} \eta_{i}^\t \sum_{a \in \calA_i} \nrm{u_{i,a}^\t - u_{i,a}^\tm}_{*,y_{i,a}^\t,\phi}^2
  \nn
  &\qquad-
  \sum_{t \in \calI} \frac{1}{16 \eta_{i}^\t} \sum_{a \in \calA_i} \nrm{y_{i,a}^\tp - y_{i,a}^\t}_{y_{i,a}^\t, \phi}^2
  +
  38 \Umax m_i
  +
  \sumT \sum_{a \in \calA_i} \zeta_{i,a}^\t
  \nn
  &\leq
  \frac{m_i^2 \log T}{\eta_{i}^{T+1}}
  +
  4 \sum_{t\in\calI} \eta_{i}^\t \sum_{a \in \calA_i} \nrm{u_{i,a}^\t - u_{i,a}^\tm}_{*,y_{i,a}^\t,\phi}^2
  \nn
  &\qquad-
  \sum_{t \in \calI} \frac{1}{16 \eta_{i}^\t} \sum_{a \in \calA_i} \nrm{y_{i,a}^\tp - y_{i,a}^\t}_{y_{i,a}^\t, \phi}^2
  +
  \Umax (38 m_i + 8)
  \com
  \label{eq:swapreg_adv_pre}
\end{align}
Here, the last inequality follows since for any $y \in \Delta(\calA_i)$, we have
\begin{align}
  \sum_{a \in \calA_i} \sumT \zeta_{i,a}^\t
  &\leq
  2 \sum_{a \in \calA_i} \sumT \nrm{u_{i,a}^\t - \ubar_{i,a}^\t}_\infty
  \tag{H\"{o}lder}
  \nn
  &=
  2 \sum_{a \in \calA_i} \sumT \prn*{1 - \frac{B^\t}{B^\tp}} \nrm{u_{i,a}^\t}_\infty
  \tag{def.~of $\ubar_{i,a}^\t$ in \Cref{eq:clipped_grad}}
  \nn
  &\leq
  4 \sum_{a \in \calA_i} \sumT \prn*{1 - \frac{B^\t}{B^\tp}} x^\t_i(a) B^\tp
  \nn
  &=
  4 \sumT \prn*{1 - \frac{B^\t}{B^\tp}} B^\tp
  \tag{$x_i^\t \in \Delta_{m_i}$}
  \nn
  &\leq
  4 B^{T+1} \tag{telescoping}
  \nn
  &\leq 
  8 U^\T
  \com
  \label{eq:reg_bias}
\end{align}
where the second inequality follows from $\nrm{u_{i,a}^\t}_\infty = x^\t_i(a) \nrm{u_{i}^\t}_\infty \leq 2 x^\t_i(a) B^\tp$ (see the proof of \Cref{lem:clip_grad_upper_onestep_back}) and the last inequality follows from \Cref{lem:BU_relation}.

We first evaluate the second term in \cref{eq:swapreg_adv_pre}.
For each $t \in \calI$, we have
\begin{align}
  &
  \sum_{a \in \calA_i}
  \nrm{\ubar_{i,a}^\t - \ubar_{i,a}^\tm}_{*,y_{i,a}^\t,\phi}^2
  =
  \sum_{a \in \calA_i}
  \nrm{u_{i,a}^\t - u_{i,a}^\tm}_{*,y_{i,a}^\t,\phi}^2
  =
  \sum_{a \in \calA_i}
  \nrm{u_i^\t x_i^\t(a) - u_i^\tm x_i^\tm(a)}_{*,y_{i,a}^\t,\phi}^2
  \nn 
  &\leq 
  2 \sum_{a \in \calA_i}
  \nrm{
    u_i^\t x_i^\t(a) 
    - 
    u_i^\tm x_i^\t(a)
  }_{*,y_{i,a}^\t,\phi}^2
  +
  2 \sum_{a \in \calA_i}
  \nrm{
    u_i^\tm x_i^\t(a)
    -
    u_i^\tm x_i^\tm(a)
  }_{*,y_{i,a}^\t,\phi}^2
  \nn 
  &=
  2 \sum_{a \in \calA_i}
  \prn{x_i^\t(a)}^2
  \nrm{
    u_i^\t
    - 
    u_i^\tm 
  }_{*,y_{i,a}^\t,\phi}^2
  +
  2 \sum_{a \in \calA_i}
  \prn{x_i^\t(a) - x_i^\tm(a)}^2
  \nrm{
    u_i^\tm 
  }_{*,y_{i,a}^\t,\phi}^2
  \nn 
  &\leq 
  2 
  \nrm{
    u_i^\t
    - 
    u_i^\tm 
  }_\infty^2 
  +
  2 (U^\tm)^2 \sum_{a \in \calA_i}
  \prn{x_i^\t(a) - x_i^\tm(a)}^2
  \nn 
  &=
  2 
  \nrm{ u_i^\t - u_i^\tm }_\infty^2
  +
  2 (U^\tm)^2 \nrm{ x_i^\t - x_i^\tm }_1^2
  \com
  \label{eq:hat_swapreg_bound}
\end{align}
where in the first equality we used $\ubar_{i,a}^\t = u_{i,a}^\t$ and $\ubar_{i,a}^\tm = u_{i,a}^\tm$ for $t \in \calI$ in \Cref{lem:doubling_clipping_prop}, and in the last inequality we used $\nrm{\cdot}_2 \leq \nrm{\cdot}_1$.

We next evaluate the third term in \cref{eq:swapreg_adv_pre}.
From \Cref{lem:stationary_stab}, the negative term in \cref{eq:swapreg_adv_pre} is evaluated as
\begin{align}
  &
  \sum_{t \in \calI} \frac{1}{16 \eta_{i}^\t} 
  \sum_{a \in \calA_i} \nrm{y_{i,a}^\tp - y_{i,a}^\t}_{y_{i,a}^\t, \phi}^2  
  \geq
  \sum_{t \in \calI} \frac{1}{2^{10} m_i \eta_{i}^\t} 
  \nrm{x_i^\tp - x_i^\t}_1^2
  \nn
  &=
  \sum_{t \in \calI} \frac{1}{2^{11} m_i \eta_{i}^\t} 
  \nrm{x_i^\tp - x_i^\t}_1^2
  +
  \sum_{t \in \calI} \frac{1}{2^{11} m_i \eta_{i}^\t} 
  \nrm{x_i^\tp - x_i^\t}_1^2
  \nn
  &
  \geq
  \frac{1}{2^{11} m_i \alpha_i}
  \sum_{t \in \calI}
  \sqrt{ \gamma (U^\tm)^2 + \sum_{j \in [n]} P_\infty^\t(\ubar_{j}) }
  \nrm{x_i^\tp - x_i^\t}_1^2
  +
  \frac{1}{2^{11} m_i \beta_i}
  \sum_{t \in \calI} B^\t
  \nrm{x_i^\tp - x_i^\t}_1^2
  \com
  \label{eq:negterm_lower}
\end{align}
where the first inequality follows from
$
\nrm{x_i^\t - x_i^\tm}_1^2 
\leq 
64 m_i \sum_{a \in \calA_i}
\nrm{ y_{i,a}^\t - y_{i,a}^\tm }_{y_{i,a}^\tm, \phi}^2
$
for $t \in \calI$ in \Cref{lem:stationary_stab},
and the second inequality follows from 
$
\eta_{i}^\t \leq 
\frac{\alpha_i}{\sqrt{\gamma (U^\tm)^2 + \sum_{j \in [n]} P_\infty^\t(\ubar_{j})}}
$
and
$\eta_{i}^\t \leq \beta_i / B^\t$
in~\cref{eq:oftrl_multiplayer}.

Combining \cref{eq:swapreg_adv_pre} with \cref{eq:hat_swapreg_bound,eq:negterm_lower}, we obtain
\begin{align}
  &
  \SwapReg_{x_i}^T
  \nn
  &
  \leq
  \frac{m_i^2 \log T}{\eta_{i}^{T+1}}
  +
  8 \sum_{t \in \calI} \eta_{i}^\t \nrm{ u_i^\t - u_i^\tm }_\infty^2
  +
  8 \sumT \eta_{i}^\t (U^\tm)^2 \nrm{ x_i^\t - x_i^\tm }_1^2
  \nn
  &\qquad
  -
  \prn*{
    \frac{1}{2^{11} m_i \alpha_i}
    \sum_{t \in \calI}
    \sqrt{ \gamma (U^\tm)^2 + \sum_{j \in [n]} P_\infty^\t(\ubar_{j}) }
    \nrm{x_i^\tp - x_i^\t}_1^2
    +
    \frac{1}{2^{11} \beta_i m_i}
    \sum_{t \in \calI} B^\t
    \nrm{x_i^\tp - x_i^\t}_1^2
  }
  \nn
  &\qquad
  +
  \Umax (38 m_i + 8)
  \nn
  &
  \leq
  \frac{m_i^2 \log T}{\eta_{i}^{T+1}}
  +
  8 \sum_{t \in \calI} \eta_{i}^\t \nrm{ u_i^\t - u_i^\tm }_\infty^2
  -
  \frac{1}{2^{11} m_i \alpha_i}
  \sum_{t \in \calI}
  \sqrt{ \gamma (U^\tm)^2 + \sum_{j \in [n]} P_\infty^\t(\ubar_{j}) }
  \nrm{x_i^\tp - x_i^\t}_1^2
  \nn
  &\qquad
  +
  \Umax (38 m_i + 8)
  \per
  \label{eq:swap_upper_mid}
\end{align}
Here, in the last inequality we used 
\begin{align}
  &
  8 \sum_{t \in \calI} \eta_{i}^\t (U^\tm)^2 \nrm{ x_i^\t - x_i^\tm }_1^2
  -
  \frac{1}{2^{11} \beta_i m_i}
  \sum_{t \in \calI} B^\t
  \nrm{x_i^\tp - x_i^\t}_1^2
  \nn
  &\leq
  8 \eta_{i}^{1} (U^{0})^2 \nrm{ x_i^{1} - x_i^{0} }_1^2
  +
  \sum_{t \in \calI \setminus \set{1}}
  \prn*{
    8 \eta_{i}^\t (U^\tm)^2 
    -
    \frac{B^\tm}{2^{11} \beta_i m_i}
  }
  \nrm{x_i^\t - x_i^\tm}_1^2
  \leq
  0
  \com
  \n
\end{align}
where the last inequality holds since $\eta_i^\t \leq \beta_i / B^\t$ with $\beta_i = 1 / (256 \sqrt{m_i})$ and the fact that $B^\tm = B^\t \leq 2 U^\tm$ for $t \in \calI$ in \Cref{lem:BU_relation}.

The first two terms in \cref{eq:swap_upper_mid} are further upper bounded as
\begin{align}
  &
  \frac{m_i^2 \log T}{\eta_{i}^{T+1}}
  +
  8 \sum_{t \in \calI} \eta_{i}^\t \nrm{ u_i^\t - u_i^\tm }_\infty^2
  \nn
  &
  \leq
  \frac{U^\T m_i^2 \log T}{\beta_i}
  +
  \frac{m_i^2 \log T}{\alpha_i}
  \sqrt{ \gamma (U^\T)^2 + \sum_{j \in [n]} P_\infty^T(\ubar_j)}
  +
  8 \alpha_i \sum_{t \in \calI} \frac{\nrm{ u_i^\t - u_i^\tm }_\infty^2}{ \sqrt{ \gamma (U^\tm)^2 + \sum_{j \in [n]} P_\infty^\t(\ubar_j)} }
  \nn
  &
  \leq
  \frac{U^\T m_i^2 \log T}{\beta_i}
  +
  \frac{m_i^2 \log T}{\alpha_i}
  \sqrt{ \gamma (U^\T)^2 + \sum_{j \in [n]} P_\infty^T(\ubar_j)}
  +
  8 \alpha_i \sum_{t \in \calI} \frac{\nrm{ u_i^\t - u_i^\tm }_\infty^2}{ \sqrt{ P_\infty^\tp(\ubar_i)} }
  \label{eq:swap_two_terms_pre}
  \\
  &
  \leq
  \frac{U^\T m_i^2 \log T}{\beta_i}
  +
  \frac{m_i^2 \log T}{\alpha_i}
  \sqrt{ \gamma (U^\T)^2 + \sum_{j \in [n]} P_\infty^T(\ubar_j)}
  +
  8 \alpha_i \sumT \frac{\nrm{ \ubar_i^\t - \ubar_i^\tm }_\infty^2}{ \sqrt{ P_\infty^\tp(\ubar_i)} }
  \nn
  &
  \leq
  \frac{U^\T m_i^2 \log T}{\beta_i}
  +
  \frac{m_i^2 \log T}{\alpha_i}
  \sqrt{ \gamma (U^\T)^2 + \sum_{j \in [n]} P_\infty^T(\ubar_j)}
  +
  16 \alpha_i \sqrt{P_\infty^{T+1}(\ubar_i)}
  \per
  \label{eq:swap_two_terms}
\end{align}
Here we used the following facts:
the first inequality follows from the definition of $\eta_i^\t$;
the second inequality (the inequality~\cref{eq:swap_two_terms_pre}) follows from 
\begin{equation}\label{eq:stab_gammaU_lower}
  \gamma (U^\tm)^2 + \sum_{j \in [n]} P_\infty^\t(\ubar_j)
  \geq
  \sum_{j \in [n]} P_\infty^tp(\ubar_j)
  \geq
  P_\infty^\tp(\ubar_i)
  \com
\end{equation}
where the first inequality is due to
\begin{align}
  \nrm{\ubar_i^\t - \ubar_i^\tm}_\infty^2
  &\leq
  2 \nrm{\ubar_i^\t}_\infty^2 + 2 \nrm{\ubar_i^\t}_\infty^2
  \nn
  &\leq
  4 (B^\tm)^2 + 4 (B^\t)^2
  \tag{by $\nrm{\ubar_i^\t} \leq 2 B^\t$ from \Cref{lem:clip_grad_upper_onestep_back}}
  \nn
  &\leq
  8 (U^\tm)^2
  \tag{by $B^\tm \leq B^\t \leq U^\tm$ from \Cref{lem:BU_relation}}
  \com
\end{align}
combined with the choice of $\gamma = 8 n$;
the third inequality follows from the fact that $\ubar_{i,a}^\t = u_{i,a}^\t$ and $\ubar_{i,a}^\tm = u_{i,a}^\tm$ for $t \in \calI$ in \Cref{lem:doubling_clipping_prop};
and the last inequality from $\sumT z^\t / \sqrt{\sum_{s=1}^t z^\s} \leq 2 \sqrt{\sumT z^\t}$ for $z^{1}, \dots, z^\T \geq 0$.

Finally, combining \cref{eq:swap_upper_mid} with \cref{eq:swap_two_terms}, we obtain
\begin{align}
  &
  \SwapReg_{x_i}^T
  \nn
  &
  \leq
  \frac{U^\T m_i^2 \log T}{\beta_i}
  +
  \frac{m_i^2 \log T}{\alpha_i}
  \sqrt{ \gamma (U^\T)^2 + \sum_{j \in [n]} P_\infty^{T+1}(\ubar_j)}
  +
  16 \alpha_i \sqrt{P_\infty^{T+1}(\ubar_i)}
  \nn
  &\qquad
  -
  \frac{1}{2^{11} m_i \alpha_i}
  \sum_{t \in \calI}
  \sqrt{ \gamma (U^\tm)^2 + \sum_{j \in [n]} P_\infty^\t(\ubar_{j}) }
  \nrm{x_i^\tp - x_i^\t}_1^2
  +
  \Umax (38 m_i + 8)
  \com
  \n
\end{align}
which is the desired bound.
\end{proof}

\subsection{Proof of \Cref{thm:indiv_swapreg}}\label{app:proof_swap_subsec}
Here we provide the proof of \Cref{thm:indiv_swapreg}.
Before proving \Cref{thm:indiv_swapreg}, we prepare the following well-known lemma, which relates the squared difference of the underlying utility, $P_\infty^T(u_i)$, and the squared differences of strategies.
\begin{lemma}\label[lemma]{lem:Pu2Px}
It holds that
\begin{equation}
  \sum_{i \in [n]}
  \nrm{u_i^\t - u_i^\tm}_\infty^2
  \leq
  \Umax^2 (n - 1)^2 \sum_{i \in [n]} \nrm{x_i^\t - x_i^\tm}_1^2
  \per
  \n
\end{equation}
\end{lemma}
\begin{proof}
Let $\calA_{-i} = \times_{j \neq i} \calA_j$.
Then,
\begin{align}
  \nrm{u_i^\t - u_i^\tm}_\infty
  &=
  \Umax
  \max_{a_i \in \calA_i}
  \abs*{
    \sum_{a_{-i} \in \calA_{-i}} u_i(a_i, a_{-i}) \prod_{j \neq i} x_j^\t(a_j)
    -
    \sum_{a_{-i} \in \calA_{-i}} u_i(a_i, a_{-i}) \prod_{j \neq i} x_j^\tm(a_j)
  }
  \nn 
  &\leq 
  \Umax
  \sum_{a_{-i} \in \calA_{-i}}
  \abs*{
    \prod_{j \neq i} x_j^\t(a_j)
    -
    \prod_{j \neq i} x_j^\tm(a_j)
  }
  \leq
  \Umax
  \sum_{j \neq i} \nrm{x_j^\t - x_j^\tm}_1
  \com 
  \n
\end{align}
where the first inequality follows from $u_i^\t(a_i, a_{-i}) \in [-\Umax,\Umax] $ and the last inequality follows from the fact that the total variation of two product distributions is bounded by the sum of the total variations of each marginal distribution.
Using this inequality and the Cauchy--Schwarz inequality, we have
\begin{align}
  \nrm{u_i^\t - u_i^\tm}_\infty^2
  &\leq
  \Umax^2 \prn*{ \sum_{j \neq i} \nrm{x_j^\t - x_j^\tm}_1 }^2
  \nn
  &\leq
  \Umax^2 \prn*{ \sqrt{(n - 1) \sum_{j \neq i} \nrm{x_j^\t - x_j^\tm}_1^2 } }^2
  \leq
  \Umax^2 (n - 1) \sum_{j \neq i} \nrm{x_j^\t - x_j^\tm}_1^2
  \per
  \n
\end{align}
This implies
\begin{equation}
  \sum_{i \in [n]}
  \nrm{u_i^\t - u_i^\tm}_\infty^2
  \leq
  \Umax^2 (n - 1) \sum_{i \in [n]} \sum_{j \neq i} \nrm{x_j^\t - x_j^\tm}_1^2
  =
  \Umax^2 (n - 1)^2 \sum_{i \in [n]} \nrm{x_i^\t - x_i^\tm}_1^2
  \com
  \n
\end{equation}
which is the desired bound.
\end{proof}

Finally, we are ready to prove \Cref{thm:indiv_swapreg}.
The analysis uses the stopping-time argument, inspired by the analysis for two-player zero-sum games in \Cref{subsec:two_player_analysis}.

\begin{proof}
Combining \Cref{lem:swapreg_pre_bound} with the fact that $\SwapReg_{x_i}^T \geq 0$, we have
\begin{align}
  &
  \sum_{t \in \calI}
  \sqrt{ \gamma (U^\tm)^2 + \sum_{j \in [n]} P_\infty^\t(\ubar_{j}) }
  \nrm{x_i^\tp - x_i^\t}_1^2
  \nn
  &
  \leq
  2^{11} m_i \alpha_i
  \prn[\Bigg]{
    \frac{U^\T m_i^2 \log T}{\beta_i}
    +
    \frac{m_i^2 \log T}{\alpha_i}
    \sqrt{ \gamma (U^\T)^2 + \sum_{j \in [n]} P_\infty^{T+1}(\ubar_j)}
    \nn
    &\qquad\qquad\qquad\qquad\qquad+
    16 \alpha_i \sqrt{P_\infty^{T+1}(\ubar_i)}
    +
    \Umax (38 m_i + 8)
  }
  \per
  \label{eq:indiv_sb}
\end{align}
We then take the summation over $[n]$ in \cref{eq:indiv_sb}.
For each $t \in [T]$, define 
\begin{equation}
  \mathbb{S}^\t \coloneqq \sum_{j \in [n]} P_\infty^\t(\ubar_j)
  =
  \sum_{j \in [n]} \sum_{s=1}^{t-1} \nrm{\ubar_j^\s - \ubar_j^\sm}_\infty^2
  \per
  \n
\end{equation}
Let $M_1 = \sum_{i \in [n]} m_i$ and $\alpha = \max_{i \in [n]} \alpha_i$ for notational simplicity, and recall $m = \max_{i\in[n]} m_i$.
From the Cauchy--Schwarz inequality, we have
\begin{equation}
  \sum_{i \in [n]}
  16 \alpha_i \sqrt{P_\infty^{T+1}(\ubar_i)}
  \leq
  16 \sqrt{ \sum_{i \in [n]} \alpha_i^2} 
  \sqrt{\sum_{j \in [n]} P_\infty^{T+1}(\ubar_j)}
  =
  16 \sqrt{ \sum_{i \in [n]} \alpha_i^2} \,
  \sqrt{\mathbb{S}^{T+1}}
  \com
  \n
\end{equation}
and thus taking the summation over $[n]$ in \cref{eq:indiv_sb} gives
\begin{align}
  &
  \sum_{i \in [n]} \sum_{t \in \calI}
  \sqrt{ \gamma (U^\tm)^2 + \mathbb{S}^\t }
  \nrm{x_i^\tp - x_i^\t}_1^2
  \nn
  &\leq
  2^{11} m \alpha
  \brk[\Bigg]{
    \sum_{i \in [n]}
    \prn*{\frac{1}{\beta_i} + \frac{2 \sqrt{\gamma}}{\alpha_i}} \Umax m_i^2 \log T
    +
    \prn*{\sum_{i \in [n]} \frac{m_i^2 \log T}{\alpha_i} + 16 \sqrt{\sum_{i \in [n]} \alpha_i^2}} \sqrt{\mathbb{S}^{T+1}}
    \nn
    &\qquad\qquad\qquad+
    \Umax (38 M_1 + 8)
  }
  \nn
  &
  \leq
  2^{11} m \alpha
  \prn{\Gamma + \Lambda \sqrt{\mathbb{S}^{T+1}}}
  \com
  \label{eq:S_upper}
\end{align}
where we defined 
\begin{equation}\label{eq:Gamma_Lambda}
  \Gamma
  \coloneqq 
  \sum_{i \in [n]}
  \prn*{\frac{1}{\beta_i} + \frac{2 \sqrt{\gamma}}{\alpha_i}} \Umax m_i^2 \log T
  + 
  \Umax (38 M_1 + 8)
  \com\
  \Lambda
  \coloneqq
  \sum_{i \in [n]} \frac{m_i^2 \log T}{\alpha_i} + 16 \sqrt{\sum_{i \in [n]} \alpha_i^2}
  \per
\end{equation}

We will lower bound the LHS of \cref{eq:S_upper} using $\mathbb{S}^{T+1}$.
From \Cref{lem:Pu2Px}, we know that
\begin{equation}
  \sum_{i \in [n]}
  \nrm{u_i^\t - u_i^\tm}_\infty^2
  \leq
  \Umax^2 (n - 1)^2 \sum_{i \in [n]} \nrm{x_i^\t - x_i^\tm}_1^2
  \per
  \n
\end{equation}
Using this we can lower bound the LHS of \cref{eq:S_upper} as 
\begin{align}
  \sum_{t \in \calI} \sqrt{ \gamma (U^\tm)^2 + \mathbb{S}^\t }
  \sum_{i \in [n]} \nrm{x_i^\tp - x_i^\t}_1^2
  &\geq
  \frac{1}{\Umax^2 (n - 1)^2}
  \sum_{t \in \calI} \sqrt{ \gamma (U^\tm)^2 + \mathbb{S}^\t }
  \sum_{i \in [n]} \nrm{u_i^\t - u_i^\tm}_\infty^2
  \nn
  &\geq
  \frac{1}{\Umax^2 (n - 1)^2}
  \sum_{t \in \calI} \sqrt{ \mathbb{S}^\tp }
  \sum_{i \in [n]} \nrm{u_i^\t - u_i^\tm}_\infty^2
  \com
  \label{eq:S_lower}
\end{align}
where in the last inequality we used $\gamma (U^\tp)^2 + \sum_{j \in [n]} P_\infty^\t(\ubar_j) \geq \sum_{j \in [n]} P_\infty^\tp(\ubar_j)$ in \cref{eq:stab_gammaU_lower}.
Combining \cref{eq:S_upper,eq:S_lower}, we obtain
\begin{equation}
  \sum_{t \in \calI} \sqrt{ \mathbb{S}^\tp } \sum_{i \in [n]} \nrm{u_i^\t - u_i^\tm}_\infty^2
  \leq
  2^{11} \Umax^2 n^2 m \alpha
  \prn{\Gamma + \Lambda \sqrt{\mathbb{S}^{T+1}}}
  \eqqcolon
  \Gamma' + \Lambda' \sqrt{\mathbb{S}^{T+1}}
  \com
  \label{eq:S_upper_2}
\end{equation}
where we defined 
\begin{equation}\label{eq:GammaP_LambdaP}
  \Gamma'
  \coloneqq 
  2^{11} \Umax^2 n^2 m \alpha \Gamma
  \com
  \quad
  \Lambda'
  \coloneqq 
  2^{11} \Umax^2 n^2 m \alpha
  \Lambda
\end{equation}
for $\Gamma$ and $\Lambda$ in~\cref{eq:Gamma_Lambda}.

We will further lower bound the LHS of \cref{eq:S_upper_2} as follows:
\begin{align}
  \sum_{t \in \calI} \sqrt{ \mathbb{S}^{\tp} } \sum_{i \in [n]} \nrm{u_i^\t - u_i^\tm}_\infty^2
  &
  =
  \sum_{t \in \calI} \sqrt{ \mathbb{S}^{\tp} } \sum_{i \in [n]} \nrm{\ubar_i^\t - \ubar_i^\tm}_\infty^2
  \nn
  &=
  \sumT \sqrt{ \mathbb{S}^\tp } \sum_{i \in [n]} \nrm{\ubar_i^\t - \ubar_i^\tm}_\infty^2
  -
  \sum_{t \in \calJ} \sqrt{ \mathbb{S}^\tp } \sum_{i \in [n]} \nrm{\ubar_i^\t - \ubar_i^\tm}_\infty^2
  \nn
  &\geq
  \sumT \sqrt{ \mathbb{S}^\tp } \sum_{i \in [n]} \nrm{\ubar_i^\t - \ubar_i^\tm}_\infty^2
  -
  \sqrt{ \mathbb{S}^{T+1} } \cdot 64 \Umax^2 n
  \com
  \label{eq:S_lower_2}
\end{align}
where in the first equality we used the fact that $\ubar_{i}^\t = u_{i}^\t$ and $\ubar_{i}^\tm = u_{i}^\tm$ for $t \in \calI$ from \Cref{lem:doubling_clipping_prop}
and the inequality follows from the second statement of \Cref{lem:ubar_sum_jump_rounds}.
Note here the sum $\sum_{t \in \calI}$ is replaced with $\sumT$ in the last inequality.

In what follows, we will derive an upper bound on $\mathbb{S}^{T+1}$ by lower bounding the RHS of~\cref{eq:S_lower_2}, by a stopping-time argument.

\paragraph{Case 1: when $\max_{t \in \calI} \mathbb{S}^\tp \geq \mathbb{S}^{T+1} / 2$}
Define
\begin{equation}
  \tau = \min\set*{t \in \calI : \mathbb{S}^\tp \geq \mathbb{S}^{T+1} / 2}
  \per
\end{equation}
Since we assumed that $\mathbb{S}^\tp \geq \mathbb{S}^{T+1} / 2$ for some $t \in \calI$, such $\tau \in \calI$ always exists.
Using this $\tau$, we can lower bound the first term in the RHS of~\cref{eq:S_lower_2} as
\begin{align}
  \sumT \sqrt{ \mathbb{S}^{\tp} } \sum_{i \in [n]} \nrm{\ubar_i^\t - \ubar_i^\tm}_\infty^2
  &\geq
  \sum_{t=\tau}^T \sqrt{ \mathbb{S}^\tp } \sum_{i \in [n]} \nrm{\ubar_i^\t - \ubar_i^\tm}_\infty^2
  \nn
  &\geq
  \sqrt{ \mathbb{S}^{\tau+1} } \sum_{t=\tau}^T \sum_{i \in [n]} \nrm{\ubar_i^\t - \ubar_i^\tm}_\infty^2
  \nn
  &=
  \sqrt{ \mathbb{S}^{\tau+1} } \prn*{ \mathbb{S}^{T+1} - \mathbb{S}^\tau }
  \nn
  &\geq
  \sqrt{ \frac{\mathbb{S}^{T+1}}{2} } \prn*{ \mathbb{S}^{T+1} - \frac{\mathbb{S}^{T+1}}{2} }
  \nn
  &\geq
  \frac{\prn*{ \mathbb{S}^{T+1} }^{3/2}}{3} 
  \com
\end{align}
where
the second inequality follows since $\prn{\mathbb{S}^\tp}_{t \in [T]}$ is nondecreasing,
the equality follows from 
$\sum_{t=\tau}^T \sum_{i \in [n]} \nrm{\ubar_i^\t - \ubar_i^\tm}_\infty^2 = \mathbb{S}^{T+1} - \mathbb{S}^\tau$, and
in the third inequality we used $\mathbb{S}^{\tau+1} \geq \mathbb{S}^{T+1} / 2$ and $\mathbb{S}^{\tau} \leq \mathbb{S}^{T+1} / 2$, which follow from the definition of $\tau$.
Therefore, continuing from \cref{eq:S_lower_2}, we obtain
\begin{equation}
  \sumT \sqrt{ \mathbb{S}^\tp } \sum_{i \in [n]} \nrm{u_i^\t - u_i^\tm}_\infty^2
  \geq
  \frac{\prn*{ \mathbb{S}^{T+1} }^{3/2}}{3}
  -
  \sqrt{ \mathbb{S}^{T+1} } \cdot 64 \Umax^2 n
  \per
  \n
\end{equation}
Finally, combining \cref{eq:S_upper_2} with the last inequality, we obtain
\begin{equation}
  \frac{\prn*{ \mathbb{S}^{T+1} }^{3/2}}{3} 
  \leq
  \Gamma' + (\Lambda' +  {64 \Umax^2 n}) \sqrt{\mathbb{S}^{T+1}}
  \per
  \n
\end{equation}
Solving this inequation with respect to $\mathbb{S}^{T+1}$, we obtain
\begin{equation}\label{eq:STp_upper_case1}
  \mathbb{S}^{T+1}
  \lesssim
  \max\set*{(\Gamma')^{2/3}, \Lambda' + \Umax^2 n}
  \per
\end{equation}

\paragraph{Case 2: when $\max_{t \in \calI} \mathbb{S}^\tp < \mathbb{S}^{T+1} / 2$}
Since $\prn{\mathbb{S}^\tp}_{t \in [T]}$ is nondecreasing in $t$, we have
\begin{align}\label{eq:STp_upper_case2}
  \mathbb{S}^{T+1}
  &
  \leq
  \max_{t \in \calI} \mathbb{S}^\tp
  + 
  \sum_{i \in [n]} \sum_{t \in \calJ} \nrm{\ubar_i^\t - \ubar_i^\tm}_\infty^2
  \nn
  &
  \leq
  \frac{\mathbb{S}^{T+1}}{2}
  + 
  \sum_{i \in [n]} \sum_{t \in \calJ} \nrm{\ubar_i^\t - \ubar_i^\tm}_\infty^2
  \leq
  \frac{\mathbb{S}^{T+1}}{2}
  +
  64 \Umax^2 n
  \com
\end{align}
where the second inequality follows from the case assumption that $\max_{t \in \calI} \mathbb{S}^\tp < \mathbb{S}^{T+1} / 2$
and
the last inequality follows from the second statement of \Cref{lem:ubar_sum_jump_rounds}.
Solving this inequation with respect to $\mathbb{S}^{T+1}$, we obtain $\mathbb{S}^{T+1} \leq 128 \Umax^2 n$.

Combining the upper bound in \cref{eq:STp_upper_case1} for Case 1 and \cref{eq:STp_upper_case2} for Case 2, we obtain
\begin{equation}\label{eq:ST_upper_final}
  \mathbb{S}^{T+1}
  \lesssim
  \max\set*{(\Gamma')^{2/3}, \Lambda' + \Umax^2 n}
  \per
\end{equation}

Finally, plugging \cref{eq:ST_upper_final} in the individual swap regret upper bound \cref{eq:swap_upper_indiv_p} in \Cref{lem:swapreg_pre_bound}, we obtain
\begin{align}
  &
  \SwapReg_{x_i}^T
  \nn
  &
  \lesssim
  \frac{U^\T m_i^2 \log T}{\beta_i}
  +
  \frac{m_i^2 \log T}{\alpha_i}
  \sqrt{ \gamma (U^\T)^2 + \sum_{j \in [n]} P_\infty^{T+1}(\ubar_j)}
  +
  \alpha_i \sqrt{P_\infty^{T+1}(\ubar_i)}
  +
  \Umax m_i
  \nn
  &
  \lesssim
  \frac{\Umax m_i^2 \log T}{\beta_i}
  +
  \frac{\Umax m_i^2 \sqrt{\gamma} \log T}{\alpha_i}
  +
  \prn*{
    \frac{m_i^2 \log T}{\alpha_i}
    +
    \alpha_i
  }
  \sqrt{\mathbb{S}^{T+1}}
  +
  \Umax m_i
  \nn
  &
  \lesssim
  \frac{\Umax m_i^2 \log T}{\beta_i}
  +
  \frac{\Umax m_i^2 \sqrt{\gamma} \log T}{\alpha_i}
  \nn
  &\qquad+
  \prn*{
    \frac{m_i^2 \log T}{\alpha_i}
    +
    \alpha_i
  }
  \sqrt{\max\set*{(\Gamma')^{2/3}, \Lambda' + \Umax^2 n}}
  +
  \Umax m_i 
  \com
  \n
\end{align}
where the last inequality follows from \cref{eq:ST_upper_final}.
Plugging $\alpha_i = m_i \sqrt{\log T}$, $\beta_i = 1/(256 \sqrt{m_i})$, and $\gamma = 8 n$ in the last inequality (recall that $\Gamma'$ and $\Lambda'$ are defined in~\cref{eq:GammaP_LambdaP}), we obtain the desired swap regret upper bound in \Cref{thm:indiv_swapreg}.
\end{proof}

\subsection{Scale-free and scale-invariant swap regret minimization without fast convergence}\label{subsec:OptHedge_swap}
Here we show that the convergence rates labeled ``online OLO'' in \Cref{table:regret_generalsum} can indeed be achieved using scale-free online learning.
Recall that the swap regret can be expressed as the sum of the external regrets of $m_i$ external regret minimizers, as discussed in \Cref{subsec:multi_player_preliminaries}.
Hence, it suffices to construct an external regret minimization algorithm for each expert $a$ of each player $i$.

Let $M_i = \max\set{4, \log m_i / 2^{3/2}}$ for each $i \in [n]$.
Then, for each expert $a \in \calA_i$ of player $i \in [n]$, we use the optimistic Hedge algorithm to compute $y_{i,a}^\t \in \Delta_{m_i}$ by
\begin{equation}
  x_i^\t(b) 
  \propto 
  \exp\prn*{ \eta_{i,a}^\t \sum_{s=1}^{t-1} u_{i,a}^\s(b) }
  \com
  \quad
  \eta_{i,a}^\t
  =
  \sqrt{\frac{M_i}{P_\infty^\t(u_{i,a})}}
  \label{eq:OptHedge_swap}
\end{equation}
for $b \in [m_i]$, where we note that $x_{i,a}^1 = \frac{1}{m_i} \ones$ an recall that $P_q^t(z) = \sum_{s=1}^{t-1} \nrm{z^\s - z^\sm}_q^2$.
If the denominator of the learning rate $\eta_{i,a}^\t$ is zero, we set $\eta_{i,a}^\t = \infty$.

\begin{proposition}\label[proposition]{prop:OptHedge_swap}
  The learning dynamics based on \cref{eq:OptHedge_swap} guarantees
  \begin{equation}
    \SwapReg_{x_i}^T = O(\Umax \sqrt{T m_i \log m_i})
    \n
  \end{equation}
  for every player $i \in [n]$.
\end{proposition}

\begin{proof}
From \Cref{lem:adahedge_opt_neg_mainbody} with $h^\t = u_{i,a}^\t$, $m^\t = \zeros$, $\nu^\t = 0$, 
for any $y \in \Delta_{m_i}$ the regret of expert $a \in \calA_i$ for player $i$ is upper bounded as
\begin{equation}
  {\Reg}_{i,a}^T(y)
  =
  \sumT \inpr{y - y_{i,a}^\t, u_{i,a}^\t}
  \leq
  \sqrt{ 32 \sumT \nrm{u_{i,a}^\t}_\infty^2 M_i}
  \leq
  \sqrt{ 32 \sumT x_i^\t(a) \nrm{u_{i}^\t}_\infty^2 M_i}
  \com
  \label{eq:Regia_upper_sqrtT}
\end{equation}
where we used $u_{i,a}^\t = x_i^\t(a) u_i^\t$.
Hence, using the reduction from swap regret minimization to the instances of external regret minimization discussed in \Cref{subsec:swapreg2reg} and \cref{eq:Regia_upper_sqrtT},
we have 
\begin{align}
  \SwapReg_{x_i}^T
  &=
  \sum_{a \in \calA_i} {\Reg}_{i,a}^T
  \leq
  \sum_{a \in \calA_i}
  \sqrt{ 32 \sumT x_i^\t(a) \nrm{u_{i}^\t}_\infty^2 M_i}
  \nn
  &\leq
  \sqrt{ 32 m_i \sumT \sum_{a \in \calA_i} x_i^\t(a) \nrm{u_{i}^\t}_\infty^2 M_i}
  \nn
  &=
  \Umax \sqrt{ 32 m_i M_i T}
  =
  O(\Umax \sqrt{T m_i \log m_i})
  \n
\end{align}
where the second inequality follows from the Cauchy--Schwarz inequality and the last inequality from $\sum_{a \in \calA_i} x_i^\t(a) = 1$ and $\nrm{u_i^\t}_\infty \leq \Umax$.
This completes the proof.
\end{proof}

\end{document}